\documentclass[12pt]{article}
\usepackage[letterpaper,margin=1.25in]{geometry}
\selectfont
\usepackage{amsmath,amssymb,amsthm}
\usepackage{booktabs}
\usepackage{array}
\usepackage{graphicx}
\usepackage[round,authoryear]{natbib}
\usepackage[colorlinks,citecolor=blue,linkcolor=black,urlcolor=blue]{hyperref}
\numberwithin{equation}{section}
\graphicspath{{figures/}}

\theoremstyle{plain}
\newtheorem{theorem}{Theorem}
\newtheorem{proposition}{Proposition}
\newtheorem{lemma}{Lemma}
\newtheorem{corollary}{Corollary}
\theoremstyle{definition}
\newtheorem{definition}{Definition}
\newtheorem{assumption}{Assumption}
\theoremstyle{remark}
\newtheorem{remark}{Remark}

\newcommand{\R}{\mathbb{R}}
\newcommand{\ones}{\mathbf{1}}
\newcommand{\Estar}{E_{\star}}
\newcommand{\Ephi}{E_{\Phi}^{\star}}
\newcommand{\EoneLB}{E_{1}^{\mathrm{LB}}}
\newcommand{\cE}{\mathcal{E}}
\newcommand{\Rmin}{R_{\min}}
\newcommand{\Emin}{E_{\min}}
\newcommand{\dPhi}{d_{\Phi}}
\newcommand{\dPhiomega}{d_{\Phi,\omega}}
\newcommand{\qbar}{\bar q}
\newcommand{\cX}{\mathcal{X}}
\newcommand{\cH}{\mathcal{H}}
\newcommand{\Bboot}{B_{\mathrm{boot}}}
\newcommand{\Iident}{\mathfrak{I}}
\newcommand{\Nextra}{N_{\mathrm{extra}}}
\newcommand{\Ccycle}{C_{M}}
\newcommand{\Deltacycle}{\Delta_{M}}
\newcommand{\Qbase}{Q_{\mathrm{base}}}
\title{What Variation Identifies Payoffs in a Dynamic Game?}
\author{%
Haojie Liu\\
University of California, Riverside\\
\texttt{hliu332@ucr.edu}
\and
Zihan Lin\\
University of California, Riverside\\
\texttt{zlin169@ucr.edu}}
\date{}

\begin{document}
\maketitle

\begin{abstract}
Observed choice in a dynamic game mixes current profit with continuation value. A rival adds a second problem: the same comparison averages over the rival's equilibrium policy. Changing the primitive transition rewrites continuation technology; changing the rival's Markov policy, holding that law fixed, rewrites the mixture over rival-contingent payoffs. The two are not substitutes. For a rival-feature payoff of rank $K$, rank identification up to location requires $\Ephi=\lceil(MK-1)/(M-1)\rceil$ policy environments, and a second kernel when payoffs are saturated. Rank can still be restored by arbitrarily small policy differences. Independent private shocks force mixed rival actions to factor, so a payoff that depends jointly on $d$ rivals is visible only at order $\eta^{d}$ near a common interior baseline. Either rank fails or the smallest identified singular value is at most $\kappa\eta^{\dPhi}$, independently of how many kernels are stacked. Oracle-GLS variance in that direction vanishes only if $n\eta^{2\dPhi}$ diverges. An Anderson--Rubin set that carries first-stage error in the design matrix covers without a vanishing-risk condition. On U.S.\ airline entry, even among rank-identified directions, the most favorable rival-dependent contrast is several times wider than observed behavior.
\end{abstract}

\noindent\textit{Keywords:} dynamic games, identification, Markov-perfect equilibrium, weak instruments, Anderson--Rubin inference.

\noindent\textit{JEL classification:} C57, C73, L13.

\section{Introduction}
\label{sec:intro}

Imagine a firm deciding whether to enter a market. The observed entry probability records two objects at once: current flow profit, and the continuation value created by today's action. A subsidy that raises every action at a state by the same amount, and that is expected to persist in a way the transition law can price, does not change choice differences. Observed behavior therefore does not read off flow payoffs.

A rival creates a second observational problem. The comparison the researcher sees is an average over the rival's equilibrium policy. A payoff that is high when the rival is in and low when the rival is out can look identical to a payoff that never depends on the rival, provided the two objects share the same policy-weighted averages. The researcher who sees only those averages cannot tell them apart.

These are different information problems, and they are not substitutes. Changing the primitive transition $P(x'\mid x,a,b)$ rewrites how current actions move future states: it changes continuation technology. Changing the rival's Markov policy $q(b\mid x)$, holding $P$ fixed, rewrites the mixture over rival-contingent current payoffs. Policy mixing cannot cancel a shift that lives in a shared continuation technology. Extra kernels cannot manufacture a missing contrast across rival actions. The formal model below makes the two operators explicit. The economic distinction does not wait on that notation.

Even when enough distinct environments restore rank, identification need not be informative. Three rival policies such as $0.500$, $0.501$, and $0.502$ can be linearly independent while remaining arbitrarily close. Rank then answers a yes-or-no question: after one location normalization, can the stacked Hotz--Miller map be inverted? The smallest identified singular value answers a different question: how much sample does that inversion require? Identified is not the same as measured.

Two integers organize the rank question. For a finite public state space of cardinality $M\ge 2$, $A\ge 2$ own actions, $B\ge 2$ opponent action profiles, and an unrestricted flow payoff $u(x,a,b)$ at a known discount factor,
\begin{equation}
\label{eq:headline-sat}
\Rmin=2,
\qquad
\Emin=\Estar(M,B):=\left\lceil\frac{MB-1}{M-1}\right\rceil.
\end{equation}
Applied specifications almost never leave rival-action payoffs unrestricted. Additive competition, a count of active rivals, or interactions up to a fixed degree constrain only the opponent-action coordinate. Coefficients $\theta(x,a)$ remain free across states and own actions, so
\[
u(x,a,b)=\phi(b)^\top\theta(x,a),
\]
with $\Phi$ of rank $K$ and $\ones_B\in\operatorname{col}\Phi$. Economic structure replaces raw $B$ by $K$. The sharp policy count becomes
\[
\Ephi=\left\lceil\frac{MK-1}{M-1}\right\rceil.
\]
That number answers how many strategically distinct environments restore rank. When $K<B$, a generic strictly positive kernel already cuts the restricted dynamic-potential intersection down to location, so one transition regime can suffice. Two regimes and $\Ephi$ product policies always suffice. Those identifying policies may be taken arbitrarily close to a common interior baseline.

The main surprise is that this last fact is about rank, not about information. Private shocks that are independent across players force equilibrium opponent mixtures to factor opponent by opponent. A payoff component that depends on one rival's action, as in additive competition, is visible at first order in the policy gap. A component that depends on two rivals jointly is visible only when both marginals move, hence at second order. Three-way dependence is third order. Write $\eta$ for the radius of observed policies around a common interior product baseline, and $\dPhi$ for the highest active interaction degree in $\operatorname{col}\Phi$. For every design built from any number of kernels and any number of policy environments inside that neighborhood, either rank identification fails or
\[
\sigma_{\min,+}\ \le\ \kappa\,\eta^{\dPhi},
\]
with $\kappa$ depending only on $(M,A,\beta,\Phi)$ and the baseline. Extra regimes rewrite continuation technology; they cannot manufacture a missing policy contrast. The exponent is attained, so it is sharp. Feature dimension $K$ says how many environments restore rank. The degree $\dPhi$ says how fast usable information disappears as those environments cluster. Additive rival effects have $\dPhi=1$ for any number of rivals; saturation has $\dPhi=J$. A restriction can be cheap on the first margin and expensive on the second.

Under a bounded measurement-covariance sequence, oracle minimum-distance variance in the weakest identified direction is of order $1/(n\sigma_{\min,+}^2)$. That variance vanishes only if $n\eta^{2\dPhi}\to\infty$ for every rank-identifying design. Ten times less usable strategic variation, at degree one, requires roughly one hundred times as much independent information; at degree two the penalty is fourth-order. The radius $\eta$ can be read off opponent behavior before any payoff is estimated. Anderson--Rubin inversion of the Hotz--Miller moment, carrying first-stage error in the design matrix, covers at the truth without a vanishing-risk condition \citep{andersonrubin1949,stockwright2000,kleibergen2005,andrewsguggenberger2017,andrewsguggenberger2019}.

The airline section is a diagnostic of that prediction, not a structural estimate of Southwest's effect on American's payoffs. Southwest is active in $95.4$ percent of the $2{,}204$ route-quarters on $241$ routes in which both carriers appear. On three of five state grids the cross-stratum dispersion of rival policy does not exceed what sampling noise produces at the observed cell sizes. On the headline grid the saturated specification has rank $42$ of $48$: some rival-dependent directions are classified as identified, and some are not. The most favorably conditioned identified rival-dependent contrast still has a profiled interval of $52.7$ logit units against an observed behavioral range of $8.9$. Under the local $n^{-1/2}$ scaling that width implies an information-equivalent sample multiplier of about $35$ relative to that behavioral range. The multiplier is a design benchmark, not a claim about how many additional routes exist. Rank diagnostics separate identified from unidentified directions; they do not say how little information the identified directions contain.

\citet{hotzmiller1993}, \citet{rust1987}, and \citet{magnacthesmar2002} characterize what conditional choice probabilities reveal about flow utilities at a maintained discount factor. Dynamic-game estimators typically treat the discount as known and impose Markov-perfect play \citep{aguirregabiriamira2007,bajaribenkardlevin2007,pesendorferschmidt2008}. Exclusion restrictions and switching-cost structure identify components of payoffs or beliefs \citep{aguirregabiriamagesan2020,komarova2018}. Multi-environment inverse reinforcement learning and inverse-game theory give rank conditions under which transition variation, discount variation, or rival-strategy variation remove reward ambiguity, including with linear features and in Markov games \citep{ngharadarussell1999,skalse2023,aminsingh2016,cao2021,rolland2022,kleinebuening2024,linadamsbeling2019,futacchetti2021,freihautramponi2025,liao2025}. \citet{schlaginhaufenkamgarpour2024} replace a binary rank condition on transition laws with principal angles for whether a recovered reward transfers under regularized IRL. Their object is the geometry of transition subspaces. The geometry here is product equilibrium policy. Rank results ask when payoff ambiguity disappears. This paper asks which equilibrium-feasible variation removes it, and how quickly usable information disappears when those policies are close. Additional transition variation does not remove the product-policy rate. Once a player's payoff is known up to location, that player's best-response map is identified under a specified counterfactual primitive and a specified opponents' policy; absolute welfare is not \citep{aguirregabiriasuzuki2014,kalouptsidi2017,kalouptsidi2021}. In the weak-instrument literature, weak identification is a possibility to be guarded against. Given $\Phi$ and the observed dispersion of opponent policies, the weakly identified payoff directions here are known before estimation.

The paper proceeds as follows. Section~\ref{sec:model} records the game and a verified two-by-two toy. Section~\ref{sec:variation} states the two obstructions and the structure-and-variation theorem. Section~\ref{sec:strength} gives the attenuation bound. Section~\ref{sec:sampling} turns that bound into a sample-size floor. Section~\ref{sec:inference} records the associated local experiment and the confidence set. Section~\ref{sec:heterogeneity} treats unobserved types. Section~\ref{sec:empirical} is a design diagnostic on U.S.\ airline entry. Unknown patience, transfer implementation, the one-regime count, and lagged-action geometry are in the appendix.

\section{Why observed choices do not reveal flow payoffs}
\label{sec:model}

\subsection{Dynamic game and observables}
\label{sec:observables}

The economic object is an infinite-horizon discounted finite-state stochastic game with public states. A researcher who sees choice probabilities in that game sees a mixture of current profit and the value of future states. Identification statements in this paper concern one player's flow payoffs, taking opponents' Markov behavior as an environment. The remaining players may themselves be strategic; what matters for the identified player is the policy they play and the kernel that maps current actions into tomorrow's public state.

There are $N\ge 2$ players. The public state space $\mathcal X$ is finite, with cardinality $M\ge 2$. Player $i$'s action set has cardinality $A\ge 2$ and a designated reference action $a_0=0$. Write $B$ for the number of opponent action profiles. Flow payoffs of the identified player are $u\in\R^{MAB}$, with coordinates $u(x,a,b)$. The discount factor is $\beta\in(0,1)$. Saturation means that $u(x,a,b)$ is unrestricted across those coordinates. Many empirical specifications instead restrict only rival-action dependence,
\begin{equation}
\label{eq:feature}
u(x,a,b)=\phi(b)^\top\theta(x,a),
\end{equation}
where $\Phi\in\R^{B\times K}$ has rank $K$ and $\ones_B\in\operatorname{col}\Phi$. The coefficients $\theta(x,a)\in\R^K$ remain unrestricted across states and own actions, so the unknown lives in a subspace $\mathcal U_\Phi\subset\R^{MAB}$ of dimension $MAK$. Only one global payoff location is normalized. The paper first records what variation must supply when $\Phi=I_B$, then shows how $\operatorname{col}\Phi$ replaces raw $B$ in the strategic count.

A primitive transition regime is a row-stochastic array $P^r$ of shape $(A,B,M,M)$, indexed by $r=1,\ldots,R$. A strategic policy environment is a strictly interior opponent policy $q^e$, a map $x\mapsto q^e(\cdot\mid x)\in\mathrm{int}\Delta(B)$, indexed by $e=1,\ldots,E$. The identified player's own conditional choice probability in regime $r$ and policy environment $e$ is $p_i^{r,e}$. The two indexes are not interchangeable: $r$ changes the law of motion, $e$ changes the mixture over opponent actions.

\begin{assumption}[Known regular additive random utility]
\label{ass:rum}
Private shocks are additive, independent across players conditional on the public state, i.i.d.\ over time, independent of the controlled public transition, and drawn from a known state-independent distribution for which the surplus $\mathcal S(z)=\mathbb E[\max_{a'}(z_{a'}+\varepsilon_{a'})]$ is $C^2$ and translation-equivariant, and $p=\nabla\mathcal S(z)$ is a $C^1$ diffeomorphism from normalized deterministic value differences onto $\mathrm{int}\Delta(A)$.
\end{assumption}

Write $d(p)$ for the inverse from interior CCPs to normalized own-action value differences, and $g(p)$ for the surplus gap $\mathcal S(z(p))-z_0(p)$. Both are $C^1$ on the interior simplex. Logit is an example. Conditional independence of private shocks implies that, in a regular mixed-strategy Markov-perfect equilibrium, each opponent's mixed action is a function of the public state, so the joint opponent mixture factorizes as a product policy.

\begin{assumption}[Common primitives across measurements]
\label{ass:common}
Across primitive transition regimes $r$ and strategic policy environments $e$, the identified player's structural flow payoff $u$, discount factor $\beta$, and private-shock law are invariant. Only the primitive transition law $P^r$, equilibrium or opponent policies $q^e$, and known experimental transfers $(\tau_{ia}^{r,e})$ may vary across measurements.
\end{assumption}

Cross-environment stability of $(u,\beta)$ and the shock law is the economic restriction that makes pooling informative. Without it, each measurement would be a separate game. Arbitrary cross-market or calendar heterogeneity is not automatically identifying variation.

The experimenter, or the data, may include additive transfers $\tau_{ia}^{r,e}(x)$. Nonreference transfer differences enter action-value differences like a current payoff difference. The reference-action transfer enters continuation values. After Hotz--Miller inversion and subtraction of the observed nonreference transfer difference, the measurement in environment $(r,e)$ is affine in $u$:
\begin{equation}
\label{eq:z}
z^{r,e}
=
C^{r,e}(\beta)u
+
b^{r,e}(\beta).
\end{equation}
The matrix $C^{r,e}(\beta)$ is the payoff operator that maps flow payoffs into own-action value differences under $(P^r,q^e,\beta)$. The intercept $b^{r,e}(\beta)$ loads the surplus gap and the reference transfer through the continuation operator $S^{r,e}(\beta)$. Because $S^{r,e}(\beta)\ones=0$, only the centered component of that continuation offset matters. Appendix~\ref{app:operators} records the stacked construction.

When $\beta$ is known, $b(\beta)$ is known from observed CCPs, the known shock law, and known transfers, and can be subtracted. Identification of $u$ is then a statement about $\mathrm{rank}\,C(\beta)$. When $\beta$ is unknown, $b(\beta)$ is a candidate-dependent affine correction, and observational equivalence is not a pure comparison of payoff column spaces.

\subsection{What one environment identifies}
\label{sec:one-env}

With only one primitive transition environment, a state-dependent transformation can change flow payoffs and continuation values together while leaving observed choice differences unchanged. Give each state an arbitrary potential $h(x)$ and shift current payoff by the current potential minus the discounted expected next-state potential. The Bellman value at $x$ then shifts by $h(x)$, so action-value differences do not move and choice probabilities are unchanged. The associated payoff perturbation is
\begin{equation}
\label{eq:G}
\bigl[G_{P,\beta}h\bigr](x,a,b)
=
h(x)-\beta\sum_{x'}P(x'\mid x,a,b)h(x').
\end{equation}
We call $\mathrm{Im}\,G_{P,\beta}$ the dynamic-potential gauge generated by $P$. One of its $M$ directions is a global location shift. For $M\ge 2$ at least $M-1$ non-location directions remain. Varying rival strategies inside the same $P$ does not remove the gauge, because every such observation inherits the same continuation technology.

\subsection{A small binary example}
\label{sec:example}

Take two states, two own actions, and two opponent actions, so the flow payoff is an array of length $8$. Then $\Ephi=\Estar(2,2)=3$ and the location-normalized target rank is $7$. A single mixed opponent policy produces two Hotz--Miller difference rows, and those rows never see $G_{P,\beta}h$. Any collection of policies inside one $P$ therefore leaves at least a two-dimensional class, one direction of which is the global constant.

Too few opponent policy mixtures expose only weighted averages of opponent-contingent payoffs. With two opponent policies the remaining averaging class is two-dimensional, distinct from the potential gauge. For binary policies that differ in every state it has the closed form
\begin{equation}
\label{eq:rho-example}
\rho_1=t\bigl(D(q^1)-D(q^2)\bigr)^{-1}\ones,
\qquad
\rho_0=c_1\ones-D(q^1)\rho_1,
\end{equation}
with $D(q)=\mathrm{diag}(q_x)$ and $q_x=\Pr(b=1\mid x)$. The parameter $c_1$ is location; $t$ is the extra hidden direction. A third policy profile removes this averaging direction. A second strictly positive kernel then removes the remaining gauge direction. Two kernels, three policies, and one leftover constant are the binary case of Theorem~\ref{thm:structure}.

Figure~\ref{fig:toy} reports that architecture on the paper's operator, with $\beta=0.9$ and three interior product policies. One kernel yields a $6\times 8$ matrix of rank $6$ and nullity $2$. Two kernels yield a $12\times 8$ matrix of rank $7$ and nullity $1$: after location normalization the map is invertible. Holding those two kernels and the three-policy design fixed, and shrinking the policy radius, leaves rank equal to $7$ at every $\eta$ in the grid while $\sigma_{\min,+}$ falls in proportion to $\eta$. The numerical log-log slope is $1.000$, the degree-one rate of Theorem~\ref{thm:attenuation} for one binary rival. Rank stays solved; measurement does not.

\begin{figure}[t]
\centering
\includegraphics[width=\textwidth]{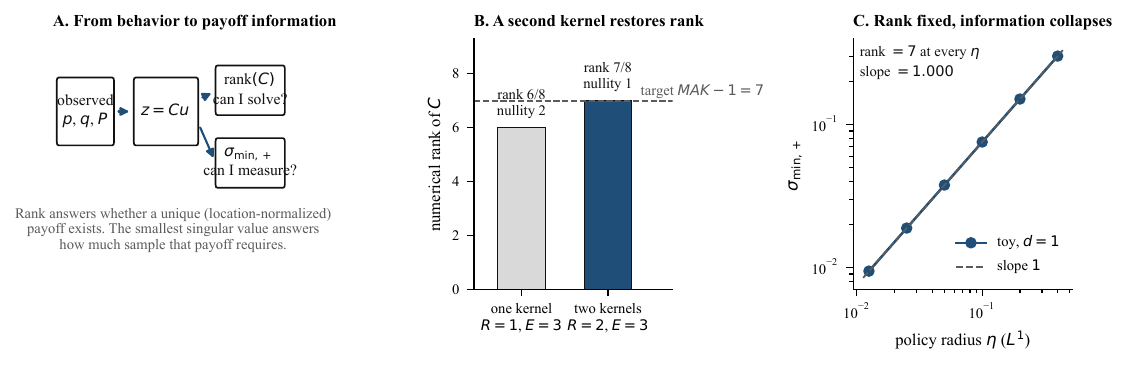}
\caption{A verified $M=A=B=2$ toy on the paper's Hotz--Miller operator. Panel A separates the rank question (can the location-normalized payoff be recovered) from the strength question (how small is the weakest identified singular value). Panel B: one primitive kernel leaves a two-dimensional kernel; a second kernel restores rank $7$. Panel C: the two-kernel, three-policy design remains rank $7$ as the policy radius shrinks, while $\sigma_{\min,+}$ tracks $\eta$.}
\label{fig:toy}
\end{figure}

The closed form \eqref{eq:rho-example} is only intuition for averaging. The general argument does not pass through it.

\section{How many environments restore rank?}
\label{sec:variation}

\subsection{Environmental obstruction}
\label{sec:env-obst}

The identification problem in one environment is that flow-payoff changes can be offset by state-dependent continuation-value changes. A subsidy that makes every action at a state more attractive by the same amount, and that is expected to persist in a way that the transition kernel can price, does not change observed choice differences. The analyst who sees only those differences cannot tell the subsidy apart from a change in continuation values. The previous subsection already gives this economic reason. The formal statement is that the dynamic-potential gauge has dimension $M$ and lies in the observational kernel of every mixed policy observed under that kernel.

\begin{proposition}[Environmental obstruction]
\label{prop:env}
Assume Assumptions~\ref{ass:rum} and~\ref{ass:common}, $\beta\in(0,1)$, and a saturated payoff $u\in\R^{MAB}$, with $M\ge 2$, $A\ge 2$, $B\ge 2$. Fix a row-stochastic kernel $P$. Then $\mathrm{Im}\,G_{P,\beta}$ lies in the observational kernel of every mixed opponent policy observed under $P$, $G_{P,\beta}$ is injective, and $\dim\mathrm{Im}\,G_{P,\beta}=M$. Consequently no collection of strategic policy environments inside a single primitive regime can reduce residual payoff ambiguity below dimension $M$ when the payoff is saturated. Identification of a saturated payoff up to only a global additive constant therefore requires $R\ge 2$.
\end{proposition}

Transfers, subsidies, or one-sided experiments that move $q$ while holding $P$ fixed are not a second primitive environment. Two kernels always share at least the constant-flow line; generic pairs share nothing more (Lemma~\ref{lem:E1}).

\subsection{Strategic obstruction}
\label{sec:strat-obst}

Even if continuation technology varies, observing too few rival-policy profiles only reveals averages of payoffs across rival actions. A payoff that is high when the rival enters and low when the rival stays out can look identical to a state-dependent payoff that does not depend on the rival at all, provided the two objects share the same policy-weighted averages. Some opponent-contingent payoff components therefore remain hidden. The formal object is an own-action-invariant payoff whose policy-weighted averages are state-constant. Current own-action differences vanish by construction. Under each observed policy the reference expected payoff is a state-constant, so every stochastic kernel annihilates the continuation term as well. Extra transitions cannot see that payoff, so the policy count is a separate requirement from a second kernel.

\begin{proposition}[Strategic-policy obstruction]
\label{prop:strat}
Let $C_E^\Phi(\beta)$ stack the feature-coordinate payoff operators over any finite collection of primitive regimes and $E$ interior policies. Then
\begin{equation}
\label{eq:avg-bound}
\dim\ker C_E^\Phi(\beta)
\ge
\max\bigl\{1,\,MK-(M-1)E\bigr\}.
\end{equation}
The bound does not depend on $R$, on heterogeneity of the kernels, on $A$, or on genericity of $P^r$. Identification up to a single global location therefore requires
\begin{equation}
\label{eq:Estar}
E\ge \Ephi=\left\lceil\frac{MK-1}{M-1}\right\rceil.
\end{equation}
When $K=B$, this is $\Estar(M,B)$. No amount of transition variation can compensate for observing too few distinct opponent policy profiles.
\end{proposition}

Propositions~\ref{prop:env} and~\ref{prop:strat} identify two universal lower-bound obstructions. They need not exhaust the observational kernel. Particular transition-policy geometries can create additional null directions. Appendix~\ref{app:lag} records one such family: lagged-action kernels leave a residual gauge of dimension at least $|\mathcal Z|$ no matter how many policies are stacked.

The two obstructions remain distinct. Fewer than $\Ephi$ policies leave a policy-averaging kernel no matter how many kernels are stacked. One transition regime leaves the full dynamic-potential gauge when the payoff is saturated. When $K<B$, that environmental obstruction is no longer universal. Table~\ref{tab:grid} records the saturated cells; Theorem~\ref{thm:structure} separates the restricted case.

\begin{table}[t]
\centering
\caption{Informational bottlenecks after rival-payoff restrictions. Averaging is indexed by $K$, not $B$. The cell $R=1$, $E\ge\Ephi$ is a saturated impossibility; when $K<B$ one-regime identification is possible. The cell $R\ge 2$, $E\ge\Ephi$ is attained at $R=2$.}
\label{tab:grid}
\small
\begin{tabular}{lcc}
\toprule
& $E<\Ephi$ & $E\ge\Ephi$ \\
\midrule
$R=1$
& averaging remains
& saturated gauge remains; not a universal block if $K<B$ \\
$R\ge 2$
& averaging remains
& known-$\beta$ identification attainable \\
\bottomrule
\end{tabular}
\end{table}

\subsection{Economic structure and variation}
\label{sec:sharp}

Definition~\ref{def:min} is understood in feature coordinates: $\Iident(R,E)=1$ when some design attains $\mathrm{rank}\,C^\Phi=MAK-1$. The minima $\Rmin$ and $\Emin$ remain useful marginal summaries. The object that records the tradeoff is the conditional frontier $\cE_\Phi(R)$ in Section~\ref{sec:frontier}.

\begin{definition}[Minimal cardinalities of variation]
\label{def:min}
Let $\Iident(R,E)=1$ when there exists a finite measurement design using at most $R$ distinct primitive transition regimes and at most $E$ distinct opponent-policy profiles such that, after payoff-location normalization, $\mathrm{rank}\,C^\Phi(\beta)=MAK-1$. Then
\[
\Rmin=\min\{R:\exists E,\ \Iident(R,E)=1\},
\qquad
\Emin=\min\{E:\exists R,\ \Iident(R,E)=1\}.
\]
\end{definition}

\begin{theorem}[Economic structure and variation]
\label{thm:structure}
Assume known $\beta\in(0,1)$, Assumptions~\ref{ass:rum} and~\ref{ass:common}, $M\ge 2$, $A\ge 2$, and a rival-feature payoff \eqref{eq:feature} with $\operatorname{rank}\Phi=K$ and $\ones_B\in\operatorname{col}\Phi$.
\begin{enumerate}
\item[(A)] For any number of primitive regimes and any kernels,
\[
\dim\ker C_E^\Phi\ge\max\{1,MK-E(M-1)\}.
\]
Hence $\Emin=\Ephi$.
\item[(B)] There exist $R=2$ strictly positive kernels and $E=\Ephi$ strictly interior conditionally independent product-policy profiles, arbitrarily close to a common interior product baseline, such that $\operatorname{rank} C^\Phi=MAK-1$.
\item[(C)] For every fixed primitive kernel $P$,
\[
\dim\bigl(\mathcal U_\Phi\cap\operatorname{Im} G_{P,\beta}\bigr)
=
M-\operatorname{rank}\mathcal R_\Phi(P).
\]
If $K=B$, every kernel leaves the $M$-dimensional saturated gauge, so $\Rmin=2$. If $K<B$, a generic strictly positive $P$ in the interior stochastic-kernel parameter space has $\operatorname{rank}\mathcal R_\Phi(P)=M-1$, and its restricted dynamic-potential intersection contains only location. The explicit one-regime construction then yields $\Rmin=1$ when $K<B$. Thus
\[
\Emin=\Ephi,
\qquad
\Rmin=
\begin{cases}
2,&K=B,\\
1,&K<B.
\end{cases}
\]
The two-regime construction attains $(R,E)=(2,\Ephi)$. An explicit one-regime construction attains $(R,E)=(1,K+\Ephi)$.
\end{enumerate}
\end{theorem}

Ranks and cardinalities depend on $\operatorname{col}\Phi$ and on $K$, not on a particular basis. For $T\in\mathrm{GL}(K)$, $\Phi\mapsto\Phi T$ merely reparameterizes $\theta$.

Take $J$ binary opponents. Table~\ref{tab:menu} translates common rival-payoff restrictions into $K$ and $\Ephi$. Additive rival structure converts an exponential strategic-environment requirement into a linear one.

\begin{table}[t]
\centering
\caption{Rival-payoff restrictions and the sharp policy count $\Ephi$ for $J$ binary opponents ($B=2^J$). The last column is $\Ephi$ when $M\ge K$.}
\label{tab:menu}
\small
\begin{tabular}{lccc}
\toprule
restriction & $K$ & $\Ephi$ & $M\ge K$ \\
\midrule
no rival-payoff dependence & $1$ & $1$ & $1$ \\
linear active-rival count & $2$ & $3$ & $3$ \\
heterogeneous additive rivals & $J+1$ & $\lceil(M(J+1)-1)/(M-1)\rceil$ & $J+2$ \\
active-count indicators & $J+1$ & same & $J+2$ \\
interactions through degree $d$ & $\sum_{k=0}^{d}\binom{J}{k}$ & $\lceil(MK-1)/(M-1)\rceil$ & $K+1$ \\
saturated & $2^J$ & $\lceil(M2^J-1)/(M-1)\rceil$ & $2^J+1$ \\
\bottomrule
\end{tabular}
\end{table}

\begin{corollary}[Saturated benchmark]
\label{thm:sharp}
\label{cor:saturated}
If $K=B$, then $\Ephi=\Estar(M,B)$, the restricted gauge is the $M$-dimensional dynamic-potential class, and $\Rmin=2$. Identification of the saturated payoff up to location is attained by $R=2$ strictly positive kernels and $E=\Estar$ interior product policies, which may be taken arbitrarily close to uniform.
\end{corollary}

The first regime can be thought of as a measurement in which continuation differences are shut down. The second reopens continuation and uses a cycle increment to recover the reference-payoff block. The $\Ephi$ policies invert the feature-moment averaging map. Appendix~\ref{app:restricted} records the constructions.

\begin{remark}[Open identifying neighborhood]
\label{rem:open}
At the two-regime construction, after one payoff-location normalization, $\mathrm{rank}\,C^\Phi(\beta)=MAK-1$. For every fixed known $\beta\in(0,1)$, an open set of strictly positive transition regimes and strictly interior product-policy profiles around that construction continues to identify the restricted payoff up to location.
\end{remark}

\label{sec:frontier}
Conditional on $R$ regimes, write $\cE_\Phi(R)$ for the smallest $E$ such that some identifying design exists. Theorem~\ref{thm:structure} already gives $\cE_\Phi(R)=\Ephi$ for every $R\ge 2$. One regime has only $EM(A-1)$ choice-difference rows, so identification up to location requires
\begin{equation}
\label{eq:E1LB}
\cE_\Phi(1)\ge\EoneLB
:=
\max
\left\{
\left\lceil\frac{MAK-1}{M(A-1)}\right\rceil,\,
\left\lceil\frac{MK-1}{M-1}\right\rceil
\right\}.
\end{equation}
An explicit construction attains $K+\Ephi$. On a $47$-cell grid of model sizes and rival-feature families the two bounds meet, so those extra $K$ policies are an artifact of the construction (Proposition~\ref{prop:e1exact}). Whether $\cE_\Phi(1)=\EoneLB$ for every $(M,A,\Phi)$ outside that grid remains open. The two-regime count $\cE_\Phi(R)=\Ephi$ for $R\ge 2$ holds for every such triple. Appendix~\ref{app:frontier} records the statements and the exact-arithmetic certificate.

\section{Why rank is not enough}
\label{sec:strength}

Theorem~\ref{thm:structure}(B) attains the sharp counts with policies that may be arbitrarily close to a common interior product baseline. Rank is restored at every positive spread. The smallest singular value is not. Private shocks that are independent across players force opponent mixtures to factor, and a $d$-way rival interaction is then exposed only when $d$ marginals move together. Extra transition regimes rewrite continuation technology. They cannot manufacture the missing policy contrast.

\subsection{Interaction filtration at an interior baseline}
\label{sec:filtration}

Write the opponent joint-action space as a product $B=\prod_{j=1}^{J}B_j$. The integer $J$ is the number of independently mixing opponent components. A conditionally independent product policy is the object that private-shock Markov-perfect play produces.

Fix an interior product baseline $\qbar=\bigotimes_{j=1}^{J}\qbar_j$ with $\qbar_j\in\operatorname{int}\Delta(B_j)$. Decompose each factor into its constant and $\qbar_j$-centered parts,
\[
\R^{B_j}=\operatorname{span}\{\ones_{B_j}\}\oplus V_j,
\qquad
V_j=\Bigl\{v\in\R^{B_j}:\textstyle\sum_b \qbar_j(b)v(b)=0\Bigr\}.
\]
For $S\subseteq\{1,\ldots,J\}$ let $\cH_S(\qbar)$ apply $V_j$ on $j\in S$ and the constant on $j\notin S$, and set
\begin{equation}
\label{eq:Hd}
\cH_d(\qbar)=\bigoplus_{|S|=d}\cH_S(\qbar),
\qquad
\R^{B}=\bigoplus_{d=0}^{J}\cH_d(\qbar).
\end{equation}
The summands are orthogonal in the $\qbar$-weighted inner product. Let $F=\operatorname{col}\Phi$ and define the tail filtration
\[
F_{\ge d}=F\cap\Bigl(\bigoplus_{r=d}^{J}\cH_r(\qbar)\Bigr),
\qquad
g_d=\dim F_{\ge d},\qquad g_{J+1}=0.
\]
The successive dimensions
\begin{equation}
\label{eq:md}
m_d^\Phi=g_d-g_{d+1}
\end{equation}
sum to $K$. Because $\ones_B\in F$, one has $m_0^\Phi=1$. The highest active degree is
\begin{equation}
\label{eq:dPhi}
\dPhi=\max\{d:g_d>0\}.
\end{equation}
Thus $K$ counts how much strategic variation is required, while $\dPhi$ records how rapidly the weakest direction in $F$ deteriorates under local product mixing. Both integers, and the exponent $\dPhi$, are properties of $\operatorname{col}\Phi$. The raw number $\sigma_{\min,+}(C_\Phi)$ is not: a nonorthogonal rescaling $\Phi\mapsto\Phi T$ reparameterizes $\theta$ and can change singular values without changing the identified payoff. At the uniform baseline $\qbar_j=B_j^{-1}\ones$, \eqref{eq:Hd} is the usual orthogonal interaction decomposition and the mixing operator is
\[
R(\eta)=\bigotimes_{j=1}^{J}(\Pi_j+\eta H_j),
\qquad
\Pi_j=B_j^{-1}\ones\ones',\quad H_j=I_{B_j}-\Pi_j,
\]
which acts on $\cH_d$ as $\eta^{d}$. Observed opponent policies cluster somewhere other than uniform, so the baseline is kept free throughout.

\begin{lemma}[Restricted mixing]
\label{lem:mix}
On $F$, the singular values of $R(\eta)|_F$ have orders $\eta^d$ with multiplicity $m_d^\Phi$ for $d=0,\ldots,J$.
\end{lemma}

The argument uses a basis $\{v_{d,\ell}\}$ adapted to $F_{\ge d}\supset F_{\ge d+1}$. Each complement vector $v_{d,\ell}\in F_{\ge d}\setminus F_{\ge d+1}$ has a nonzero degree-$d$ projection, and those projections are linearly independent, so
\[
R(\eta)v_{d,\ell}=\eta^d\bigl(P_d v_{d,\ell}+O(\eta)\bigr).
\]
Assembling columns yields $R(\eta)V=W(\eta)D_\eta$ with $W(\eta)\to W_0$ of full column rank and
\[
D_\eta=\operatorname{diag}(\eta^0 I_{m_0^\Phi},\eta^1 I_{m_1^\Phi},\ldots,\eta^J I_{m_J^\Phi}).
\]
Appendix~\ref{app:strength} records the details.

\subsection{A bound that no design can beat}
\label{sec:attenuation}

Call a measurement design \emph{$\eta$-clustered at $\qbar$} if every opponent-policy environment is a strictly interior conditionally independent product policy whose marginals satisfy
\begin{equation}
\label{eq:cluster}
\max_{e\le E}\ \max_{j\le J}\ \max_{x\in\cX}\ \bigl\|q^{e}_j(\cdot\mid x)-\qbar_j\bigr\|_1\le\eta .
\end{equation}
We measure policy dispersion in $L^1$ distance. This equals twice the conventional total-variation distance and avoids an irrelevant factor of two in the rate statements. No restriction is placed on the number of environments, on the number of primitive transition regimes, or on the kernels themselves. Because stacking more measurement blocks mechanically inflates every singular value of the raw stack, $C^{\Phi}$ is normalized per environment, $C^{\Phi}=(RE)^{-1/2}\bigl[C^{r,e}\bigr]_{r,e}$. This is also the normalization under which a fixed total sample is split across environments, so that the factor cancels in the oracle variance of Section~\ref{sec:sampling}.

\begin{theorem}[Design-free strategic attenuation]
\label{thm:attenuation}
Assume Assumptions~\ref{ass:rum} and \ref{ass:common}, a known $\beta\in(0,1)$, and a rival-feature payoff \eqref{eq:feature} with $\operatorname{rank}\Phi=K$ and $\ones_B\in\operatorname{col}\Phi$. Fix an interior product baseline $\qbar$ and let $\dPhi\ge 1$ be given by \eqref{eq:dPhi}. There is a finite constant
\[
\kappa=\kappa(M,A,\beta,\Phi,\qbar)
\]
such that every $\eta$-clustered design at $\qbar$, using any number $R$ of strictly positive primitive transition regimes and any number $E$ of policy environments, satisfies the following dichotomy: either the restricted payoff is not identified modulo the global payoff location, or
\[
\sigma_{\min,+}\bigl(C^{\Phi}(\beta)\bigr)\ \le\ \kappa\,\eta^{\dPhi},
\]
where $\sigma_{\min,+}$ is the smallest singular value after the global payoff location is quotiented out. Thus rank failure is already a stronger failure, while every rank-identifying design is subject to the same $\eta^{\dPhi}$ strength bound. The constant does not depend on $R$, on $E$, on the transition kernels, or on the placement of the policies inside the neighborhood \eqref{eq:cluster}. One admissible choice is
\begin{equation}
\label{eq:kappa}
\kappa=\frac{2\beta}{1-\beta}\cdot\sqrt{M(A-1)}\cdot\frac{c_\Phi(\qbar)}{\delta_\Phi},
\end{equation}
where $c_\Phi(\qbar)$ is the attenuation constant of the top filtration layer and $\delta_\Phi$ is the distance from the normalized test direction to the payoff-location line, both defined in Appendix~\ref{app:attenuation}.
\end{theorem}

The proof isolates why transition variation cannot help. Choose a payoff direction $\psi$ that is invariant across the identified player's own actions and whose feature image $\Phi\psi$ lies in the top layer $F_{\ge\dPhi}$. Three facts then combine. First, own-action invariance makes every current-payoff difference vanish identically, so $\psi$ can reach the data only through continuation values. Second, the continuation term loads on $q^{e}(\cdot\mid x)'\Phi\psi$, and $\Phi\psi$ is $\qbar$-centered on at least $\dPhi$ factors, so the product expansion $\bigotimes_j(\qbar_j+\eta h^{e}_j)$ annihilates every term with fewer than $\dPhi$ moving marginals and leaves $O(\eta^{\dPhi})$. Third, the continuation operator $\beta(P_a-P_0)(I-\beta P_0)^{-1}$ has norm at most $2\beta/(1-\beta)$ for every row-stochastic kernel. Hence $\|C^{\Phi}\psi\|\le\kappa\eta^{\dPhi}\|\psi\|$ uniformly. If the design fails to identify the restricted payoff modulo location, that is already a stronger failure. If it identifies, $\sigma_{\min,+}$ inherits the bound. Appendix~\ref{app:attenuation} gives the details.

In the identifying constructions of Theorem~\ref{thm:structure}, rank is restored at every positive spread, but the smallest nonzero singular value collapses at the rate governed by the interaction filtration. Designs that retain additional null directions are weaker still. The cardinalities are therefore necessary for identification and still leave precision governed by $\eta^{\dPhi}$; no redesign of the transition environment closes that gap among identifying designs. Adding regimes, adding policies, or separating the kernels sharply all leave $\eta^{\dPhi}$ untouched.

The bound is uniform over transition geometries: it concerns the strategic margin alone, and the environmental spread $\alpha$ in Section~\ref{sec:strength-thm} is a property of the canonical construction, not of the problem. It is also uniform in $E$, so collecting more policy environments inside a fixed neighborhood cannot help. Finally it depends on the payoff restriction only through $\dPhi$, not through $K$. The number of environments needed and the precision obtainable are governed by different features of $\operatorname{col}\Phi$.

\begin{corollary}[Attenuation is attained]
\label{cor:attain}
For the canonical construction of Appendix~\ref{app:strength} with environmental spread bounded away from zero, $R=2$, and $E=\Ephi$, $\sigma_{\min,+}\asymp\eta^{\dPhi}$. The exponent in Theorem~\ref{thm:attenuation} is therefore sharp.
\end{corollary}

\subsection{Canonical restricted spectrum}
\label{sec:strength-thm}

The canonical two-regime construction refines Theorem~\ref{thm:attenuation} by resolving the whole spectrum, at the cost of specializing the transition geometry. Let $\eta$ be the strategic policy spread and $\alpha$ the environmental transition spread. After the global payoff constant is quotiented out, current-payoff differences contribute $M(A-1)m_d^\Phi$ directions of order $\eta^d$. The reference payoff contributes $M-1$ directions of order $\alpha$ and, for $d\ge 1$, $M m_d^\Phi$ directions of order $\alpha\eta^d$.

\begin{theorem}[Canonical restricted spectrum]
\label{thm:strength}
Under the filtration-preserving canonical two-regime, $\Ephi$-policy construction adapted to $F=\operatorname{col}\Phi$ and specified in Appendix~\ref{app:strength}, let the strategic policy spread be $\eta$ and the environmental transition spread be $\alpha$. After quotienting the global payoff constant,
\[
\sigma_{\min,+}\asymp\alpha\eta^{\dPhi}.
\]
On the isotropic path $\alpha=\eta=\varepsilon$,
\[
\sigma_{\min,+}\asymp\varepsilon^{\dPhi+1}.
\]
Current-payoff-difference blocks contribute $M(A-1)m_d^\Phi$ directions of order $\eta^d$. The reference-payoff block contributes $M-1$ directions of order $\alpha$ and $M m_d^\Phi$ directions of order $\alpha\eta^d$ for $d\ge 1$.
\end{theorem}

The stacked two-regime operator is not block diagonal. A bounded row operation with condition number independent of $(\alpha,\eta)$ subtracts the first regime from the second and permits the current-difference and reference blocks to be read separately; Appendix~\ref{app:strength-R} records that step. The extra factor $\alpha$ relative to Theorem~\ref{thm:attenuation} is the price of near-uniform kernels: a well-separated transition pair recovers the design-free rate $\eta^{\dPhi}$ but not more.

Additive heterogeneous binary rivals have $F=\cH_0\oplus F_1$, hence $\dPhi=1$ and $\sigma_{\min,+}\asymp\varepsilon^2$ for any $J$. Interactions through degree $d$ have $\dPhi=d$ and rate $\varepsilon^{d+1}$. Saturation has $\dPhi=J$ and recovers $\varepsilon^{J+1}$. Table~\ref{tab:strength-menu} collects these cases.

\begin{table}[t]
\centering
\caption{Examples for $J$ binary opponent components. The integer $K$ governs the sharp policy count $\Ephi$; the highest active degree $\dPhi$ governs local strength. The design-free column is the exponent of Theorem~\ref{thm:attenuation}, valid for every rank-identifying transition geometry; the canonical column adds the environmental spread of Theorem~\ref{thm:strength} on the isotropic path $\alpha=\eta=\varepsilon$.}
\label{tab:strength-menu}
\small
\begin{tabular}{lcccc}
\toprule
restriction & $K$ & $\dPhi$ & design-free & canonical \\
\midrule
no rival-payoff dependence & $1$ & $0$ & $1$ & $\varepsilon$ \\
heterogeneous additive rivals & $J+1$ & $1$ & $\varepsilon$ & $\varepsilon^{2}$ \\
interactions through degree $d$ & $\sum_{k=0}^{d}\binom{J}{k}$ & $d$ & $\varepsilon^{d}$ & $\varepsilon^{d+1}$ \\
saturated & $2^{J}$ & $J$ & $\varepsilon^{J}$ & $\varepsilon^{J+1}$ \\
\bottomrule
\end{tabular}
\end{table}

\begin{corollary}[Saturated recovery]
\label{cor:spectrum}
If $K=B$, then $m_d^\Phi=\dim\cH_d=c_d$ as in \eqref{eq:cd} below, $\dPhi=J$, and Theorem~\ref{thm:strength} recovers $\sigma_{\min,+}\asymp\alpha\eta^{J}$.
\end{corollary}

The combinatorial counts used for saturation are
\begin{equation}
\label{eq:cd}
c_d
=
[z^d]
\prod_{j=1}^{J}
\bigl(1+(B_j-1)z\bigr).
\end{equation}
Figure~\ref{fig:strength} plots the product-versus-joint comparison for the saturated construction. Product slopes track $\varepsilon^{J+1}$ and the joint comparison tracks $\varepsilon^{2}$ inside this geometry; Theorem~\ref{thm:attenuation} is the statement that applies to every geometry among rank-identifying designs.

\begin{figure}[t]
\centering
\includegraphics[width=0.82\textwidth]{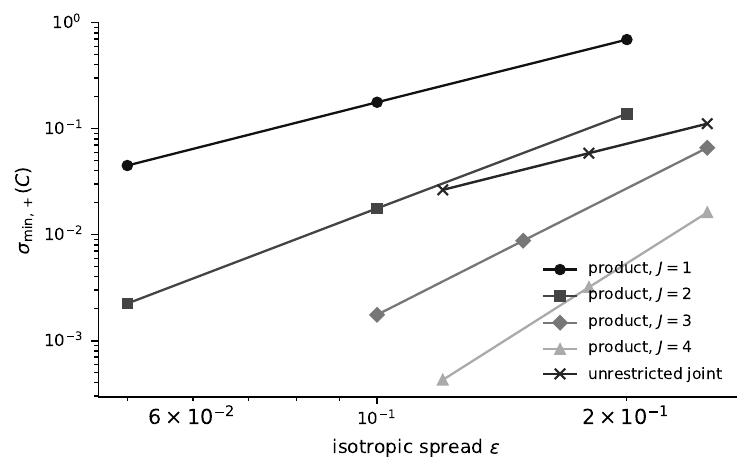}
\caption{Why product equilibrium policy is costly. Local singular-value scaling in the canonical sharp product-policy construction and an unrestricted joint-policy comparison, under isotropic environmental and strategic spread $\varepsilon$. Product series correspond to the saturated case of Corollary~\ref{cor:spectrum} with $J=1,2,3,4$ independently mixing opponent components. The remaining series is the unrestricted joint-policy comparison for the same environmental perturbation. Product slopes track $\varepsilon^{J+1}$ and the joint comparison tracks $\varepsilon^{2}$ inside this geometry. Unrestricted joint policies are not feasible under Assumption~\ref{ass:rum}; they isolate the product restriction.}
\label{fig:strength}
\end{figure}

\begin{remark}[Joint-policy comparison]
\label{rem:joint}
\label{cor:joint}
Under Assumption~\ref{ass:rum}, equilibrium opponent policies are conditionally independent product policies. Unrestricted correlated joint-policy perturbations are not feasible equilibrium-policy designs under the maintained model. Their role is to isolate the effect of relaxing the product restriction. In the canonical sharp geometry, product feasibility generates the interaction hierarchy. Relaxing it removes that hierarchy, although other transition geometries can also improve local strength while retaining product policies.
\end{remark}

\section{How much data the weakest direction needs}
\label{sec:sampling}

Work in the location-normalized identified payoff coordinates, and write $C_\varepsilon$ for the filtration-preserving canonical restricted design along a path of spreads. With fixed-length panels, $n$ denotes the number of independent market or route clusters. Suppose
\[
\widehat z_{n,\varepsilon}
=
C_\varepsilon\theta_0
+
n^{-1/2}\xi_{n,\varepsilon},
\]
and define $\Omega_{n,\varepsilon}=\operatorname{Var}(\xi_{n,\varepsilon})$. Assume uniformly along the considered sequence
\[
0<cI\preceq\Omega_{n,\varepsilon}\preceq CI.
\]
Oracle GLS uses the true covariance,
\[
\widehat\theta
=
\arg\min_\theta
(\widehat z_{n,\varepsilon}-C_\varepsilon\theta)'\Omega_{n,\varepsilon}^{-1}(\widehat z_{n,\varepsilon}-C_\varepsilon\theta).
\]
The estimator is linear, so
\begin{equation}
\label{eq:var}
\operatorname{Var}(\widehat\theta)
=
\frac1n
\bigl(C_\varepsilon'\Omega_{n,\varepsilon}^{-1}C_\varepsilon\bigr)^{-1}.
\end{equation}
The identity does not use a normality assumption. Uniform eigenvalue bounds yield
\begin{equation}
\label{eq:var-smin}
\lambda_{\max}\operatorname{Var}(\widehat\theta)
\asymp
\frac{1}{n\sigma_{\min,+}(C_\varepsilon)^2}.
\end{equation}
Theorem~\ref{thm:strength} then gives
\[
\lambda_{\max}\operatorname{Var}(\widehat\theta)
\asymp
\frac{1}{n\alpha^2\eta^{2\dPhi}}.
\]
On an isotropic triangular sequence $\alpha=\eta=\varepsilon_n\to 0$,
\begin{equation}
\label{eq:var-eps}
\lambda_{\max}\operatorname{Var}(\widehat\theta)
\asymp
\frac{1}{n\varepsilon_n^{2(\dPhi+1)}}.
\end{equation}
Along that sequence, $n\varepsilon_n^{2(\dPhi+1)}\to\infty$ is necessary and sufficient for the largest oracle variance eigenvalue to vanish. Holding precision fixed along the design-free bound of Theorem~\ref{thm:attenuation} gives the iso-information relation $n'\eta'^{2\dPhi}=n\eta^{2\dPhi}$, or $n'/n=(\eta/\eta')^{2\dPhi}$. This is a local scaling benchmark, not a finite-sample extrapolation of Anderson--Rubin widths.

\begin{proposition}[Oracle sampling-noise amplification]
\label{prop:noise}
Under the measurement-noise sequence above, \eqref{eq:var}--\eqref{eq:var-eps} hold for the filtration-preserving canonical restricted design of Theorem~\ref{thm:strength}.
\end{proposition}

The variance identity is algebraic. The Gaussian experiment below is what a triangular-array CLT for $\widehat z_{n,\varepsilon}$ delivers, and it supplies a minimax lower bound of the same order.

\begin{proposition}[Gaussian two-point lower bound]
\label{prop:twopoint}
In the location-normalized coordinates of this section, consider the Gaussian experiment
\[
\widehat z_n
=
C\theta
+
n^{-1/2}\xi,
\qquad
\xi\sim N(0,\Omega),
\]
with $cI\preceq\Omega\preceq CI$ and $\sigma_{\min,+}(C)>0$. There exist a pair $\{\theta_0,\theta_1\}$ with $\|\theta_1-\theta_0\|\asymp(n\sigma_{\min,+}^2)^{-1/2}$ and a constant $c_*>0$ depending only on $c$ such that
\[
\inf_{\widehat\theta}
\sup_{\theta\in\{\theta_0,\theta_1\}}
E_\theta\|\widehat\theta-\theta\|^2
\ \ge\
\frac{c_*}{n\sigma_{\min,+}^2}.
\]
The infimum is over every estimator of $\theta$, including those that use $C$ and $\Omega$.
\end{proposition}

\subsection{A sample-size floor that no design can lower}
\label{sec:floor}

Equation~\eqref{eq:var-smin} is a statement about whatever design is in hand. Combining it with Theorem~\ref{thm:attenuation} turns it into a requirement on the data that holds before any design is chosen. Proposition~\ref{prop:twopoint} upgrades the same quantity from an oracle-GLS variance to a minimax risk.

\begin{corollary}[Strategic sample-size floor]
\label{cor:floor}
Let $\{(R_n,E_n,P_n,q_n)\}$ be any sequence of measurement designs that is $\eta_n$-clustered at a fixed interior product baseline $\qbar$ and identifies the restricted payoff modulo global location for every $n$, with total sample size $n$. Under the measurement-noise conditions of Proposition~\ref{prop:noise},
\[
\lambda_{\max}\operatorname{Var}(\widehat\theta)
\ \ge\
\frac{1}{\kappa^{2}}\cdot\frac{1}{n\,\eta_n^{2\dPhi}} ,
\]
so the largest oracle-GLS variance eigenvalue vanishes only if
\begin{equation}
\label{eq:floor}
n\,\eta_n^{2\dPhi}\longrightarrow\infty .
\end{equation}
In the Gaussian experiment of Proposition~\ref{prop:twopoint} the same rate is necessary for vanishing minimax risk. Neither requirement depends on the number of transition regimes, on the number of policy environments, or on the transition kernels. Equivalently, a degree-$\dPhi$ rival interaction read from opponent policies dispersed by $\eta$ requires at least an effective-sample factor of order $\eta^{-2\dPhi}$ before that oracle variance, or that Gaussian risk, can vanish.
\end{corollary}

If a design is rank deficient beyond location, consistent estimation of the full location-normalized payoff already fails, so the rank-identifying case is the relevant boundary for the rate calculation. The theorem gives a necessary lower bound; a particular design can be worse.

Condition \eqref{eq:floor} is the empirically operative form of the theory, because $\eta$ is measurable directly from observed opponent behavior. It also separates the two integers cleanly. The feature dimension $K$ says how many distinct policy environments must be found at all; the highest active degree $\dPhi$ says how much data each of them must carry. A restriction can be generous on the first margin and punitive on the second.

In the airline panel of Section~\ref{sec:empirical} the raw cross-stratum rival-policy dispersion is $\widehat\eta=1.24$, but after netting binomial sampling noise the corrected radius is $0.68$, and the width ratio of the confidence set suggests, as a heuristic scale comparison, an effective dispersion near $0.17$; this number is not used for inference. At that heuristic benchmark a degree-one rival effect already requires an effective sample of order $35$ relative to a rival-independent contrast, and a degree-two effect of order $10^{3}$. The panel contains $2{,}204$ route-quarters on $241$ route clusters; in the clustered asymptotics, the independent sampling unit is the route cluster. The frontier between what this design can and cannot measure therefore falls between a payoff that ignores the rival and a payoff that lets the rival enter, and it falls there for reasons that no reweighting or richer transition model can remove without genuinely more dispersed rival-policy variation.

\section{Inference that does not require vanishing risk}
\label{sec:inference}

Theorem~\ref{thm:attenuation} says that some payoff directions are measured arbitrarily badly, and Corollary~\ref{cor:floor} says how badly. Those directions, together with the observed radius $\eta$, index a local-to-zero experiment for the Hotz--Miller moment. A confidence interval built from an asymptotic normal approximation to $\widehat\theta$ presumes that the design pins $\theta$ down well enough for that approximation to bite. The set below is Anderson--Rubin inversion of the same moment, with first-stage error in $\widehat C_\Phi$ inside $\Omega_n(\theta)$. Proposition~\ref{prop:ar} gives coverage without a bound on $\sigma_{\min,+}$. Corollary~\ref{cor:width} and Proposition~\ref{prop:local} give the two sampling regimes the attenuation bound distinguishes: vanishing risk, and a noncentral $\chi^2$ experiment indexed by $\tau=\lim\sqrt n\,\sigma_{\min,+}$.

\subsection{The moment function and the variance that is usually omitted}
\label{sec:moment}

Fix $\beta$ and a feature matrix $\Phi$, and stack \eqref{eq:z} over the measured environments. Write $C_\Phi=C(\beta)J_\Phi$ for the operator restricted to $\mathcal U_\Phi$, where $J_\Phi$ embeds $\theta$ into $\R^{MAB}$, and write
\begin{equation}
\label{eq:amoment}
a
=
d(p_i)-\Delta\tau-S(\beta)\bigl[g(p_i)+\tau_{i0}\bigr]
\end{equation}
for the part of the measurement that is linear in the payoff, so that $a=C_\Phi\theta$ holds exactly at the truth. The sample analogue replaces the conditional choice probabilities and the transition kernels by estimates. Both objects in \eqref{eq:amoment} move: $\widehat a$ inherits first-stage error through $d(\widehat p_i)$, $g(\widehat p_i)$, and $\widehat S(\beta)$, and $\widehat C_\Phi$ inherits it through the estimated opponent policy and the estimated kernels. Define
\begin{equation}
\label{eq:mtheta}
m_n(\theta)
=
\widehat a-\widehat C_\Phi\theta .
\end{equation}

The variance of \eqref{eq:mtheta} is not the variance of $\widehat a$. It also contains the design-matrix variance
\[
\operatorname{Var}(\widehat C_\Phi\theta),
\]
which is quadratic in $\theta$, and the covariance between $\widehat a$ and $\widehat C_\Phi\theta$, which is linear in $\theta$. Standard practice plugs in estimated first-stage objects and then proceeds as though the design matrix were known. These omitted terms can be negligible when the payoff is well determined, but can become first-order or dominant in weakly measured directions, because a weakly identified direction is one along which $\theta$ must be large to move the measurement at all. In the environment of Theorem~\ref{thm:attenuation}, dropping first-stage error in the design matrix therefore reports a well-measured payoff along directions the design barely sees.

\subsection{The confidence set}
\label{sec:arset}

Let $\Omega_n(\theta)$ denote the asymptotic variance of $\sqrt n\,m_n(\theta)$ and let $\widehat\Omega_n(\theta)$ be a consistent estimator. Define the Anderson--Rubin statistic and set
\begin{equation}
\label{eq:ar}
\mathrm{AR}_n(\theta)
=
n\,m_n(\theta)'\widehat\Omega_n(\theta)^{+}m_n(\theta),
\qquad
\mathcal C_n(1-\alpha)
=
\bigl\{\theta:\ \mathrm{AR}_n(\theta)\le c_{L,1-\alpha}\bigr\},
\end{equation}
with $L=\operatorname{rank}\Omega_n$ and $c_{L,1-\alpha}$ the $(1-\alpha)$ quantile of $\chi^2_L$.

\begin{proposition}[Coverage that does not depend on identification strength]
\label{prop:ar}
Let the sample consist of $n$ independent markets. Suppose the first-stage estimator $\widehat\psi$ of $\psi=(p_i,q,P)$ satisfies $\sqrt n(\widehat\psi-\psi)\Rightarrow N(0,V)$, that the map $\psi\mapsto(a,C_\Phi)$ is continuously differentiable at $\psi$ with Jacobian $D(\theta)$ for the composite $\psi\mapsto a-C_\Phi\theta$, and that $\widehat\Omega_n(\theta_0)\to_p\Omega(\theta_0)=D(\theta_0)VD(\theta_0)'$. Let $L=\operatorname{rank}\Omega(\theta_0)$, and suppose
\[
\Pr\bigl\{\operatorname{rank}\widehat\Omega_n(\theta_0)=L\bigr\}\to 1.
\]
The nonzero eigenvalues of $\Omega(\theta_0)$ are uniformly bounded in $[c,C]\subset(0,\infty)$. Then for every $\theta_0$ satisfying the model,
\[
\mathrm{AR}_n(\theta_0)\Rightarrow\chi^2_L,
\qquad
\Pr\bigl\{\theta_0\in\mathcal C_n(1-\alpha)\bigr\}\to 1-\alpha .
\]
The convergence is uniform over any class of designs and data-generating processes on which those eigenvalue bounds hold and the first-stage remainder is uniformly negligible. No condition is imposed on $\sigma_{\min,+}(C_\Phi)$, and none is available: by Theorem~\ref{thm:attenuation} that quantity can be made arbitrarily small by clustering opponent policies, which is a property of the design rather than of the sampling problem.
\end{proposition}

The proposition says only that the set covers, whatever the design. Shape is governed by the attenuation theorem.

\subsection{Projections, unbounded directions, and the width the theorem forces}
\label{sec:projection}

Applied work reports scalar summaries: a switching cost, an entry cost, the payoff consequence of one more active rival. Each is a linear functional $\lambda'\theta$. The projection interval is
\[
\mathrm{CI}_\lambda(1-\alpha)
=
\Bigl[\inf_{\theta\in\mathcal C_n}\lambda'\theta,\ \sup_{\theta\in\mathcal C_n}\lambda'\theta\Bigr],
\]
which covers $\lambda'\theta_0$ with asymptotic probability at least $1-\alpha$ because $\theta_0\in\mathcal C_n$ implies $\lambda'\theta_0\in\mathrm{CI}_\lambda$. The set $\mathcal C_n$ is defined with $\widehat\Omega_n(\theta)$ evaluated at the candidate, which is the object whose coverage Proposition~\ref{prop:ar} proves. Holding $\widehat\Omega_n$ at a preliminary value makes the set an ellipsoid and the interval \eqref{eq:proj} explicit; that ellipsoid is a computational approximation, not the reported set. Write $\widehat W=\widehat C_\Phi'\widehat\Omega_n(\widehat\theta)^{+}\widehat C_\Phi$ and let $\widehat\theta$ be the minimizer of $\mathrm{AR}_n$ at that frozen covariance. Then
\begin{equation}
\label{eq:proj}
\mathrm{CI}_\lambda^{\mathrm{ellip}}
=
\lambda'\widehat\theta
\pm
\sqrt{\tfrac1n\bigl(c_{L,1-\alpha}-\mathrm{AR}_n(\widehat\theta)\bigr)\,\lambda'\widehat W^{+}\lambda}
\quad\text{if }\lambda\in\operatorname{range}\widehat W,
\end{equation}
and $\mathrm{CI}_\lambda^{\mathrm{ellip}}=\R$ otherwise. An unbounded interval reports that the observed variation does not restrict that payoff contrast: $\lambda$ lies in the null space the design failed to remove.

The reported interval inverts $\mathrm{AR}_n(\theta)\le c_{L,1-\alpha}$ with $\widehat\Omega_n(\theta)$ recomputed at each candidate. For a linear contrast the inversion is the profile
\begin{equation}
\label{eq:profile-ci}
\mathrm{CI}_\lambda(1-\alpha)
=
\bigl\{t:\ \min_{\theta:\,\lambda'\theta=t}\mathrm{AR}_n(\theta)\le c_{L,1-\alpha}\bigr\},
\end{equation}
which is the same as $\bigl[\inf_{\theta\in\mathcal C_n}\lambda'\theta,\ \sup_{\theta\in\mathcal C_n}\lambda'\theta\bigr]$. Restricting the minimizer to the ray $\theta=\widehat\theta+s\lambda$ yields a subset of that interval; ray inversion is used only as a computational lower bound on width. Because $m_n$ is affine in $\theta$, one cluster-bootstrap sample of $(\widehat a^{*},\widehat C_\Phi^{*})$ delivers $\widehat\Omega_n(\theta)$ everywhere. Bounded directions of that set are governed by the same spectrum that Theorem~\ref{thm:attenuation} bounds.

\begin{corollary}[Width under a slack condition]
\label{cor:width}
Under the conditions of Proposition~\ref{prop:ar} and the eigenvalue bounds of Proposition~\ref{prop:noise}, suppose $\sqrt n\,\sigma_{\min,+}(\widehat C_\Phi)\to_p\infty$. Let $\widehat\lambda$ be a unit right singular vector of $\widehat C_\Phi$ associated with $\sigma_{\min,+}(\widehat C_\Phi)$. For every fixed $\delta\in(0,c_{L,1-\alpha})$, on the event $\{\mathrm{AR}_n(\widehat\theta)\le c_{L,1-\alpha}-\delta\}$,
\[
\bigl|\mathrm{CI}_{\widehat\lambda}(1-\alpha)\bigr|
\ \ge\
\frac{\sqrt{\delta}\,c_\Omega^{1/2}}{\sqrt n\,\sigma_{\min,+}(\widehat C_\Phi)}\,(1+o_p(1))
\ \ge\
\frac{\sqrt{\delta}\,c_\Omega^{1/2}}{\kappa\sqrt n\,\eta^{\dPhi}}\,(1+o_p(1)),
\]
where $c_\Omega$ lower-bounds the eigenvalues of $\Omega_n$ on its range. On that event,
\[
\bigl|\mathrm{CI}_{\widehat\lambda}\bigr|
=
\Omega_p\bigl((\sqrt n\,\sigma_{\min,+})^{-1}\bigr).
\]
The slack event has limiting probability one when the design is just identified. When it is overidentified by $L-r$ degrees the event has limiting probability $P(\chi^2_{L-r}\le c_{L,1-\alpha}-\delta)>0$, and on the complementary event the set can be arbitrarily thin because the minimized statistic can sit arbitrarily close to the cutoff. The requirement $\sqrt n\,\sigma_{\min,+}\to_p\infty$ is \eqref{eq:floor} in the operator's units: it makes the local comparison with the frozen-covariance ellipsoid take place on a shrinking neighborhood of $\widehat\theta$. Along a triangular array with $\sigma_{\min,+}\to 0$ faster than $n^{-1/2}$, that comparison is not available.
\end{corollary}

\begin{proposition}[Local-to-zero experiment]
\label{prop:local}
Under the conditions of Proposition~\ref{prop:ar}, consider a triangular array of rank-identifying, $\eta_n$-clustered designs, and write $\sigma_n=\sigma_{\min,+}(C_{\Phi,n})$. Theorem~\ref{thm:attenuation} gives
\[
\limsup_{n\to\infty}\sqrt n\,\sigma_n
\ \le\
\kappa\limsup_{n\to\infty}\sqrt{n\eta_n^{2\dPhi}},
\]
allowing either limsup to be infinite. Coverage of $\mathcal C_n$ at $\theta_0$ continues to hold along the array, uniformly over the class of Proposition~\ref{prop:ar}.

Suppose in addition that $\sqrt n\,\sigma_n\to\tau\in[0,\infty)$, that $\lambda_n$ is a unit right singular vector of $C_{\Phi,n}$ for $\sigma_n$ with left vector $u_n=C_{\Phi,n}\lambda_n/\sigma_n$ when $\sigma_n>0$, and that $(\lambda_n,u_n)\to(\lambda,u)$. Let $\theta_n(a)=\theta_0+a\lambda_n$, and suppose the eigenvalue bounds of Proposition~\ref{prop:ar} hold uniformly for $\theta$ in a fixed ball about $\theta_0$. Then for every fixed $a\in\R$,
\[
\mathrm{AR}_n\bigl(\theta_n(a)\bigr)
\ \Rightarrow\
\chi^2_L\bigl(\delta(a,\tau)\bigr),
\qquad
\delta(a,\tau)
=
a^2\tau^2\,u'\Omega(\theta_0+a\lambda)^{+}u,
\]
with $\delta(a,0)=0$. Consequently:
\begin{enumerate}
\item[(i)] If $\tau=0$, then $\mathrm{AR}_n(\theta_n(a))\Rightarrow\chi^2_L$ for every fixed $a$, and $\bigl|\mathrm{CI}_{\lambda_n}(1-\alpha)\bigr|\to_p\infty$.
\item[(ii)] If $\tau\in(0,\infty)$, the ray inversion through $\lambda_n$ has width $O_p(1)$ and $\Omega_p(1)$, so the profiled interval is $\Omega_p(1)$.
\item[(iii)] If $\sqrt n\,\sigma_{\min,+}(\widehat C_{\Phi,n})\to_p\infty$, Corollary~\ref{cor:width} applies and the interval shrinks.
\end{enumerate}
\end{proposition}

The local parameter $a$ is order one in payoff units because $\sigma_n$ is order $n^{-1/2}$. The generic weakly identified GMM experiment is that of \citet{andrewsguggenberger2017,andrewsguggenberger2019}. Product-policy geometry supplies the sequence $\tau$ and the contrast $\lambda_n$: Theorem~\ref{thm:attenuation} bounds $\tau$ by $\kappa\sqrt{\gamma}$ along $n\eta^{2\dPhi}\to\gamma$.

\subsection{Estimating the variance without an analytic Jacobian}
\label{sec:bootstrap}

The Jacobian $D(\theta)$ in Proposition~\ref{prop:ar} is available in closed form, but it is unpleasant, and it changes with every specification of the first stage. A market-level cluster bootstrap of the entire pipeline delivers $\widehat\Omega_n(\theta)$ at every candidate. Resample markets, re-estimate choice probabilities and kernels, and rebuild $(\widehat a^{*(b)},\widehat C_\Phi^{*(b)})$. Write $m_n^{*(b)}(\theta)=\widehat a^{*(b)}-\widehat C_\Phi^{*(b)}\theta$ and
\begin{equation}
\label{eq:boot-omega}
\begin{aligned}
\xi_n^{*(b)}(\theta)
&=
\sqrt n\bigl(m_n^{*(b)}(\theta)-m_n(\theta)\bigr),\\
\widehat\Omega_n(\theta)
&=
\frac1{\Bboot-1}
\sum_{b=1}^{\Bboot}
\bigl(\xi_n^{*(b)}(\theta)-\bar\xi_n^*(\theta)\bigr)
\bigl(\xi_n^{*(b)}(\theta)-\bar\xi_n^*(\theta)\bigr)',
\end{aligned}
\end{equation}
The factor $\sqrt n$ puts $\widehat\Omega_n$ on the scale of $\operatorname{Var}(\sqrt n\,m_n)$ in \eqref{eq:ar}; omitting it would estimate $\operatorname{Var}(m_n)$ and inflate the statistic by $n$. Here $\bar\xi_n^*(\theta)$ is the average of the $\Bboot$ scores. Because $m_n$ is affine in $\theta$, one set of draws serves every $\theta$. The reported set inverts that statistic. The ellipsoid \eqref{eq:proj} is computed from the same draws at a single preliminary $\theta$ and is not the reported object. Proposition~\ref{prop:ar} is a $\chi^2_L$ limit for a consistent covariance, which requires $\Bboot\to\infty$ or an analytic Jacobian. With $\Bboot>L$ held fixed, Lemma~\ref{lem:hotelling} replaces the cutoff by the Hotelling value $c_{L,1-\alpha}=L(\Bboot-1)(\Bboot-L)^{-1}F_{L,\Bboot-L,1-\alpha}$ in the Gaussian limit experiment for the bootstrap scores. As $\Bboot\to\infty$ the two cutoffs coincide.

\subsection{Monte Carlo evidence}
\label{sec:mc}

We simulate a dynamic entry game with $J$ binary rivals and logit shocks, draw a panel of markets, estimate the first stage from counts, and invert \eqref{eq:ar} with $\widehat\Omega_n(\theta)$ evaluated at the candidate. Table~\ref{tab:mc} records the geometry and the conventional interval. Table~\ref{tab:mc-n} is the size check for Proposition~\ref{prop:ar}. Within each layer of the interaction filtration the reported contrast is the most favorably conditioned singular direction, so the widths are the most favorable local-conditioning statement each design supports. This selection is conservative for the local conditioning comparison. We do not claim that the selected direction globally minimizes the candidate-dependent profiled projection width. Joint coverage inverts $\mathrm{AR}_n(\theta_0)$ with $\widehat\Omega_n(\theta_0)$. The AR column of Table~\ref{tab:mc} inverts the profiled statistic \eqref{eq:profile-ci} at the true contrast. The width column inverts only along the identified ray through $\widehat\theta$; that interval is contained in the profiled projection, so the reported widths are a lower bound on the width of \eqref{eq:profile-ci}.

At $8{,}000$ markets the joint AR frequencies are $0.950$, $0.975$, and $0.950$ in the three designs (Table~\ref{tab:mc-n}). The two $J=1$ designs share a payoff and a baseline policy and differ only in the policy radius, by a factor of twelve in $\eta$ and ten in $\sigma_{\min,+}$; coverage is no worse in the weaker design. Projection intervals in Table~\ref{tab:mc} cover at $1.000$ everywhere, which is the expected conservatism of a projection. At $500$ markets, where $L$ is $36$ or $75$, the minimized statistic still carries first-stage bias and joint coverage is $0.845$, $0.895$, and $0.770$ in the same designs. Raising the number of bootstrap draws from $300$ to $600$ at the $J=2$ design leaves coverage at $0.73$ and $0.72$, so the covariance estimate is not the problem.

The conventional interval does not share the large-$n$ size. It is below nominal in every design at $500$ markets, at $0.905$ and $0.900$ for the degree-zero contrast in the two $J=1$ designs and as low as $0.780$ in the $J=2$ design. At the weakly identified contrasts the Wald interval is not obviously deficient, because those intervals are so wide that any procedure covers. Undercoverage appears at the well-measured contrasts, where the interval is short enough for the omitted design-matrix variance to matter, and it is worst in the design with the highest interaction degree. At $8{,}000$ markets Wald coverage is still $0.750$ and $0.775$ for the two rival-dependent $J=2$ contrasts, and $0.900$ for the well-measured $J=1$ contrast.

The widths in Table~\ref{tab:mc} reproduce the filtration. At $\eta=0.05$ the degree-one contrast is measured worse than the degree-zero contrast by a factor of $25.0$, against the $\eta^{-1}=20$ that Corollary~\ref{cor:width} predicts. At $\eta=0.6$ the same ratio is $2.00$ against a predicted $1.67$. In the two-rival design the degree-one and degree-two ratios are $16.6$ and $3756$ at $\eta=0.093$, where $\eta^{-1}=10.8$ and $\eta^{-2}=116$; the bound is an inequality, and both layers sit above it, with the gap widening in the degree exactly as the exponent requires.

The fifth column of Table~\ref{tab:mc} checks Theorem~\ref{thm:attenuation} across designs. The implied constant $\sigma_{\min,+}/\eta^{\dPhi}$ is computed for designs that differ in their transition kernels and in their policy radius by more than an order of magnitude. If the attenuation exponent were an artifact of a particular construction, that constant would move with $\eta$. Between the two $J=1$ designs it moves by a factor of $1.2$ while $\eta$ moves by a factor of twelve.

Table~\ref{tab:mc-gamma} varies $\eta$ at a fixed panel of $2{,}000$ markets, so $\gamma=n\eta^{2}$ indexes the local-to-zero sequence of Proposition~\ref{prop:local} for $J=1$ and $\dPhi=1$. Joint AR coverage stays between $0.888$ and $0.950$ as $\gamma$ runs from $0.8$ to $320$. The degree-zero ray width is stable. The degree-one ray width falls from $8.77$ to $0.30$, and the ratio to the degree-zero width tracks $\eta^{-1}$.

\begin{table}[t]
\centering
\caption{Monte Carlo geometry and conventional intervals. Panels of $500$ markets over $14$ periods, logit shocks, $\beta=0.9$, nominal $95$ percent, cluster bootstrap of the entire pipeline. $L$ is the number of moments and $p$ the payoff dimension. The AR column inverts the profiled statistic \eqref{eq:profile-ci} at the true contrast. Widths invert only along the identified ray through $\widehat\theta$ and are therefore a lower bound on the profiled projection. Wald is the conventional interval that treats the design matrix as known. Degree $d$ indexes the layer of the interaction filtration, and the reported contrast is the most favorably conditioned singular direction of that layer. Width ratios are relative to the degree-zero contrast of the same design. Joint coverage of the vector $\theta_0$ is Table~\ref{tab:mc-n}, because at $500$ markets the first-stage inversion still biases the minimized statistic.}
\label{tab:mc}
\small
\begin{tabular}{lrrrrrrrr}
\toprule
design & $L/p$ & $\eta$ & $\sigma_{\min,+}$ & $\sigma_{\min,+}/\eta^{\dPhi}$ & $d$ & AR & Wald & width ratio \\
\midrule
$J=1$ separated & $36/16$ & $0.600$ & $3.6\times10^{-2}$ & $0.060$ & 0 & $1.000$ & $0.905$ & $1.00$ \\
                &         &         &                    &         & 1 & $1.000$ & $0.935$ & $2.00$ \\
\addlinespace
$J=1$ clustered & $36/16$ & $0.050$ & $3.7\times10^{-3}$ & $0.075$ & 0 & $1.000$ & $0.900$ & $1.00$ \\
                &         &         &                    &         & 1 & $1.000$ & $0.945$ & $25.0$ \\
\addlinespace
$J=2$ clustered & $75/24$ & $0.093$ & $1.3\times10^{-5}$ & $0.0015$ & 0 & $1.000$ & $0.840$ & $1.00$ \\
                &         &         &                    &          & 1 & $1.000$ & $0.780$ & $16.6$ \\
                &         &         &                    &          & 2 & $1.000$ & $0.840$ & $3756$ \\
\bottomrule
\end{tabular}
\end{table}

\begin{table}[t]
\centering
\caption{Size of the joint Anderson--Rubin set and of the conventional interval, designs held fixed. Proposition~\ref{prop:ar} is an $n\to\infty$ statement. At $8{,}000$ markets the AR joint frequencies are $0.950$, $0.975$, and $0.950$. At $500$ markets they are not, because $L$ is $36$ or $75$ and the Hotz--Miller inversion is nonlinear. Wald coverage in weakly measured directions does not rise with $n$. At $500$ and $2{,}000$ markets, $60$ replications; at $8{,}000$, $80$ for $J=1$ and $40$ for $J=2$. A dash means the design has no contrast of that degree.}
\label{tab:mc-n}
\small
\begin{tabular}{lrrrrr}
\toprule
design & markets & AR joint & Wald $d=0$ & Wald $d=1$ & Wald $d=2$ \\
\midrule
$J=1$ separated & $500$ & $0.850$ & $0.950$ & $0.967$ & --- \\
 & $2{,}000$ & $0.967$ & $0.917$ & $0.950$ & --- \\
 & $8{,}000$ & $0.950$ & $0.900$ & $0.888$ & --- \\
\addlinespace
$J=1$ clustered & $500$ & $0.933$ & $0.950$ & $0.900$ & --- \\
 & $2{,}000$ & $0.917$ & $0.983$ & $0.867$ & --- \\
 & $8{,}000$ & $0.975$ & $0.950$ & $0.912$ & --- \\
\addlinespace
$J=2$ clustered & $500$ & $0.733$ & $0.783$ & $0.783$ & $0.883$ \\
 & $2{,}000$ & $0.900$ & $0.800$ & $0.700$ & $0.933$ \\
 & $8{,}000$ & $0.950$ & $0.900$ & $0.750$ & $0.775$ \\
\bottomrule
\end{tabular}
\end{table}

\begin{table}[t]
\centering
\caption{Local-to-zero sequence, $J=1$, $\dPhi=1$, $2{,}000$ markets, $80$ replications, nominal $95$ percent. The design and payoff are those of the $J=1$ rows in Table~\ref{tab:mc}; only the policy radius changes. $\gamma=n\eta^{2}$. Widths are ray inversions, as in Table~\ref{tab:mc}. Joint AR coverage does not fall with $\gamma$. The degree-one width does.}
\label{tab:mc-gamma}
\small
\begin{tabular}{rrrrrrr}
\toprule
$\eta$ & $\gamma$ & $\sqrt n\,\sigma_{\min,+}$ & AR joint & $|\mathrm{CI}|_{d=0}$ & $|\mathrm{CI}|_{d=1}$ & ratio \\
\midrule
$0.02$ & $0.8$ & $0.050$ & $0.950$ & $0.110$ & $8.77$ & $79.8$ \\
$0.05$ & $5$ & $0.155$ & $0.900$ & $0.115$ & $2.69$ & $23.3$ \\
$0.10$ & $20$ & $0.277$ & $0.912$ & $0.111$ & $1.21$ & $10.8$ \\
$0.20$ & $80$ & $0.387$ & $0.888$ & $0.121$ & $0.575$ & $4.74$ \\
$0.40$ & $320$ & $0.724$ & $0.938$ & $0.107$ & $0.296$ & $2.77$ \\
\bottomrule
\end{tabular}
\end{table}

\section{Unobserved heterogeneity}
\label{sec:heterogeneity}

Every result so far treats the identified player's payoff as common across the measurements that are pooled, which is Assumption~\ref{ass:common}. Applied work on dynamic games rarely believes that literally, and the standard remedy is a finite mixture over latent market types. When the mixture is resolved the bound survives intact and binds through the worst-served type. When it is not resolved, the apparent strategic variation that the design relies on can be an artifact of the mixture.

\subsection{Resolved types}
\label{sec:types-resolved}

Let markets carry a latent type $\omega\in\{1,\ldots,T\}$ with weights $\pi_\omega>0$, and let the payoff $u_\omega$, the kernels $P^r_\omega$, and the opponent policies $q^e_\omega$ all be type-specific. Suppose the mixture is resolved, in the sense that type-specific choice probabilities and kernels are identified from the panel by an auxiliary argument \citep{kasaharashimotsu2009,hushum2012}. Then \eqref{eq:z} holds type by type, the unknown is the concatenation $(u_1,\ldots,u_T)$, and the stacked operator is block diagonal because no measurement of one type restricts the payoff of another.

\begin{proposition}[Heterogeneity does not relax the attenuation bound]
\label{prop:het}
Suppose each type's design is $\eta_\omega$-clustered at an interior product baseline $\qbar_\omega$, and let $\dPhiomega$ be the highest active degree of the interaction filtration of $F=\operatorname{col}\Phi$ at $\qbar_\omega$. Then the stacked restricted operator satisfies
\[
\sigma_{\min,+}\bigl(C^\Phi_{1:T}\bigr)
=
\min_{\omega\le T}\sigma_{\min,+}\bigl(C^\Phi_\omega\bigr)
\ \le\
\min_{\omega\le T}
\kappa_\omega\,\eta_\omega^{\dPhiomega},
\]
with $\kappa_\omega=\kappa(M,A,\beta,\Phi,\qbar_\omega)$. Moreover, for every type $\omega$, vanishing oracle-GLS variance and vanishing Gaussian minimax risk require
\[
n\,\pi_\omega\,\eta_\omega^{2\dPhiomega}\longrightarrow\infty.
\]
\end{proposition}

Block diagonality gives the equality. Theorem~\ref{thm:attenuation} applies within each identified type-specific block with its own baseline-adapted degree $\dPhiomega$ and constant $\kappa_\omega$. Taking the minimum across blocks gives the displayed bound. The effective sample size for type $\omega$ is $n\pi_\omega$, yielding the type-specific rate condition. Allowing heterogeneity multiplies the number of payoff parameters by $T$ and divides the effective sample by $\pi_\omega$, while the binding attenuation is set by whichever type sees the least rival-policy dispersion at its own filtration degree. A design can therefore look adequate on the pooled data and be hopeless within the type that matters. The mixture redistributes the variation that is already there.

\subsection{Unresolved types}
\label{sec:types-unresolved}

If types are not resolved, the analyst observes the mixed choice probability $\bar p=\sum_\omega\pi_\omega p_{i,\omega}$ and the mixed kernel. A mixture of type-specific equilibria is not, in general, an equilibrium of the pooled primitives: the Hotz--Miller map and the resolvent are both nonlinear, so the operator built from pooled objects need not be the operator of any homogeneous type. Assumption~\ref{ass:common} then fails, and the mixed objects do not supply identifying variation for a common payoff. The statistic of Section~\ref{sec:arset} can reject the pooled model when the design is overidentified, but an empty set does not isolate heterogeneity from other misspecification, and when $L=\operatorname{rank}C_\Phi$ the minimized statistic is zero by construction. We do not record a theorem for this case. The operational content is in Section~\ref{sec:composition} and in the airline diagnostic: composition across strata contaminates measured $\eta$, and the common-payoff assumption is maintained rather than tested.

\subsection{Composition is not policy variation}
\label{sec:composition}

Where $\eta$ comes from is part of the design. Applied work usually takes policy environments from strata: calendar windows, geographic cells, size classes. If markets in different strata differ in type composition, the estimated opponent policies differ across strata even when every type plays exactly the same policy everywhere. That difference is a composition effect. It contributes to a measured $\eta$, and it identifies nothing, because Assumption~\ref{ass:common} fails across the strata being pooled: the payoff attached to the measurement is a different mixture in each.

The comparison is testable. Under the null that all strata share one policy at each state, the estimated cross-stratum dispersion still has a nondegenerate distribution driven by cell sizes alone. Comparing the measured $\eta$ with that null distribution asks whether the design has any strategic variation to work with before asking what the variation identifies. Section~\ref{sec:empirical} runs the comparison on the airline panel.

Three further results sit in the appendix because they are not needed to read the attenuation bound. Unknown patience is identified at the same cardinalities by an arbitrarily small centered continuation-offset randomization (Appendix~\ref{app:unknown}). Those offsets, and the clustered product policies themselves, are locally implementable by transfers at a regular interior equilibrium (Appendix~\ref{app:transfer}). Lagged-action transition families preserve a residual gauge of dimension at least $|\mathcal Z|$ no matter how many policies are stacked (Appendix~\ref{app:lag}).

\section{A design diagnostic: U.S.\ airline entry}
\label{sec:empirical}

The section is a diagnostic of the interaction filtration in an observed, rank-deficient design, not a structural payoff estimate and not a direct application of the full-rank strength corollary. The observed design contains too little clean strategic variation to measure rival-dependent payoff components precisely. With one binary rival, $\dPhi=1$. Rank classifies some rival-dependent directions as identified; it does not say those directions are measured. Calendar windows and distance terciles also differ in fuel prices and route composition, so Assumption~\ref{ass:common} does not hold exactly across the strata being pooled. Measured $\eta$ is therefore an upper bound on usable strategic dispersion, and the intervals below understate how weakly rival-dependent payoffs are measured.

The source is the T-100 Domestic Segment file for 2016--2019 \citep{bts_t100}. The identified player is American, the rival is Southwest, and the universe is undirected city-pairs among the forty largest 2014 cities. An action is at least twelve performed departures in the city-pair-quarter. After the lag and lead, $2{,}204$ route-quarters remain on $241$ routes where both carriers appear. Southwest is active in $95.4$ percent of those cells. The public state is $x=(s,a_{i,-1},a_{j,-1})$ with $s$ a 2014 demand tercile, so $M=12$ and $B=2$. The discount $\beta=0.95$ is maintained. With one binary rival the filtration has two layers, $\dPhi=1$. In rank-restoring designs with one binary rival, the theory predicts first-order sensitivity of rival-dependent directions to rival-policy dispersion.

Markets are partitioned into six strata by calendar window and route-distance tercile. Within each stratum the rival CCP and the demand transition are estimated from counts. Nothing is perturbed. The six stacked measurements give $L=72$ moments.

\begin{table}[t]
\centering
\caption{Is there rival-policy variation beyond sampling noise? Measured $\widehat\eta$ against a null that redraws rival actions from the pooled policy at the observed cell sizes. Four hundred draws. $p$-values use $\hat p=(1+\#\{T_b\ge T_{\mathrm{obs}}\})/(B+1)$. The corrected radius nets binomial variance out of the cross-stratum spread.}
\label{tab:eta-placebo}
\small
\begin{tabular}{lrrrrrc}
\toprule
grid & strata & median cell & $\widehat\eta$ & corrected & $95\%$ of null & $p$ \\
\midrule
headline, $M=12$ & 6 & $8$ & $1.236$ & $0.684$ & $0.796$ & $<0.003$ \\
four strata, $M=12$ & 4 & $10$ & $0.894$ & $0.454$ & $0.908$ & $0.080$ \\
coarse demand, $M=8$ & 4 & $18.5$ & $0.924$ & $0.296$ & $0.938$ & $0.087$ \\
distance only, $M=12$ & 3 & $16$ & $1.139$ & $0.960$ & $0.775$ & $<0.003$ \\
sixty-state grid, $M=60$ & 6 & $0$ & $0.875$ & $0.503$ & $1.035$ & $0.342$ \\
\bottomrule
\end{tabular}
\end{table}

Raw cross-stratum dispersion is $\widehat\eta=1.24$. The median cell has eight observations, so most of that number is binomial noise. On the headline grid the observed value exceeds the $95$th percentile of the null; on three of five grids it does not (Table~\ref{tab:eta-placebo}). Calendar and distance strata also differ in route composition, and a composition effect contributes to measured $\eta$ without identifying anyone's payoff (Section~\ref{sec:composition}). The noise-corrected radius $0.68$ is therefore an upper bound on the strategic dispersion the design supplies. The width ratio in the confidence set below suggests, as a heuristic scale comparison, an effective dispersion near $0.17$; this number is not used for inference.

\begin{table}[t]
\centering
\caption{Headline grid, $\beta=0.95$, nominal $95$ percent. Rank is not full: $K=2$ has rank $42$ of $48$. Widths invert the profiled statistic \eqref{eq:profile-ci} in logit units for the most favorably conditioned singular direction of each interaction degree. Observed behavior spans $8.9$ such units. A dash means the specification has no contrast of that degree. The ellipsoid computed from a preliminary $\widehat\Omega_n$ is recorded in the appendix and is not this table.}
\label{tab:airline-ar}
\small
\begin{tabular}{lrrrrrrr}
\toprule
specification & $L$ & rank & $\sigma_{\min,+}$ & $\min\mathrm{AR}$ & cutoff & $|\mathrm{CI}|_{d=0}$ & $|\mathrm{CI}|_{d=1}$ \\
\midrule
$K=1$ & 72 & 20 & $2.2\times10^{-3}$ & $87.5$ & $106.0$ & $0.29$ & --- \\
$K=2$ & 72 & 42 & $8.8\times10^{-5}$ & $11.1$ & $106.0$ & $9.0$ & $52.7$ \\
\bottomrule
\end{tabular}
\end{table}

Table~\ref{tab:airline-ar} is the filtration in the panel. A payoff that ignores the rival is not rejected on this grid ($\mathrm{AR}=87.5$ against $106.0$), nor on the other four. Holding $\widehat\Omega_n$ at a preliminary least-squares $\theta$ produces a minimized statistic of $120.9$ and a rejection; that ellipsoid is not the set Proposition~\ref{prop:ar} covers, and the rejection does not survive the candidate-dependent inversion the proposition requires. Under saturation the reported set is never empty and never informative about the rival: the most favorably conditioned rival-dependent interval is $52.7$ units against a behavioral range of $8.9$, while the most favorably conditioned rival-independent interval in the same specification is $9.0$, already the width of observed behavior. Rank is $42$ of $48$. The sixty-state grid, whose median cell is empty, still reports rank $192$ of $240$. Rank on the stacked operator classifies some rival-dependent payoff directions as identified. The confidence set shows that even the most favorably conditioned such direction can remain extremely weakly measured: it records $\eta^{\dPhi}$ at the layer Theorem~\ref{thm:attenuation} names, which rank does not.

Appendix~\ref{app:emp-robust} records the other grids, the two other discounts, and the rank-ceiling exercise in which the identifying policies were constructed rather than observed.

Figure~\ref{fig:airline} collects the design facts. On three of five grids, raw cross-stratum dispersion of rival policy does not exceed the sampling-noise benchmark (panel A). On every grid the most favorably conditioned rival-dependent profiled interval is wider than the range of observed behavior (panel B). Under the local $n^{-1/2}$ benchmark, matching the headline degree-one width of $52.7$ logit units to the $8.9$-unit behavioral range requires an information-equivalent sample multiplier $m=(52.7/8.9)^{2}\approx 35$. Matching it to the rival-independent width of $9.0$ gives $m\approx 34$. These numbers are not the exact number of additional routes required, and finite-sample Anderson--Rubin nonlinearity would make a literal $n$-extrapolation misleading. They are a local information-equivalent benchmark. Panel C records the theoretical iso-information curves: at degree one, ten times less usable variation requires roughly one hundred times as much independent information; at degree two the penalty is fourth-order.

\begin{figure}[t]
\centering
\includegraphics[width=\textwidth]{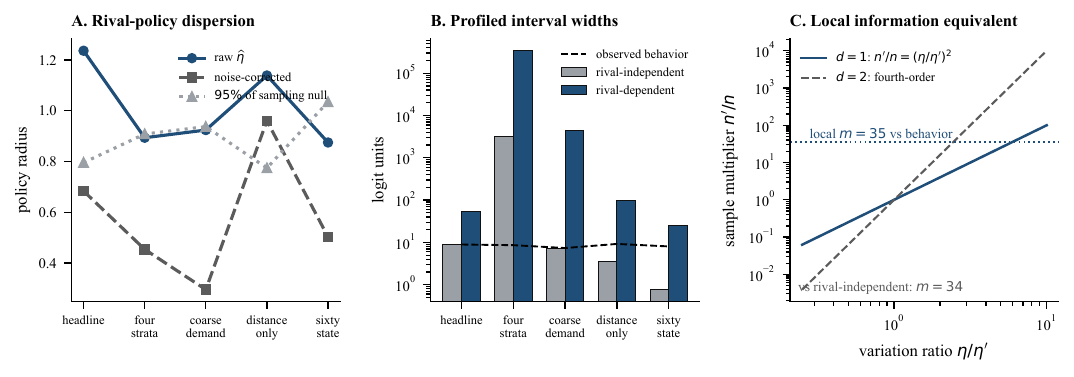}
\caption{What the observed airline design measures. Panel A: rival-policy radius against a cell-size sampling null. Panel B: most favorably conditioned profiled widths by interaction degree, with the range of observed own-choice logits marked. Panel C: theoretical iso-information curves $n'/n=(\eta/\eta')^{2d}$ and the local sample multiplier implied by the headline degree-one interval. The application is a design diagnostic, not a structural estimate of Southwest's payoff effect.}
\label{fig:airline}
\end{figure}

\section{Discussion}
\label{sec:disc}

Feature dimension $K$ says how many distinct opponent-policy environments restore rank. The highest active degree $\dPhi$ says how much data each of those environments must carry. Extra transition regimes do not close the second gap. A researcher who can read $\eta$ off opponent behavior can diagnose which interaction layers are intrinsically hard to measure before estimating payoffs; exact rank still depends on the full transition-policy geometry. Corollary~\ref{cor:width} and Proposition~\ref{prop:local} are the intervals that match that diagnosis in the two sampling regimes.

Once a player's pair $(u,\beta)$ is known up to a global additive constant, that player's action-value differences and best-response map are identified under a specified counterfactual primitive and a specified opponents' policy. Absolute lifetime welfare is not, because a constant added to flow payoffs shifts values by $c/(1-\beta)$ without changing behavior. A single player's identified payoff does not identify the equilibrium correspondence.

\clearpage
\appendix
\section{Hotz--Miller operator and gauge identities}
\label{app:operators}

Given $(P,q)$, let $B_a(q)$ average current payoffs against $q_x$ and let $P_a(q)$ be the corresponding own-action transition. After eliminating the value function,
\begin{equation}
\label{eq:A}
A_a(q;P,\beta)
=
B_a(q)-B_0(q)
+\beta\bigl(P_a(q)-P_0(q)\bigr)\bigl(I-\beta P_0(q)\bigr)^{-1}B_0(q).
\end{equation}
Stacking $a\neq 0$ produces the block $C^{r,e}(\beta):=A(q^e;P^r,\beta)$, of size $M(A-1)\times MAB$. Let $S^{r,e}(\beta)$ be the continuation operator $\beta(P_a-P_0)(I-\beta P_0)^{-1}$ stacked over $a\neq 0$. Action-value differences satisfy
\begin{equation}
\label{eq:d-expand}
d(p_i^{r,e})
=
C^{r,e}(\beta)u
+
\Delta\tau_i^{r,e}
+
S^{r,e}(\beta)\bigl[g(p_i^{r,e})+\tau_{i0}^{r,e}\bigr],
\end{equation}
where $\Delta\tau_i^{r,e}=(\tau_{ia}^{r,e}-\tau_{i0}^{r,e})_{a\neq a_0}$. The observable dependent variable subtracts the contemporaneous difference once:
\begin{equation}
\label{eq:zdef}
z^{r,e}
=
d(p_i^{r,e})
-
\Delta\tau_i^{r,e}.
\end{equation}
The full Hotz--Miller observation is then \eqref{eq:z} in the text, with
\begin{equation}
\label{eq:bdef}
b^{r,e}(\beta)
=
S^{r,e}(\beta)\bigl[g(p_i^{r,e})+\tau_{i0}^{r,e}\bigr].
\end{equation}
Because $S^{r,e}(\beta)\ones=0$, only the centered component of $g(p_i^{r,e})+\tau_{i0}^{r,e}$ matters. When $\beta$ is unknown, a candidate $\tilde\beta$ is observationally admissible if and only if
\begin{equation}
\label{eq:Rbeta}
R(\tilde\beta)=\min_u\bigl\|z-b(\tilde\beta)-C(\tilde\beta)u\bigr\|^2
\end{equation}
is zero.

\section{Universal obstructions}
\label{app:gauge}
\label{app:avg}

\begin{proof}[Proof of Proposition~\ref{prop:env}]
Fix a row-stochastic kernel $P$ and $h\in\R^M$. The array $u=G_{P,\beta}h$ adds $h(x)-\beta\sum_{x'}P(x'\mid x,a,b)h(x')$ to the flow payoff. For any opponent policy and any own strategy, the continuation value at state $x$ shifts by $h(x)$. Own-action value differences, choice probabilities, and best responses are therefore unchanged, so $\mathrm{Im}\,G_{P,\beta}$ lies in the observational kernel of every mixed policy observed under $P$.

If $G_{P,\beta}h=0$, then $h=\beta P_{ab}h$ for every action profile $(a,b)$. Taking the sup norm gives $\|h\|_\infty\le\beta\|h\|_\infty$. Since $0<\beta<1$, $h=0$. Thus $G_{P,\beta}$ is injective and $\dim\mathrm{Im}\,G_{P,\beta}=M$. Constants $h\equiv c$ produce the constant-flow line $(1-\beta)c\ones$. Stacking policies cannot leave the image.
\end{proof}

\begin{lemma}[Two-regime intersection]
\label{lem:E1}
For every row-stochastic pair $(P,Q)$ and every $\beta\in(0,1)$,
\[
\dim\bigl(\mathrm{Im}\,G_{P,\beta}\cap\mathrm{Im}\,G_{Q,\beta}\bigr)\ge 1,
\]
and the constant-flow line is contained in the intersection. The intersection equals that line on a nonempty Zariski-open set of pairs: $\mathrm{rank}[G_P,-G_Q]=2M-1$ at the identity/cycle witness, hence on a Zariski-open neighborhood.
\end{lemma}

\begin{proof}
$G_P\ones=G_Q\ones=(1-\beta)\ones$, so the intersection contains the constant-flow line. Equivalently $\mathrm{rank}[G_P,-G_Q]\le 2M-1$. For the witness, let every block of $P$ equal $I_M$ and let every block of $Q$ equal $I_M$ except one block equal to the $M$-cycle $\Ccycle$. If $G_Ph=G_Qk$, identity blocks force $h=k$. The cycle block then forces $h(x)-\beta h(x)=h(x)-\beta h(x+1)$, hence $\Ccycle h=h$, hence $h=c\ones$. Thus the intersection is exactly the constant line and $\mathrm{rank}[G_P,-G_Q]=2M-1$. Entries of $[G_P,-G_Q]$ are affine in the pair of row-stochastic arrays. Any $(2M-1)\times(2M-1)$ minor is a polynomial, nonzero at the witness, hence not identically zero. The vanishing locus is a proper Zariski-closed subset.
\end{proof}

Theorem~\ref{thm:sharp} supplies exact two-regime sufficiency without passing through this lemma. The lemma records only that two kernels always share location and that generic pairs share nothing more.

\begin{proof}[Proof of Proposition~\ref{prop:strat}]
Let $q^1,\ldots,q^E$ be interior. At each state $x$ form $Q_x\in\R^{E\times B}$ with rows $q^e(\cdot\mid x)^\prime$. Consider own-action-invariant perturbations $\Delta u(x,a,b)=h_x(b)$ and seek $c\in\R^E$ such that $Q_x h_x=c$ for all $x$, i.e.\ $q^e(\cdot\mid x)^\prime h_x=c_e$ for all $e,x$.

The unknown pair $(h_1,\ldots,h_M,c)$ has dimension $MB+E$. The system contains $ME$ linear equations, so the solution space $\mathcal H$ has dimension at least $MB+E-ME=MB-(M-1)E$. If $h=0$ then $c=0$, so projection onto $h$ is injective on $\mathcal H$. Independently, the global payoff constant lies in every observational kernel. Hence
\[
\dim\ker C_E(\beta)\ge\max\bigl\{1,\,MB-(M-1)E\bigr\}.
\]

Fix any observed policy $e$, any primitive regime, and any nonreference own action $a$. Because $\Delta u$ is common across own actions, its current payoff-difference contribution is zero. Under $q^e$, the reference expected payoff perturbation is $c_e\ones$. Let $P_0^e$ and $P_a^e$ be the induced transitions. Stochasticity gives $(I-\beta P_0^e)^{-1}\ones=(1-\beta)^{-1}\ones$ and $(P_a^e-P_0^e)\ones=0$, so
\[
\beta(P_a^e-P_0^e)(I-\beta P_0^e)^{-1}(c_e\ones)=0.
\]
Every Hotz--Miller choice-difference equation is unchanged. The argument does not use the particular transition law.

Identification up to one location requires the right-hand side of the dimension bound to equal $1$, which is \eqref{eq:Estar}. For $B=2$, $(2M-1)/(M-1)=2+1/(M-1)$, so $\Estar=3$. Equivalently $\Estar=B+\lceil(B-1)/(M-1)\rceil$.
\end{proof}

For $B=2$ and $E=2$, write $h_x(0)=\rho_0(x)$ and $h_x(1)=\rho_0(x)+\rho_1(x)$. If $q^1_x\neq q^2_x$ at every $x$, then $D_1-D_2$ is invertible and \eqref{eq:rho-example} parametrizes the two-dimensional kernel: $c_1$ is location and $t$ is the extra direction. If the policies coincide on a nonempty set of states, a nonzero vector supported there lies in $\ker(D_1-D_2)$ and again yields an extra null direction.

\section{Product-policy construction (saturated case)}
\label{app:product}
\label{app:suf}

Let $U=M^{-1}\ones_M\ones_M'$ and let $\Ccycle$ be the $M$-cycle permutation. Set $\Deltacycle=\Ccycle-I$ and $0<\alpha<1/M$, so that $U+\alpha\Deltacycle$ is strictly positive and row-stochastic. Let opponent $j$ have $B_j$ actions, so $B=\prod_j B_j$. For $0<\eta<1$,
\begin{equation}
\label{eq:Rj}
R_j(\eta)=\eta I_{B_j}+(1-\eta)B_j^{-1}\ones\ones',
\qquad
R(\eta)=\bigotimes_j R_j(\eta).
\end{equation}
With $\Nextra=\Estar-B$, state-dependent selection matrices $G_x$ yield
\begin{equation}
\label{eq:Qxprod}
Q_x
=
\begin{bmatrix} I_B \\ G_x \end{bmatrix} R(\eta)
\end{equation}
satisfying $Q_x\ones_B=\ones_E$, $\mathrm{rank}\,Q_x=B$, and
\begin{equation}
\label{eq:cap}
\bigcap_{x=1}^M\mathrm{col}(Q_x)=\mathrm{span}\{\ones_E\}.
\end{equation}

\begin{lemma}[Product-policy intersection]
\label{lem:product}
Let $E=\Estar(M,B)$ and $\Nextra=E-B$. Let opponent $j$ have $B_j\ge 2$ actions, with $B=\prod_j B_j$. For $0<\eta<1$ define $R_j(\eta)$ and $R(\eta)$ by \eqref{eq:Rj}. Then $R(\eta)$ is invertible, row-stochastic, and strictly interior, and every row is a product distribution. There exist selection matrices $G_x\in\R^{\Nextra\times B}$ such that the matrices $Q_x$ in \eqref{eq:Qxprod} satisfy $Q_x\ones_B=\ones_E$, $\mathrm{rank}\,Q_x=B$, every row of every $Q_x$ is a product distribution,
\[
\bigcap_{x=1}^M\mathrm{col}(Q_x)=\mathrm{span}\{\ones_E\},
\]
and the family $\{Q_x\}_{x=1}^M$ comprises exactly $\Estar$ distinct Markov opponent-policy profiles. As $\eta\to 0$, all individual factor policies, and hence all joint product policies, approach uniform. Correlated joint opponent policies are not used.
\end{lemma}

\begin{proof}
Each $R_j(\eta)$ is row-stochastic and strictly interior. Its eigenvalues are $1$ (once) and $\eta$ (multiplicity $B_j-1$), so $R_j(\eta)$ is invertible on $(0,1)$. A Kronecker product of invertible matrices is invertible, hence so is $R(\eta)$. A Kronecker product of row-stochastic matrices is row-stochastic. Each row of a Kronecker product of row-stochastic matrices is the product of the corresponding factor rows, so every row of $R(\eta)$ is a strictly interior product distribution on the joint action space.

The definition of $\Estar$ yields $\Nextra=\Estar-B$ and
\[
(M-1)\Nextra\ge B-1\ge \Nextra.
\]
Index extra rows by $k=1,\ldots,\Nextra$ and states by $x=1,\ldots,M$. Construct a star with center $1$: every extra row of $G_1$ equals $e_1'$. For $x\ge 2$, extra row $k$ of $G_x$ equals $e_{j(x,k)}'$. Assign distinct non-1 indices in $\{2,\ldots,B\}$ to the $\Nextra$ extra profiles first (possible because $B-1\ge \Nextra$), so that every extra profile $k$ has at least one pair with $j(x,k)\neq 1$. Use the remaining pairs $(x,k)$ with $x\ge 2$ to cover any leftover base indices, so that the image of $j$ includes $\{2,\ldots,B\}$ (possible because $(M-1)\Nextra\ge B-1$). The associated bipartite graph, with left vertices $\{1,\ldots,B\}$ (joint actions) and right vertices $\{1,\ldots,\Nextra\}$ (extra rows), and an edge $\{b,k\}$ if some state selects action $b$ in extra row $k$, is connected.

The first $B$ rows of $R(\eta)$ are linearly independent, hence distinct as state-independent policies. Each extra profile equals row $1$ of $R(\eta)$ at state $1$ and a non-1 row at some $x\ge 2$, so it is not state-constant and therefore distinct from the first $B$ policies; distinct first non-1 assignments make the extra profiles distinct from one another. The construction therefore contains exactly $\Estar$ distinct Markov opponent-policy profiles, not merely $\Estar$ rows that may duplicate one another as complete state-contingent policies.

Every row of $Q_x$ is a row of $R(\eta)$, hence a product distribution, and $Q_x\ones_B=\ones_E$. Because $R(\eta)$ is invertible,
\[
\mathrm{col}(Q_x)
=
\mathrm{col}
\begin{bmatrix} I_B \\ G_x \end{bmatrix}.
\]
The stacked matrix on the right has full column rank $B$, so $\mathrm{rank}\,Q_x=B$.

A vector $c$ lies in $\bigcap_x\mathrm{col}(Q_x)$ if and only if there exist $v_x\in\R^B$ with
\[
c
=
\begin{bmatrix} I_B \\ G_x \end{bmatrix} v_x
\qquad\forall x.
\]
The first $B$ coordinates force $v_x=v$ independent of $x$. The extra coordinates then force $G_x v$ independent of $x$. Comparing state $1$ to state $x\ge 2$ gives $v_1=v_{j(x,k)}$ for every assigned pair. Covering $\{2,\ldots,B\}$ and connectedness of the bipartite incidence graph yield $v\in\mathrm{span}\{\ones_B\}$. Hence $c\in\mathrm{span}\{\ones_E\}$. The reverse inclusion holds because every $Q_x$ is row-stochastic.

As $\eta\to 0$, each $R_j(\eta)$ tends to the uniform kernel on $B_j$ actions.
\end{proof}

\begin{proof}[Proof of Corollary~\ref{thm:sharp}]
The obstructions of Propositions~\ref{prop:env} and~\ref{prop:strat} give $\Rmin\ge 2$ and $\Emin\ge\Estar$. It remains to show $\Iident(2,\Estar)=1$. Let $U=M^{-1}\ones_M\ones_M'$, let $\Ccycle$ be the $M$-cycle, and $\Deltacycle=\Ccycle-I$. Choose $0<\alpha<1/M$, so $U+\alpha\Deltacycle$ is strictly positive and row-stochastic. Take the product-policy matrices $Q_x$ of Lemma~\ref{lem:product}. Because $\beta$ is known, $b(\beta)$ is known from observed CCPs, the shock law, and transfers, and is subtracted. The remaining map is $C(\beta)$.

Regime $0$ sets $P^0_{ab}=U$ for every $(a,b)$. Continuation differences vanish. For nonreference $a$, the measurement under policy $e$ is $y^{0,e}_a(x)=q^e(\cdot\mid x)^\prime\delta_a(x,\cdot)$. Stacking the $E$ measurements at $x$ yields $Q_x\delta_a(x,\cdot)$. Full column rank of $Q_x$ identifies $\delta_a(x,b)$ for every $x,b$ and every nonreference $a$, hence $M(A-1)B$ coordinates.

Regime $1$ sets $P^1_{a_0b}=U$ for every $b$. For a distinguished nonreference action $a_\star$, set $P^1_{a_\star b}=U+\alpha\Deltacycle$ for every $b$. Remaining own-action kernels are any strictly positive stochastic matrices. Let $r(x,b)=u(x,a_0,b)$ and $r_e(x)=q^e(\cdot\mid x)^\prime r(x,\cdot)$. The resolvent identity $(I-\beta U)^{-1}=I+\beta(1-\beta)^{-1}U$ and $\Deltacycle U=0$ give $\Deltacycle(I-\beta U)^{-1}=\Deltacycle$. After subtracting the already identified current difference for $a_\star$, the observation is $\alpha\beta \Deltacycle r_e$. Because $\alpha\beta\neq 0$, the reference block is identified up to the kernel of $r\mapsto(\Deltacycle r_1,\ldots,\Deltacycle r_E)$. If this map annihilates $r$, then $r_e=c_e\ones_M$ for each $e$, i.e.\ $Q_x r(x,\cdot)=c$ for all $x$. Thus
\[
c\in\bigcap_x\mathrm{col}(Q_x)=\mathrm{span}\{\ones_E\},
\]
so $c=\lambda\ones_E$. Full column rank of $Q_x$ together with $Q_x\ones_B=\ones_E$ yields $r(x,\cdot)=\lambda\ones_B$ for all $x$. The only unidentified reference-payoff direction is the global constant, contributing rank $MB-1$. Total rank is $M(A-1)B+(MB-1)=MAB-1$. Hence $\Iident(2,\Estar)=1$, so $\Rmin=2$ and $\Emin=\Estar$. As $\eta\to 0$ the product policies approach uniform, so the identifying policy variation may be arbitrarily small.
\end{proof}

\section{Product-policy construction, structure theorem, and conditional frontier}
\label{app:restricted}

This appendix records the locked arguments for Theorem~\ref{thm:structure} and Proposition~\ref{prop:frontier}. The feature restriction applies only to opponent actions. Write
\[
u(x,a,b)=\phi(b)^\top\theta(x,a),
\]
with $\Phi\in\R^{B\times K}$ of rank $K$ and $\ones_B\in\operatorname{col}\Phi$. The unknown is $\theta\in\R^{MAK}$. Let $h_\Phi\in\R^K$ be the unique coefficient with $\Phi h_\Phi=\ones_B$. The associated global payoff constant is the vector that repeats $h_\Phi$ at every $(x,a)$.

\subsection{Feature-moment product-policy lemma}
\label{app:feature-moment}

Opponent joint actions factor as $B=\prod_{j=1}^{J}B_j$. For $0<\eta<1$ define $R_j(\eta)=\eta I_{B_j}+(1-\eta)B_j^{-1}\ones\ones'$ and $R(\eta)=\bigotimes_j R_j(\eta)$. Each $R_j(\eta)$ is invertible, and every row of $R(\eta)$ is a strictly interior product distribution, so $\operatorname{rank}(R(\eta)\Phi)=K$. Select $K$ rows whose feature-moment matrix $W\in\R^{K\times K}$ is nonsingular. Write
\[
\Ephi=K+\Bigl\lceil\frac{K-1}{M-1}\Bigr\rceil.
\]
Use the star construction of Lemma~\ref{lem:product} on those $K$ labels: the first $K$ rows of each selector $S_x$ are $I_K$, and each extra row is a standard-basis selector. For state $x$ the feature-moment matrix is $M_x=S_x W$.

\begin{lemma}[Feature-moment intersection]
\label{lem:feature-moment}
For every state, $\operatorname{rank} M_x=K$. If the induced hypergraph on the $K$ labels is connected, then
\[
\bigcap_x\operatorname{col}(M_x)=\operatorname{span}\{\ones_E\}.
\]
\end{lemma}

\begin{proof}
Because $W$ is invertible, $\operatorname{col}(M_x)=\operatorname{col}(S_x)$ and $\operatorname{rank} M_x=K$. If $y\in\bigcap_x\operatorname{col}(S_x)$, its first $K$ coordinates determine $z\in\R^K$. For each extra profile, membership in every $\operatorname{col}(S_x)$ forces all coordinates of $z$ corresponding to labels used by that profile across states to be equal. Connectedness of the induced hypergraph therefore forces $z=c\ones_K$, hence $y=c\ones_E$. The reverse inclusion holds because each $M_x$ is row-stochastic on the constant $h_\Phi$ direction: $M_x h_\Phi=\ones_E$.
\end{proof}

\subsection{Necessity}
\label{app:restricted-nec}

\begin{proof}[Proof of Theorem~\ref{thm:structure}, necessity]
Restrict attention to payoff perturbations invariant to own action, $\theta(x,a)=\psi_x$ for all $a$. For environment $e$, write $m_e(x)'=q^e(\cdot\mid x)'\Phi$. If $m_e(x)'\psi_x=c_e$ for all states $x$, the induced expected flow-payoff perturbation is the same constant $c_e$ at every state and every own action in that environment. It therefore changes lifetime utility only by an action-independent constant and leaves all observed choice differences unchanged, regardless of transition kernels. For each $e$, eliminating the nuisance scalar $c_e$ gives at most $M-1$ independent restrictions on the $MK$-dimensional vector $\psi$. Hence
\[
\dim\ker C_E^\Phi\ge MK-E(M-1),
\]
with the global payoff constant always remaining, so $\dim\ker C_E^\Phi\ge\max\{1,MK-E(M-1)\}$. Full identification modulo location therefore requires $E\ge\Ephi$.
\end{proof}

\subsection{Two-regime sufficiency}
\label{app:restricted-suf}

\begin{proof}[Proof of Theorem~\ref{thm:structure}, two-regime sufficiency]
Let $U=M^{-1}\ones\ones'$, $\Deltacycle=\Ccycle-I$ with $\Ccycle$ an $M$-cycle, and $0<\alpha<1/M$ small enough that $U+\alpha\Deltacycle$ is strictly positive. Regime $0$ sets $P^0_{ab}=U$ for every $(a,b)$. Regime $1$ sets $P^1_{a_0b}=U$ for every $b$, and for one distinguished nonreference action $a_\star$ sets $P^1_{a_\star b}=U+\alpha\Deltacycle$; remaining own-action kernels in regime $1$ equal $U$. Take the $\Ephi$ interior product profiles of Lemma~\ref{lem:feature-moment}. Reparameterize $r_x=\theta(x,a_0)$ and $\delta_{a,x}=\theta(x,a)-\theta(x,a_0)$. Under regime $0$, continuation differences vanish and the homogeneous equations are $M_x\delta_{a,x}=0$. Full column rank of $M_x$ yields $\delta_{a,x}=0$ for every $x$ and every $a\neq a_0$. For the distinguished action under regime $1$, after $\delta=0$, the remaining equation is proportional to $\Deltacycle r_e=0$, where $r_e(x)=m_e(x)'r_x$. Hence every $r_e$ is constant across states: there exists $c\in\R^E$ with $M_x r_x=c$ for all $x$. Therefore
\[
c\in\bigcap_x\operatorname{col}(M_x)=\operatorname{span}\{\ones_E\}.
\]
Because $M_x h_\Phi=\ones_E$ and $M_x$ has full column rank, $r_x=c_0 h_\Phi$ for all $x$. The only kernel direction is the global payoff constant, so $\operatorname{rank} C^\Phi=MAK-1$.
\end{proof}

\subsection{Restricted dynamic gauge}
\label{app:restricted-gauge}

Let $F=\operatorname{col}\Phi$ and let the columns of $N_\Phi$ span $F^\perp$. For a primitive kernel $P$ write $T_{x,a}(P)$ for the $B\times M$ matrix with row $b$ equal to $P_{ab}(x,\cdot)$, and
\[
\mathcal{R}_\Phi(P)=\operatorname{stack}_{x,a} N_\Phi^\top T_{x,a}(P).
\]

\begin{proof}[Proof of Theorem~\ref{thm:structure}, restricted gauge]
For $h\in\R^M$,
\[
(G_{P,\beta}h)(x,a,\cdot)=h(x)\ones_B-\beta T_{x,a}(P)h.
\]
Because $\ones_B\in F$, the first term always belongs to $F$. Therefore $G_{P,\beta}h\in\mathcal U_\Phi$ if and only if $T_{x,a}(P)h\in F$ for every $(x,a)$, or equivalently $N_\Phi^\top T_{x,a}(P)h=0$ for all $(x,a)$. The admissible potentials are $\ker\mathcal{R}_\Phi(P)$. For $0<\beta<1$ the gauge map is injective, so
\[
\dim\bigl(\mathcal U_\Phi\cap\operatorname{Im} G_{P,\beta}\bigr)=\dim\ker\mathcal{R}_\Phi(P)=M-\operatorname{rank}\mathcal{R}_\Phi(P).
\]
Always $\ones_M\in\ker\mathcal{R}_\Phi(P)$, hence $\operatorname{rank}\mathcal{R}_\Phi(P)\le M-1$. If $K=B$, then $F=\R^B$, $N_\Phi$ is empty, $\operatorname{rank}\mathcal{R}_\Phi(P)=0$, and the intersection has dimension $M$.

If $K<B$, let $U$ denote the uniform kernel on $\mathcal X$. Select $M-1$ distinct state--own-action cells $(x_\ell,a_\ell)$, $\ell=1,\ldots,M-1$. Let $v_1,\ldots,v_{M-1}$ be a basis of $\ones_M^\perp$. Take a nonzero $n\in\operatorname{col}(\Phi)^\perp$. Because $\ones_B\in\operatorname{col}(\Phi)$, one has $n^\top\ones_B=0$. For each $\ell=1,\ldots,M-1$ set
\[
P_{a_\ell b}(x_\ell,\cdot)
=
U(x_\ell,\cdot)
+
\varepsilon n_b v_\ell^\top.
\]
Set every remaining primitive transition row equal to $U$. Take $\varepsilon>0$ small enough that all transition probabilities are strictly positive. Then
\[
n^\top T_{x_\ell,a_\ell}(P)
=
\varepsilon\|n\|^2 v_\ell^\top.
\]
Linear independence of $v_1,\ldots,v_{M-1}$ yields $\operatorname{rank}\mathcal{R}_\Phi(P)\ge M-1$. Combined with $\ones_M\in\ker\mathcal{R}_\Phi(P)$,
\[
\operatorname{rank}\mathcal{R}_\Phi(P)=M-1.
\]
A corresponding maximal minor is a nonzero polynomial in transition probabilities. Thus outside a proper algebraic subset of the interior stochastic-kernel parameter space, $\operatorname{rank}\mathcal{R}_\Phi(P)=M-1$, and for generic strictly positive $P$,
\[
\mathcal U_\Phi\cap\operatorname{Im} G_{P,\beta}=\operatorname{span}\{\text{global constant}\}.
\]
\end{proof}

\subsection{One-regime existence}
\label{app:one-regime}

The following construction is sufficient. Combined with Proposition~\ref{prop:frontier} it yields the upper bound $\cE_\Phi(1)\le K+\Ephi$.

\begin{proof}[Proof of Theorem~\ref{thm:structure}, one-regime existence]
Assume $K<B$. Choose $K$ strictly interior product distributions $p_1,\ldots,p_K$ with nonsingular feature-moment matrix $W_0$. Write $\Qbase$ for the $K\times B$ matrix with those rows. Then $\Qbase$ has rank $K$, so there is $w\neq 0$ with $\Qbase w=0$.

If $w=\Phi c$ for some $c\in\R^K$, then
\[
\Qbase w=\Qbase\Phi c=W_0 c.
\]
But $\Qbase w=0$ and $W_0$ is nonsingular, hence $c=0$ and $w=0$, a contradiction. Therefore $w\notin\operatorname{col}\Phi$.

Use one primitive regime with $P_{a_0b}=U$ and $P_{ab}=U+\alpha w_b\Deltacycle$ for $a\neq a_0$, with $\alpha$ small enough for strict positivity. Under the first $K$ constant policy environments, $p_i^\top w=0$, so continuation differences vanish. Nonsingularity of $W_0$ identifies every current payoff difference, $\delta_a(x)=0$.

The set of $K$-tuples of product policies with nonsingular feature-moment matrix is open and nonempty. For nonzero $w$, $q\mapsto q^\top w$ is not identically zero on the interior product-policy manifold because product distributions span $\R^B$. Its zero set therefore has empty interior. Thus the nonsingular-moment set can be perturbed, arbitrarily slightly, so that every selected policy also satisfies $q^\top w\neq 0$. Use these $K$ policies as bases for the connected-hypergraph construction and extend them to exactly $\Ephi$ state-dependent product-policy profiles. Then $q_x^{e\top}w\neq 0$ for every state and environment, and $\bigcap_x\operatorname{col}(M_x)=\operatorname{span}\{\ones\}$. Once $\delta=0$, the remaining homogeneous equation is $D(q^e w)\Deltacycle r_e=0$. Every diagonal entry is nonzero, so $\Deltacycle r_e=0$. The feature-moment intersection then forces the reference payoff to be a global constant. Hence $\operatorname{rank} C^\Phi=MAK-1$ at $R=1$ and $E=K+\Ephi$.
\end{proof}

\subsection{Conditional frontier}
\label{app:frontier}

\begin{proposition}[Conditional identifying frontier]
\label{prop:frontier}
Assume the conditions of Theorem~\ref{thm:structure}, and $K<B$. Then
\[
\cE_\Phi(R)=\Ephi\qquad\text{for all }R\ge 2,
\]
and
\[
\EoneLB\le\cE_\Phi(1)\le K+\Ephi.
\]
If $K=1$, then $\Ephi=1$ and $\EoneLB=2$, while the explicit one-regime construction uses $K+\Ephi=2$, so
\[
\cE_\Phi(1)=2,
\qquad
\cE_\Phi(R)=1\quad(R\ge 2).
\]
\end{proposition}

A one-regime design with $E$ policy environments has only $EM(A-1)$ choice-difference rows, so $\operatorname{rank} C_{1,E}^\Phi\le EM(A-1)$. Combined with the averaging bound,
\begin{equation}
\label{eq:ker1}
\dim\ker C_{1,E}^\Phi
\ge
\max
\bigl\{
1,\,
MAK-EM(A-1),\,
MK-E(M-1)
\bigr\}.
\end{equation}

\begin{proposition}[The one-regime count is the counting bound]
\label{prop:e1exact}
Let $K<B$ and let $(M,A,\Phi)$ range over the grid of model sizes $M\in\{2,3,4\}$, own-action counts $A\in\{2,3\}$, and rival-feature families in Table~\ref{tab:e1}, subject to $MAK\le 72$. On every one of the $47$ resulting cells,
\[
\cE_\Phi(1)=\EoneLB .
\]
In $29$ of them $\EoneLB<K+\Ephi$, so the one-regime construction of Theorem~\ref{thm:structure}(C) is not minimal.
\end{proposition}

\begin{table}[t]
\centering
\caption{Rival-feature families used in Proposition~\ref{prop:e1exact}, for $J$ binary rivals. Only proper restrictions $K<B$ enter, since $K=B$ leaves the full dynamic-potential gauge under one regime and no $E$ suffices.}
\label{tab:e1}
\small
\begin{tabular}{lcc}
\toprule
family & $\phi(b)$ & $K$ \\
\midrule
no rival dependence & $1$ & $1$ \\
active-rival count & $(1,\textstyle\sum_j b_j)$ & $2$ \\
additive rivals & $(1,b_1,\ldots,b_J)$ & $J+1$ \\
pairwise interactions & $(1,b_j,b_ib_j)$ & $1+J+\binom J2$ \\
\bottomrule
\end{tabular}
\end{table}

\begin{proof}[Proof of Proposition~\ref{prop:frontier}]
Necessity in Theorem~\ref{thm:structure}(A) yields $\cE_\Phi(R)\ge\Ephi$ for every $R$. Part (B) attains $E=\Ephi$ at $R=2$. For $R>2$, append any additional strictly positive primitive regimes to that identifying two-regime stack; adding measurement rows cannot reduce rank. Hence $\cE_\Phi(R)=\Ephi$ for all $R\ge 2$.

A one-regime design contributes $EM(A-1)$ choice-difference rows, so $\operatorname{rank} C_{1,E}^\Phi\le EM(A-1)$ and $\dim\ker C_{1,E}^\Phi\ge MAK-EM(A-1)$. Combined with the averaging bound, \eqref{eq:ker1} holds. Identification up to location requires the right-hand side of \eqref{eq:ker1} to equal $1$, which is \eqref{eq:E1LB}. The construction of Appendix~\ref{app:one-regime} yields $\cE_\Phi(1)\le K+\Ephi$.

If $K=1$, then $\Ephi=1$. For $M,A\ge 2$,
\[
1<\frac{MA-1}{M(A-1)}\le 2,
\]
so $\lceil(MA-1)/(M(A-1))\rceil=2$ and $\EoneLB=2$. The explicit construction uses $K+\Ephi=2$, so $\cE_\Phi(1)=2$ and $\cE_\Phi(R)=1$ for $R\ge 2$.
\end{proof}

\begin{proof}[Proof of Proposition~\ref{prop:e1exact}]
The lower bound is \eqref{eq:E1LB}. For each of the $47$ cells we exhibit a strictly positive rational kernel and $\EoneLB$ strictly interior rational product policies, map every entry into $\mathbb F_p$ with $p=2{,}147{,}483{,}647$, and compute $\operatorname{rank}C_{1,\EoneLB}^\Phi$ by exact elimination. The rank equals $MAK-1$ in every cell, and equals at most $MAK-2$ at $\EoneLB-1$ whenever $\EoneLB>1$. A nonzero minor over $\mathbb Q$ remains nonzero over $\mathbb F_p$ unless $p$ divides it, so the rank over $\mathbb F_p$ is a lower bound for the rank over $\mathbb Q$. Lower semicontinuity of rank extends each exhibit to an open dense set of kernels and policies. The exhibits, the prime, and the elimination are frozen in the regression suite.
\end{proof}

The grid is finite. Whether $\cE_\Phi(1)=\EoneLB$ for every $(M,A,\Phi)$ outside it remains open. The two-regime count $\cE_\Phi(R)=\Ephi$ for $R\ge 2$ is not.

\subsection{Basis invariance}
\label{app:basis}

For $T\in\mathrm{GL}(K)$, the map $\Phi\mapsto\Phi T$ is an invertible reparameterization of $\theta$. Column spaces, $K$, $\Ephi$, $\dPhi$, algebraic ranks, kernel dimensions, and $\operatorname{rank}\mathcal{R}_\Phi(P)$ are unchanged. Singular values of $C_\Phi$, and therefore the raw number $\sigma_{\min,+}$, need not be invariant under nonorthogonal basis scaling. The exponent in Theorem~\ref{thm:attenuation} is.

The proof of Corollary~\ref{thm:sharp} in Appendix~\ref{app:suf} is the special case $K=B$ of the two-regime argument above, with $W=R(\eta)$ and $M_x=Q_x$.

\begin{proof}[Proof of Remark~\ref{rem:open}]
After location normalization, $\mathrm{rank}\,C(\beta)=MAB-1$ at the construction, so some maximal minor is nonzero. For every fixed known $\beta\in(0,1)$, the entries of $C(\beta)$ are continuous functions of the strictly positive transition probabilities and of the interior product-policy parameters. A nonzero polynomial (equivalently, a nonzero continuous minor) remains nonzero on a neighborhood. Therefore an open set of strictly positive transition regimes and strictly interior product-policy profiles around the construction continues to identify the saturated payoff up to location.
\end{proof}

\section{Design-free strategic attenuation}
\label{app:attenuation}

This appendix proves Theorem~\ref{thm:attenuation}. Throughout, $\qbar=\bigotimes_j\qbar_j$ is a fixed interior product baseline, $F=\operatorname{col}\Phi$, and $\dPhi\ge 1$ is the highest active degree \eqref{eq:dPhi} of the $\qbar$-adapted filtration \eqref{eq:Hd}.

\subsection{The product expansion}
\label{app:atten-lemma}

Everything rests on one algebraic fact: a product perturbation pairs a $d$-way contrast only against a $d$-way perturbation.

\begin{lemma}[Attenuation]
\label{lem:atten}
Let $q=\bigotimes_{j=1}^{J}(\qbar_j+\eta h_j)$ be a product probability vector with $\ones'h_j=0$ and $\|h_j\|_1\le 1$ for every $j$, and let $0<\eta\le 1$. If $f\in\bigoplus_{r\ge d}\cH_r(\qbar)$, then
\[
|q'f|\ \le\ c(f)\,\eta^{d},
\qquad
c(f)=\sum_{t\,:\,|\mathrm{supp}(t)|\ge d}|c_t|\prod_{j\in \mathrm{supp}(t)}\|e^{j}_{t_j}\|_\infty,
\]
where $f=\sum_t c_t\bigotimes_j e^{j}_{t_j}$ is the expansion of $f$ in any factorwise $\qbar_j$-orthogonal basis whose zeroth element is the constant.
\end{lemma}

\begin{proof}
Expanding the product,
\[
q'f=\sum_{T\subseteq\{1,\ldots,J\}}\eta^{|T|}
\Bigl(\bigotimes_{j\in T}h_j\otimes\bigotimes_{j\notin T}\qbar_j\Bigr)'f .
\]
Take a single basis tensor $\bigotimes_j e^{j}_{t_j}$ with support $S=\mathrm{supp}(t)$. The pairing factorizes as
\[
\prod_{j\in T\cap S}\bigl(h_j'e^{j}_{t_j}\bigr)
\prod_{j\in T\setminus S}\bigl(h_j'\ones\bigr)
\prod_{j\in S\setminus T}\bigl(\qbar_j'e^{j}_{t_j}\bigr)
\prod_{j\notin S\cup T}\bigl(\qbar_j'\ones\bigr).
\]
For $j\in T\setminus S$ the factor is $h_j'\ones=0$, so the pairing vanishes unless $T\subseteq S$. For $j\in S\setminus T$ the factor is $\qbar_j'e^{j}_{t_j}=0$ by $\qbar_j$-centering, so the pairing vanishes unless $S\subseteq T$. Hence only $T=S$ survives and the basis tensor contributes exactly $\eta^{|S|}\prod_{j\in S}(h_j'e^{j}_{t_j})$. Because $f$ lies in the tail $\bigoplus_{r\ge d}\cH_r(\qbar)$, every $t$ with $c_t\ne 0$ has $|S|\ge d$. Finally $|h_j'e^{j}_{t_j}|\le\|h_j\|_1\|e^{j}_{t_j}\|_\infty\le\|e^{j}_{t_j}\|_\infty$, and $\eta^{|S|}\le\eta^{d}$ for $\eta\le 1$.
\end{proof}

Write $c_\Phi(\qbar)=c(f_\star)$ for a fixed unit-norm witness $f_\star\in F_{\ge\dPhi}$, which is nonzero by the definition of $\dPhi$.

\subsection{Proof of Theorem~\ref{thm:attenuation}}
\label{app:atten-proof}

\begin{proof}
Let $f_\star\in F_{\ge\dPhi}$ be the unit-norm witness and pick $c_\star\in\R^{K}$ with $\Phi c_\star=f_\star$, which exists because $f_\star\in F$. Define the payoff direction $\theta_\star\in\R^{MAK}$ by
\[
\theta_\star(x,a)=c_\star
\qquad\text{for every }(x,a),
\]
so $\theta_\star$ is invariant across the identified player's own actions.

\emph{Step 1: current-payoff differences vanish.} In the notation of Appendix~\ref{app:operators}, the current-payoff block satisfies $\bigl(B_a(q^e)\theta_\star\bigr)(x)=q^e(\cdot\mid x)'\Phi c_\star=q^e(\cdot\mid x)'f_\star$, which does not depend on $a$. Hence $B_a(q^e)\theta_\star-B_0(q^e)\theta_\star=0$ for every nonreference $a$, every policy environment, and every regime.

\emph{Step 2: the surviving channel is $O(\eta^{\dPhi})$.} By Step 1 and \eqref{eq:A}, the measurement reduces to the continuation term,
\[
C^{r,e}\theta_\star
=
\beta\bigl(P^{r,e}_a-P^{r,e}_0\bigr)\bigl(I-\beta P^{r,e}_0\bigr)^{-1}v^{e},
\qquad
v^{e}(x)=q^{e}(\cdot\mid x)'f_\star .
\]
Under \eqref{eq:cluster} each marginal is $q^e_j(\cdot\mid x)=\qbar_j+\eta h^{e}_j(x,\cdot)$ with $\ones'h^e_j(x,\cdot)=0$ and $\|h^e_j(x,\cdot)\|_1\le 1$, so Lemma~\ref{lem:atten} applies state by state and gives
\[
\|v^{e}\|_\infty\le c_\Phi(\qbar)\,\eta^{\dPhi}
\qquad\text{for every }e .
\]

\emph{Step 3: the continuation operator is uniformly bounded.} For any row-stochastic $P^{r,e}_0$ one has $\|\beta P^{r,e}_0\|_\infty=\beta<1$, hence $\|(I-\beta P^{r,e}_0)^{-1}\|_\infty\le(1-\beta)^{-1}$, and $\|P^{r,e}_a-P^{r,e}_0\|_\infty\le 2$. Therefore
\[
\bigl\|C^{r,e}\theta_\star\bigr\|_\infty
\le
\frac{2\beta}{1-\beta}\,c_\Phi(\qbar)\,\eta^{\dPhi}
\]
for every regime and every environment, with no dependence on the kernels. Each block contributes $M(A-1)$ rows and the stack is normalized by $(RE)^{-1/2}$, so
\[
\bigl\|C^{\Phi}\theta_\star\bigr\|_2
\le
\sqrt{M(A-1)}\,\frac{2\beta}{1-\beta}\,c_\Phi(\qbar)\,\eta^{\dPhi},
\]
uniformly in $R$ and $E$.

\emph{Step 4: from the test direction to $\sigma_{\min,+}$.} If $\ker C^{\Phi}$ is larger than the global payoff-location line, the first branch of the dichotomy holds and there is nothing further to prove. Otherwise, by Theorem~\ref{thm:structure} the kernel is the single payoff-location line spanned by $\theta_{\mathrm{loc}}$, the vector repeating $h_\Phi$ with $\Phi h_\Phi=\ones_B$. Let $\delta_\Phi=\|\theta_\star-\Pi_{\mathrm{loc}}\theta_\star\|_2>0$ be the distance from $\theta_\star$ to that line. The quantity $\delta_\Phi$ is fixed by $(\Phi,M,A,\qbar)$, because the top-filtration witness is defined relative to the baseline $\qbar$. Since $C^{\Phi}\theta_\star=C^{\Phi}(\theta_\star-\Pi_{\mathrm{loc}}\theta_\star)$ and $\sigma_{\min,+}$ is the minimum of $\|C^{\Phi}\theta\|_2/\|\theta\|_2$ over $\theta$ orthogonal to the kernel,
\[
\sigma_{\min,+}\bigl(C^{\Phi}\bigr)
\le
\frac{\|C^{\Phi}\theta_\star\|_2}{\delta_\Phi}
\le
\frac{2\beta}{1-\beta}\sqrt{M(A-1)}\,\frac{c_\Phi(\qbar)}{\delta_\Phi}\,\eta^{\dPhi},
\]
which is \eqref{eq:kappa}. If $\delta_\Phi=0$ then $\theta_\star$ lies on the location line, so $f_\star\in\operatorname{span}\{\ones_B\}\subseteq\cH_0(\qbar)$, contradicting $\dPhi\ge 1$.
\end{proof}

The proof never inspects the transition kernels beyond row stochasticity, which is why the bound survives arbitrarily sharp environmental variation. It never bounds the number of environments, because the attenuation is a property of each measurement separately rather than of their aggregate. And it uses the product form of opponent policies only through $\ones'h_j=0$ and the factorization of the pairing, so it is exactly the conditional independence delivered by Assumption~\ref{ass:rum} that generates the hierarchy. Correlated joint perturbations, which Remark~\ref{rem:joint} discusses and which are not feasible equilibrium policies here, would break the factorization in Lemma~\ref{lem:atten}.

\subsection{Attainment}
\label{app:atten-attain}

\begin{proof}[Proof of Corollary~\ref{cor:attain}]
Lemma~\ref{lem:mix} gives $R(\eta)V=W(\eta)D_\eta$ with $W(\eta)$ uniformly well conditioned and $D_\eta$ carrying exactly $m^\Phi_{\dPhi}$ diagonal entries of order $\eta^{\dPhi}$ and none smaller. Appendices~\ref{app:strength-current} and \ref{app:strength-ref} show that the corresponding directions are exposed at that order and no lower under the canonical construction with environmental spread bounded away from zero. Combining with Theorem~\ref{thm:attenuation} gives matching upper and lower orders, so $\sigma_{\min,+}\asymp\eta^{\dPhi}$. The matching lower order is shown only for this construction.
\end{proof}

\section{Restricted strength and interaction filtration}
\label{app:strength}

The spectrum below is for the filtration-preserving canonical representative
\[
P^0_{ab}=U
\quad\forall a,b,
\qquad
P^1_{a_0b}=U,
\qquad
P^1_{a_\star b}=U+\alpha\Deltacycle,
\qquad
P^1_{ab}=U
\quad(a\notin\{a_0,a_\star\}),
\]
together with the product-policy family whose $K$ base rows are indexed by the fixed set $I_\Phi$ constructed in Appendix~\ref{app:strength-sel}, extended by the star/hypergraph of Lemma~\ref{lem:feature-moment} to exactly $\Ephi$ profiles. Cardinality identification in Theorem~\ref{thm:structure} may use any $K$ independent rows of $R(\eta)\Phi$. The spectrum of Theorem~\ref{thm:strength} uses this fixed filtration-preserving selector. Other sharp-cardinality representatives can have different local singular-value rates.

\subsection{Regime subtraction}
\label{app:strength-R}

The two regime blocks are stacked, not initially block diagonal. Write
\[
\mathcal R
=
\begin{bmatrix}
I & 0 \\
-I & I
\end{bmatrix},
\qquad
\mathcal R^{-1}
=
\begin{bmatrix}
I & 0 \\
I & I
\end{bmatrix}.
\]
Both operators have condition numbers bounded by constants independent of $(\alpha,\eta)$. Left multiplication by $\mathcal R$ therefore preserves all singular-value asymptotic orders up to constant factors. After this regime subtraction, the current-difference block and the environmental/reference block can be analyzed separately. Multiplicities below are those of the subtracted operator, hence those of the original stacked operator up to $(\alpha,\eta)$-independent constants.

\subsection{Restricted mixing}
\label{app:strength-mix}

\begin{proof}[Proof of Lemma~\ref{lem:mix}]
Choose a fixed basis $\{v_{d,\ell}\}$ adapted to the filtration $F_{\ge d}\supset F_{\ge d+1}$. For every complement basis vector $v_{d,\ell}\in F_{\ge d}\setminus F_{\ge d+1}$, its degree-$d$ projection $P_d v_{d,\ell}$ is nonzero, and the set of those degree-$d$ projections is linearly independent. Therefore
\[
R(\eta)v_{d,\ell}
=
\eta^d\bigl(P_d v_{d,\ell}+O(\eta)\bigr).
\]
Writing all columns together gives $R(\eta)V=W(\eta)D_\eta$, where $W(\eta)\to W_0$, $W_0$ has full column rank, and
\[
D_\eta
=
\operatorname{diag}
\bigl(
\eta^0 I_{m_0^\Phi},\,
\eta^1 I_{m_1^\Phi},\,
\ldots,\,
\eta^J I_{m_J^\Phi}
\bigr).
\]
Hence $W(\eta)$ is uniformly well conditioned for small $\eta$, and the singular values of $R(\eta)|_F$ have orders $\eta^d$ with multiplicity $m_d^\Phi$. Because same-degree projections are linearly independent and distinct degrees lie in orthogonal interaction subspaces, $\operatorname{rank} W_0=K$.
\end{proof}

\subsection{Filtration-preserving selector}
\label{app:strength-sel}

Let $V=[v_{d,\ell}]$ be the fixed filtration-adapted basis of Lemma~\ref{lem:mix}, and write $W_0=[P_d v_{d,\ell}]_{d,\ell}$ for the $B\times K$ leading matrix. Select a fixed set of $K$ row indices $I_\Phi$ such that $S_{I_\Phi}W_0$ is nonsingular. These indices are read from $W_0$ once and do not depend on $\eta$. For each $\eta>0$, the $K$ base product policies are the rows of $R(\eta)$ indexed by $I_\Phi$. Then
\[
R(\eta)V=W(\eta)D_\eta,
\qquad
S_{I_\Phi}R(\eta)V=[S_{I_\Phi}W(\eta)]D_\eta.
\]
Because $S_{I_\Phi}W(\eta)\to S_{I_\Phi}W_0$ is nonsingular, $S_{I_\Phi}W(\eta)$ is uniformly well conditioned for sufficiently small $\eta$. The selected base feature-moment matrix therefore has singular orders $\eta^d$ with multiplicity $m_d^\Phi$. The same star/hypergraph construction as in Lemma~\ref{lem:feature-moment} extends these $K$ rows to exactly $\Ephi$ strictly interior product policies, which converge to the common uniform product baseline as $\eta\to 0$.

\subsection{Tensor decomposition}
\label{app:strength-tensor}

Write $B=\prod_{j=1}^J B_j$ and
\[
\Pi_j=B_j^{-1}\ones\ones',
\qquad
H_j=I_{B_j}-\Pi_j.
\]
Then
\[
R_j(\eta)=\Pi_j+\eta H_j,
\qquad
R(\eta)=\bigotimes_{j=1}^J R_j(\eta).
\]
This is the same family as \eqref{eq:Rj}. For each subset $S\subset\{1,\ldots,J\}$, let $\mathcal H_S$ be the tensor subspace that applies $H_j$ for $j\in S$ and $\Pi_j$ for $j\notin S$, and set
\[
\mathcal H_d
=
\bigoplus_{|S|=d}\mathcal H_S.
\]
Each factor contributes the eigenvalue $1$ on $\mathrm{Im}\,\Pi_j$ and the eigenvalue $\eta$ on $\mathrm{Im}\,H_j$. Hence
\[
R(\eta)v=\eta^d v
\qquad\text{for every }v\in\mathcal H_d.
\]
The dimension of $\mathcal H_d$ is the number of ways to choose $d$ non-constant factors and then a mean-zero contrast in each chosen factor,
\[
\dim\mathcal H_d
=
c_d
=
[z^d]
\prod_{j=1}^J
\bigl(1+(B_j-1)z\bigr).
\]
The spaces $\mathcal H_0,\ldots,\mathcal H_J$ exhaust $\R^B$ and are mutually orthogonal.

\subsection{Current-payoff-difference block}
\label{app:strength-current}

Reparameterize into current differences in feature coordinates $\delta_a(x)=\theta(x,a)-\theta(x,a_0)$ and the reference array $r(x)=\theta(x,a_0)$. After the regime subtraction of Appendix~\ref{app:strength-R}, continuation differences vanish in the first block, so the measurements at state $x$ are the feature-moment averages $M_x\delta_a(x)$. Write $A$ for the fixed invertible change from those feature coordinates into the $V$-basis of $F$, so $\Phi=VA$. For each state the filtration-preserving construction gives
\[
M_x
=
S_x S_{I_\Phi}R(\eta)V A.
\]
Every star selector $S_x$ contains $I_K$ as a block, so the relevant singular values of $S_x$ are bounded above and below by $\eta$-independent constants. Combined with uniform well-conditioning of $S_{I_\Phi}W(\eta)$ from Appendix~\ref{app:strength-sel}, the current-payoff-difference block contributes $M(A-1)m_d^\Phi$ directions of order $\eta^d$. When $F=\R^B$, $m_d^\Phi=c_d$.

\subsection{Reference-payoff block}
\label{app:strength-ref}

Regime $1$ exposes the reference array through the cycle increment $\alpha\beta \Deltacycle r_e$. After the already identified current difference is subtracted, factor $D_\eta$ out of each state's feature coordinate by $R(\eta)V=W(\eta)D_\eta$. The remaining map on the $V$-coordinates is $\alpha$ times a combination of the star selectors $S_x$, the limiting block $B_0:=S_{I_\Phi}W_0$, and the cycle increment $\Deltacycle$. The matrix $B_0$ is nonsingular by construction of $I_\Phi$.

Suppose the limiting reference operator annihilates the state-feature coordinates $z_x$. Then $\Deltacycle r_e=0$ for every environment $e$. Hence each $r_e$ is constant across states, so there exists $c\in\R^E$ such that
\[
S_x B_0 z_x=c\qquad\forall x.
\]
Therefore
\[
c\in\bigcap_x\operatorname{col}(S_x)=\operatorname{span}\{\ones_E\}.
\]
Write $c=\lambda\ones_E$. Because $S_x\ones_K=\ones_E$ and every star selector $S_x$ contains $I_K$, one has $S_x B_0 z_x=S_x(\lambda\ones_K)$, hence $B_0 z_x=\lambda\ones_K$. Nonsingularity of $B_0$ then yields
\[
z_x=\lambda B_0^{-1}\ones_K\qquad\forall x.
\]
The limiting kernel is therefore one-dimensional. This direction is the global payoff-location direction. After quotienting location, the limiting operator is injective. Finite dimensionality implies that its smallest singular value is strictly positive, and independent of $(\alpha,\eta)$ for small $\eta$.

Right multiplication by $I_M\otimes D_\eta$ then produces $M-1$ reference directions of order $\alpha$, corresponding to the state-constant opponent component $m_0^\Phi=1$, and $M m_d^\Phi$ directions of order $\alpha\eta^d$ for $d\ge 1$. When $F=\R^B$ these multiplicities are $M c_d$.

\subsection{Smallest singular value}
\label{app:strength-smin}

The weakest identified block is the reference component of degree $\dPhi$, of order $\alpha\eta^{\dPhi}$ and multiplicity $M m_{\dPhi}^\Phi$. Hence, under this construction,
\[
\sigma_{\min,+}
\asymp
\alpha\eta^{\dPhi}.
\]
On the isotropic path $\alpha=\eta=\varepsilon$,
\[
\sigma_{\min,+}\asymp\varepsilon^{\dPhi+1}.
\]
Saturation has $\dPhi=J$ and recovers $\varepsilon^{J+1}$. This is Theorem~\ref{thm:strength} and Corollary~\ref{cor:spectrum}.

\subsection{Construction-specific spectrum}
\label{app:strength-remark}

A different pair of kernels can expose a high-order opponent contrast through action-dependent continuation variation that is already present at first order in the environmental difference. In that case the same $J$-fold contrast need not be multiplied by $\alpha$. Cardinality-minimality and local strength are therefore distinct design criteria.

In the unrestricted joint-policy comparison of Remark~\ref{rem:joint}, every joint-action contrast can be moved at first order. Under the same canonical environmental perturbation the full dynamic rate is then $\Theta(\alpha\eta)$, and isotropically $\Theta(\varepsilon^{2})$. The rates are those of this construction.

\subsection{Oracle sampling-noise amplification}
\label{app:noise}

\begin{proof}[Proof of Proposition~\ref{prop:noise}]
Oracle GLS is linear in $\widehat z_{n,\varepsilon}$ and uses the true covariance $\Omega_{n,\varepsilon}$, so
\[
\operatorname{Var}(\widehat\theta)
=
\frac1n
\bigl(C_\varepsilon'\Omega_{n,\varepsilon}^{-1}C_\varepsilon\bigr)^{-1}.
\]
The identity uses linearity; it does not use a normal approximation. The uniform bounds $cI\preceq\Omega_{n,\varepsilon}\preceq CI$ sandwich $C'\Omega^{-1}C$ between constant multiples of $C'C$. The smallest eigenvalue of $C'C$ is $\sigma_{\min,+}^2$, so $\lambda_{\max}\operatorname{Var}(\widehat\theta)\asymp 1/(n\sigma_{\min,+}^2)$. Theorem~\ref{thm:strength} supplies $\sigma_{\min,+}\asymp\alpha\eta^{\dPhi}$. On the isotropic triangular sequence $\alpha=\eta=\varepsilon_n\to 0$, $\sigma_{\min,+}\asymp\varepsilon_n^{\dPhi+1}$.
\end{proof}

\begin{proof}[Proof of Proposition~\ref{prop:twopoint}]
Work after the location quotient, so $C$ has full column rank on the coordinates in which $\sigma_{\min,+}(C)$ is computed. Let $v$ be a unit right singular vector of $C$ for $\sigma_{\min,+}$, and set $\theta_1=\theta_0+t v$ with $t=\gamma/(\sqrt n\,\sigma_{\min,+})$ and $\gamma>0$ to be chosen. The two Gaussian laws $N(C\theta_0,\Omega/n)$ and $N(C\theta_1,\Omega/n)$ have Kullback--Leibler divergence
\[
\mathrm{KL}
=
\frac n2\,(\theta_1-\theta_0)'C'\Omega^{-1}C(\theta_1-\theta_0)
\le
\frac n{2c}\,\|C(\theta_1-\theta_0)\|_2^2
=
\frac{\gamma^2}{2c}.
\]
Take $\gamma=\sqrt{c}$, so that $\mathrm{KL}\le 1/2$. Pinsker's inequality then gives $\mathrm{TV}\le\sqrt{\mathrm{KL}/2}\le 1/2$. Le Cam's two-point method yields
\[
\inf_{\psi}\bigl(P_{\theta_0}(\psi\neq 0)+P_{\theta_1}(\psi\neq 1)\bigr)
\ \ge\
1-\mathrm{TV}
\ \ge\
\tfrac12.
\]
Any estimator $\widehat\theta$ induces a test by nearest neighbor of $\{\theta_0,\theta_1\}$. On the event that the test is wrong, $\|\widehat\theta-\theta\|^2\ge t^2/4$, so
\[
\inf_{\widehat\theta}
\sup_{\theta\in\{\theta_0,\theta_1\}}
E_\theta\|\widehat\theta-\theta\|^2
\ \ge\
\frac{t^2}{16}
=
\frac{c}{16n\sigma_{\min,+}^2}.
\]
Combined with Theorem~\ref{thm:attenuation}, $n\eta^{2\dPhi}\to\infty$ is necessary for this risk to vanish along any $\eta$-clustered sequence of designs.
\end{proof}

\section{Identification-robust inference}
\label{app:inference}

\subsection{Coverage}
\label{app:ar-coverage}

\begin{proof}[Proof of Proposition~\ref{prop:ar}]
Write $\psi=(p_i,q,P)$ for the first-stage parameter and $h(\psi,\theta)=a(\psi)-C_\Phi(\psi)\theta$ for the composite map, so that $m_n(\theta)=h(\widehat\psi,\theta)$ and $h(\psi_0,\theta_0)=0$ at the truth. Both $a$ and $C_\Phi$ are compositions of the Hotz--Miller inverse, the resolvent $(I-\beta P_0)^{-1}$, and finitely many products of the estimated primitives. At an interior $\psi_0$, where choice probabilities are bounded away from $0$ and $1$ and $\beta\in(0,1)$ keeps the resolvent bounded, each of those operations is continuously differentiable, and so is $h$ in $\psi$ for fixed $\theta$. Let $D(\theta_0)=\partial_\psi h(\psi_0,\theta_0)$.

A first-order expansion gives
\[
\sqrt n\,m_n(\theta_0)
=
D(\theta_0)\sqrt n(\widehat\psi-\psi_0)+o_p(1)
\Rightarrow
N\bigl(0,\ D(\theta_0)VD(\theta_0)'\bigr)
=
N\bigl(0,\Omega(\theta_0)\bigr).
\]
Let $L=\operatorname{rank}\Omega(\theta_0)$. By covariance consistency, rank stability $\Pr\{\operatorname{rank}\widehat\Omega_n(\theta_0)=L\}\to 1$, and the spectral gap on the nonzero eigenvalues, the Moore--Penrose inverse is continuous along the estimator sequence, so
\[
\widehat\Omega_n(\theta_0)^{+}\to_p\Omega(\theta_0)^{+}.
\]
Continuous mapping then gives
\[
\mathrm{AR}_n(\theta_0)
=
\bigl(\sqrt n\,m_n(\theta_0)\bigr)'\widehat\Omega_n^{+}\bigl(\sqrt n\,m_n(\theta_0)\bigr)
\Rightarrow
Z'\Omega^{+}Z,
\qquad
Z\sim N(0,\Omega),
\]
and $Z'\Omega^{+}Z\sim\chi^2_L$ because $\Omega^{+/2}Z$ is standard normal on the $L$-dimensional range of $\Omega$. Hence $\Pr\{\mathrm{AR}_n(\theta_0)\le c_{L,1-\alpha}\}\to 1-\alpha$, which is the coverage statement since $\theta_0\in\mathcal C_n(1-\alpha)$ if and only if $\mathrm{AR}_n(\theta_0)\le c_{L,1-\alpha}$.

Uniformity is the Lindeberg--Feller statement for the triangular array $\{D(\theta_0)(\widehat\psi-\psi_0)\}$ together with uniform negligibility of the expansion remainder, both of which hold by assumption on the stated class; the eigenvalue bounds make the convergence $\widehat\Omega_n^{+}\to_p\Omega^{+}$ uniform as well.

Nothing in the argument refers to $C_\Phi$ except through $D(\theta_0)$ and $\Omega(\theta_0)$. In particular $\sigma_{\min,+}(C_\Phi)$ may be arbitrarily small, or $C_\Phi$ may be rank deficient, without affecting the display. This is the point of inverting a test on the moment rather than centering an approximation at an estimator: the statistic is evaluated at the null value $\theta_0$, where the model holds exactly, rather than at $\widehat\theta$, whose distribution does depend on how well the design determines $\theta$.
\end{proof}

\begin{remark}[Why the usual interval fails here]
The Wald interval replaces $\mathcal C_n$ by an ellipsoid centered at $\widehat\theta$ with the variance of $\widehat a$ alone. Two things go wrong at small $\sigma_{\min,+}$. The centering is far from $\theta_0$ relative to its own reported scale, because $\widehat\theta-\theta_0$ is amplified by $\widehat W^{+}$; and the reported scale omits the contribution of $\widehat C_\Phi\theta$ to $\operatorname{Var}(m_n)$. The omitted design-matrix variance is quadratic in $\theta$, while its covariance with $\widehat a$ is linear in $\theta$; both are largest exactly along the directions in which $\theta$ must be large to move the measurement. The two errors compound rather than offset, which is what Section~\ref{sec:mc} measures.
\end{remark}

\begin{lemma}[Finite-bootstrap cutoff]
\label{lem:hotelling}
Let $\Bboot>L$ be fixed, and define $\widehat\Omega_n(\theta)$ by \eqref{eq:boot-omega}. Suppose $\sqrt n\,m_n(\theta_0)\Rightarrow N(0,\Omega(\theta_0))$ as in Proposition~\ref{prop:ar}, and that, jointly in the $\Bboot$ draws,
\[
\Bigl\{
\sqrt n\bigl[m_n^{*(b)}(\theta_0)-m_n(\theta_0)\bigr]
\Bigr\}_{b=1}^{\Bboot}
\ \Big|\ \mathcal D_n
\Rightarrow
\mathrm{iid}\ N\bigl(0,\Omega(\theta_0)\bigr)
\]
as $n\to\infty$, with this limiting vector independent of the original-sample limit $Z$. Then
\[
\frac{\Bboot-L}{L(\Bboot-1)}\,\mathrm{AR}_n(\theta_0)
\Rightarrow
F_{L,\Bboot-L}.
\]
The cutoff $c_{L,1-\alpha}=L(\Bboot-1)(\Bboot-L)^{-1}F_{L,\Bboot-L,1-\alpha}$ therefore has asymptotic coverage $1-\alpha$. As $\Bboot\to\infty$ it converges to the $\chi^2_L$ quantile of Proposition~\ref{prop:ar}.
\end{lemma}

\begin{proof}[Proof of Lemma~\ref{lem:hotelling}]
Restrict attention to the $L$-dimensional range of $\Omega(\theta_0)$, on which the limiting covariance is positive definite. Write $Z_n=\sqrt n\,m_n(\theta_0)$ and $\xi_n^{*(b)}=\sqrt n\bigl(m_n^{*(b)}(\theta_0)-m_n(\theta_0)\bigr)$. The original-sample CLT together with the conditional Gaussian limit for the scores yield that $(Z_n,\xi_n^{*(1)},\ldots,\xi_n^{*(\Bboot)})$ converges unconditionally to independent $N(0,\Omega)$ vectors. The estimator \eqref{eq:boot-omega} is the unbiased sample covariance of those $\Bboot$ scores, so its limit $S$ satisfies $(\Bboot-1)S\sim W_L(\Bboot-1,\Omega)$ and is independent of $Z$. Hotelling's theorem \citep{hotelling1931,anderson2003} then gives
\[
\frac{\Bboot-L}{L(\Bboot-1)}\,Z'S^{-1}Z
\sim
F_{L,\Bboot-L}.
\]
Continuous mapping yields the same law for $\mathrm{AR}_n(\theta_0)=Z_n'\widehat\Omega_n(\theta_0)^{+}Z_n$. Finite fourth moments alone do not produce a Wishart limit at fixed $\Bboot$; the Gaussian approximation for the scores is the input that does. As $\Bboot\to\infty$, $S\to_p\Omega$ and $Z'S^{-1}Z\Rightarrow\chi^2_L$.
\end{proof}

The condition $\Bboot>L$ is primitive: on the sixty-state grid $L=360$ and $\Bboot=400$, so the second degrees of freedom are $40$. That grid is a robustness check, not the headline design. On the headline grid $L=72$ and $\Bboot=400$, and repeating the inversion at $\Bboot=1000$ leaves the reported conclusions unchanged (Appendix~\ref{sec:emp-ar}). The lemma is an $n\to\infty$ statement at fixed $\Bboot$. It does not claim that finite-sample cluster-bootstrap draws of choice probabilities are exactly Gaussian.

\subsection{Projection width}
\label{app:ar-width}

\begin{proof}[Proof of Corollary~\ref{cor:width}]
Write $\widehat\theta$ for a minimizer of $\mathrm{AR}_n(\theta)$ and $\widehat\Omega_n(\theta)$ for the bootstrap covariance at the candidate. The map $\theta\mapsto\widehat\Omega_n(\theta)$ is a quadratic polynomial in $\theta$, hence continuous. The hypothesis $\sqrt n\,\sigma_{\min,+}(\widehat C_\Phi)\to_p\infty$ makes a neighborhood of $\widehat\theta$ of radius $O_p(n^{-1/2}\sigma_{\min,+}^{-1})$ shrinking. On that neighborhood $\widehat\Omega_n(\theta)=\widehat\Omega_n(\widehat\theta)+o_p(1)$ uniformly, so the candidate-dependent set and the ellipsoid computed from $\widehat\Omega_n(\widehat\theta)$ coincide up to $o_p(n^{-1/2})$ in the $W(\widehat\theta)$ metric. Without the rate, the neighborhood need not shrink and the local comparison is not available.

On that ellipsoid, for $\widehat\lambda\in\operatorname{range}\widehat W$ the half-width is
\[
\tfrac12\bigl|\mathrm{CI}_{\widehat\lambda}^{\mathrm{ellip}}\bigr|
=
\sqrt{\tfrac1n\bigl(c_{L,1-\alpha}-\mathrm{AR}_n(\widehat\theta)\bigr)\widehat\lambda'\widehat W^{+}\widehat\lambda}.
\]
Take $\widehat\lambda$ to be the unit right singular vector of $\widehat C_\Phi$ associated with $\sigma_{\min,+}(\widehat C_\Phi)$. By Proposition~\ref{prop:noise} the eigenvalues of $\Omega_n$ on its range lie in $[c_\Omega,C]$, so $\widehat W\preceq c_\Omega^{-1}\widehat C_\Phi'\widehat C_\Phi$ and therefore $\widehat\lambda'\widehat W^{+}\widehat\lambda\ge c_\Omega\,\sigma_{\min,+}(\widehat C_\Phi)^{-2}$ up to the $o_p(1)$ first-stage error in $\widehat C_\Phi$. Restricting to the event $\{\mathrm{AR}_n(\widehat\theta)\le c_{L,1-\alpha}-\delta\}$ replaces the slack by a number at least $\delta$, and
\[
\bigl|\mathrm{CI}_{\widehat\lambda}^{\mathrm{ellip}}\bigr|
\ \ge\
\frac{\sqrt{\delta}\,c_\Omega^{1/2}}{\sqrt n\,\sigma_{\min,+}(\widehat C_\Phi)}.
\]
The $o_p(1)$ comparison with the candidate-dependent set yields the first display of the corollary. On the rank-stable identified subspace,
\[
\sigma_{\min,+}(\widehat C_\Phi)
=
\sigma_{\min,+}(C_\Phi)\,(1+o_p(1)).
\]
Theorem~\ref{thm:attenuation} then supplies $\sigma_{\min,+}(C_\Phi)\le\kappa\eta^{\dPhi}$ along rank-identifying designs, and substituting gives the second inequality. Singular-vector convergence is not required.

When the design is just identified, $\mathrm{AR}_n(\widehat\theta)=0$ identically on the identified subspace, so the slack event is the whole space and the $\Omega_p$ lower bound is unconditional. When it is overidentified by $L-r$ degrees, $\mathrm{AR}_n(\widehat\theta)\Rightarrow\chi^2_{L-r}$ under the model, and $P(\chi^2_{L-r}\le c_{L,1-\alpha}-\delta)>0$ for every $\delta\in(0,c_{L,1-\alpha})$. The $\Omega_p$ claim is then a statement on that event. It is not a deterministic lower bound on every realization: if the minimized statistic approaches the cutoff, the projection width approaches zero, which is why the slack restriction cannot be dropped.
\end{proof}

\begin{proof}[Proof of Proposition~\ref{prop:local}]
Theorem~\ref{thm:attenuation} gives $\sigma_n\le\kappa\eta_n^{\dPhi}$ on every rank-identifying $\eta_n$-clustered design, which is the displayed limsup. Coverage of $\mathcal C_n$ at $\theta_0$ is Proposition~\ref{prop:ar}; that argument does not use $\sigma_n$, and its uniformity statement is inherited.

Now take $\sqrt n\,\sigma_n\to\tau\in[0,\infty)$ with $(\lambda_n,u_n)\to(\lambda,u)$. Write $h(\psi,\theta)=a(\psi)-C_\Phi(\psi)\theta$ as in the proof of Proposition~\ref{prop:ar}, so $m_n(\theta)=h(\widehat\psi,\theta)$ and $h(\psi_0,\theta_0)=0$. Continuous differentiability in $\psi$ gives, at $\theta_n(a)=\theta_0+a\lambda_n$,
\[
\sqrt n\,m_n\bigl(\theta_n(a)\bigr)
=
\sqrt n\,h\bigl(\psi_0,\theta_n(a)\bigr)
+
D\bigl(\theta_n(a)\bigr)\sqrt n(\widehat\psi-\psi_0)
+o_p(1).
\]
The population remainder is $h(\psi_0,\theta_n(a))=-a C_{\Phi,n}\lambda_n$, so $\sqrt n\,h(\psi_0,\theta_n(a))=-a\sqrt n\,\sigma_n\,u_n\to -a\tau u$. The first-stage term converges jointly with $\sqrt n(\widehat\psi-\psi_0)\Rightarrow N(0,V)$ to $N(0,\Omega(\theta_0+a\lambda))$, because $\theta_n(a)\to\theta_0+a\lambda$ and $D$ is continuous at interior $\psi_0$. Hence
\[
\sqrt n\,m_n\bigl(\theta_n(a)\bigr)
\ \Rightarrow\
N\bigl(-a\tau u,\ \Omega(\theta_0+a\lambda)\bigr).
\]
The map $\theta\mapsto\widehat\Omega_n(\theta)$ is a quadratic polynomial. Uniform eigenvalue bounds on a ball about $\theta_0$ give $\widehat\Omega_n(\theta_n(a))^{+}\to_p\Omega(\theta_0+a\lambda)^{+}$. Continuous mapping yields the noncentral $\chi^2_L$ limit with noncentrality $a^2\tau^2 u'\Omega(\theta_0+a\lambda)^{+}u$. If $\tau=0$ the noncentrality vanishes for every fixed $a$.

For (i), the noncentral limit at each fixed $a$ is central. For every $M<\infty$,
\[
\Pr\bigl\{\mathrm{AR}_n(\theta_n(\pm M))\le c_{L,1-\alpha}\bigr\}\to 1-\alpha,
\]
so $\lambda_n'\theta_0\pm M$ both lie in $\mathrm{CI}_{\lambda_n}$ with probability tending to at least $1-2\alpha$, and $\bigl|\mathrm{CI}_{\lambda_n}\bigr|\to_p\infty$.

For (ii), $\tau>0$ and $u$ lies in the range of $C_{\Phi}$. The map $a\mapsto a\tau u$ is a nondegenerate line in moment space, so the limiting sublevel set of the ray statistic is almost surely a bounded interval of positive length. Ray inversion is therefore $O_p(1)$ and $\Omega_p(1)$. The profiled interval contains the ray inversion, which gives the $\Omega_p(1)$ claim for \eqref{eq:profile-ci}.

Part (iii) is Corollary~\ref{cor:width}.
\end{proof}

\section{Unobserved heterogeneity}
\label{app:het}

\begin{proof}[Proof of Proposition~\ref{prop:het}]
When types are resolved, measurement $e$ of type $\omega$ restricts $u_\omega$ and no other type's payoff, because both the choice probabilities and the kernel entering \eqref{eq:z} carry the index $\omega$. Ordering the stacked rows by type therefore gives $C^\Phi_{1:T}=\operatorname{diag}(C^\Phi_1,\ldots,C^\Phi_T)$ after the same ordering of columns. For a block diagonal matrix the set of singular values is the union of the blocks' singular values, so
\[
\sigma_{\min,+}\bigl(C^\Phi_{1:T}\bigr)=\min_\omega\sigma_{\min,+}\bigl(C^\Phi_\omega\bigr),
\]
which is the equality. Theorem~\ref{thm:attenuation} applies within each identified type-specific block with its own baseline-adapted degree $\dPhiomega$ and constant $\kappa_\omega=\kappa(M,A,\beta,\Phi,\qbar_\omega)$. Taking the minimum across blocks gives the displayed bound. The effective sample size for type $\omega$ is $n\pi_\omega$, yielding the type-specific rate condition $n\pi_\omega\eta_\omega^{2\dPhiomega}\to\infty$.
\end{proof}

\section{Unknown patience}
\label{app:unknown}
\label{app:cross}
\label{app:affine}
\label{sec:patience}

\begin{theorem}[Unknown patience at the sharp cardinalities]
\label{thm:patience}
For every $M,A,B\ge 2$, there exist exactly $R=2$ strictly positive primitive transition regimes and $E=\Estar(M,B)$ strictly interior conditionally independent product-policy profiles such that the following holds. For every fixed normalized structural payoff $u$ and true discount factor $\beta\in(0,1)$, draw one centered continuation-offset vector in each regime-policy environment from an absolutely continuous distribution supported in a sufficiently small feasible open neighborhood. Then
\[
\Pr_\omega
\bigl\{
(u,\beta)
\text{ is globally identified up to payoff location}
\bigr\}
=1.
\]
The randomization support may be arbitrarily small.
\end{theorem}

The construction uses the same product-policy family as Theorem~\ref{thm:sharp}. Centered continuation offsets are an additional continuous experimental primitive. Own-action payoff differences generate beta-free directions. False discount candidates can only survive through the continuation block. Centered offsets move that block in $2E(M-1)$ independent directions. A quotient argument leaves at least three payoff directions outside that offset span, so the false-discount incidence set has codimension at least one.

The kernels below use exactly $E=\Estar(M,B)$ profiles. Let $U=M^{-1}\ones\ones'$ and $\Deltacycle=\Ccycle-I$, where $\Ccycle$ is the $M$-cycle. Fix $0<\kappa<1$ and a strictly positive vector $(d_b)_{b=1}^B$. Choose $\alpha>0$ with
\[
\alpha\max_b d_b
<
\min\{\kappa,1/M\}.
\]
The $P$ regime is
\[
P_{0b}=\kappa I+(1-\kappa)U,
\qquad
P_{ab}=\kappa \Ccycle+(1-\kappa)U
\quad(a\neq 0).
\]
The $Q$ regime is
\[
Q_{0b}=U,
\qquad
Q_{ab}=U+\alpha d_b \Deltacycle
\quad(a\neq 0).
\]
Every block is strictly positive and row-stochastic. Write
\[
\gamma(\beta)
=
\frac{\beta\kappa}{1-\beta\kappa}.
\]
Define
\[
r(x,b)=u(x,0,b),
\qquad
\delta_a(x,b)=u(x,a,b)-r(x,b),
\]
and, for each policy environment $e$,
\[
r_e(x)
=
q^e(\cdot\mid x)'r(x,\cdot),
\qquad
s_e(x)
=
q^e(\cdot\mid x)'d.
\]
Let $Q_e\delta_a$ denote the state vector whose coordinate at $x$ is $q^e(\cdot\mid x)'\delta_a(x,\cdot)$, and write $D(s_e)=\mathrm{diag}(s_e(x))_{x}$. Under regime $P$, every nonreference own action shares the continuation operator $\gamma(\beta)\Deltacycle$, so
\[
y_{P,e,a}
=
Q_e\delta_a
+
\gamma(\beta)\Deltacycle r_e.
\]
Under regime $Q$, the resolvent of $U$ and $\Deltacycle U=0$ give the continuation operator $\alpha\beta D(s_e)\Deltacycle$, so
\[
y_{Q,e,a}
=
Q_e\delta_a
+
\alpha\beta D(s_e)\Deltacycle r_e.
\]

\subsection{Known-discount rank}
\label{app:unknown-known}

Suppose all measurements vanish. Subtracting the two regime observations at a fixed $(e,a)$ yields
\[
\bigl[
\alpha\beta D(s_e)
-
\gamma(\beta)I
\bigr]Lr_e
=
0.
\]
The coefficient matrix is diagonal in state coordinates. Because $d$ is strictly positive and $q^e$ is a probability,
\[
s_e(x)
=
\sum_b q^e(b\mid x)\,d_b
\le
\max_b d_b.
\]
The construction takes $\alpha\max_b d_b<\kappa$. At the same time
\[
\frac{\gamma(\beta)}{\beta}
=
\frac{\kappa}{1-\beta\kappa}
>
\kappa,
\]
so $\alpha s_e(x)<\gamma(\beta)/\beta$ and therefore
\[
\alpha\beta s_e(x)-\gamma(\beta)<0
\]
for every state $x$ and every policy $e$. The diagonal entries never vanish. Hence
\[
Lr_e=0
\]
for every $e$. Substituting back into the $P$-equation gives $Q_e\delta_a=0$ for every nonreference $a$. Stacking environments at a fixed state produces $Q_x\delta_a(x,\cdot)=0$. Lemma~\ref{lem:product} supplies $\mathrm{rank}\,Q_x=B$, so
\[
\delta_a(x,\cdot)=0
\qquad\text{for every }a\neq 0\text{ and every }x.
\]
The identities $Lr_e=0$ mean that each policy average $r_e$ is constant across states. Thus there is a vector $c\in\R^E$ with $Q_x r(x,\cdot)=c$ for every $x$, i.e.\
\[
c
\in
\bigcap_x\mathrm{col}(Q_x)
=
\mathrm{span}\{\ones_E\}
\]
by \eqref{eq:cap}. Full column rank of each $Q_x$ together with $Q_x\ones_B=\ones_E$ forces $r(x,b)=c$ for a single global constant. Therefore
\[
\ker C(\beta)=\mathrm{span}\{\ones\},
\qquad
\mathrm{rank}\,C(\beta)=MAB-1
\]
for every $\beta\in(0,1)$.

\subsection{Offset response}
\label{app:unknown-offset}

Now fix a false candidate $\tilde\beta\neq\beta$. The $P$-regime centered-offset response block is
\[
\bigl[\gamma(\tilde\beta)-\gamma(\beta)\bigr]\Deltacycle.
\]
The map $\gamma$ is strictly increasing on $(0,1)$, so $\gamma(\tilde\beta)-\gamma(\beta)\neq 0$. The $Q$-regime response block is
\[
\alpha(\tilde\beta-\beta)D(s_e)\Deltacycle.
\]
Because $d$ is strictly positive and policies are strictly interior, $s_e(x)>0$ at every state, so $D(s_e)$ is invertible. Each block therefore has rank $M-1$. The $2E$ regime-policy environments enter block-diagonally, and
\[
\mathrm{rank}\,D_{\beta,\tilde\beta}
=
2E(M-1).
\]
At the sharp cardinality $E=\Estar$ this is $2\Estar(M-1)$.

\subsection{Quotient by continuation images}
\label{app:unknown-lambda}

Take $A=2$ and define linear functionals on the stacked measurements by
\[
\lambda_{P,e}(y)
=
\ones'y_{P,e},
\qquad
\lambda_{Q,e}(y)
=
\ones'D(s_e)^{-1}y_{Q,e}.
\]
Let $\Lambda$ collect these $2E$ functionals. On a $P$-block, $\ones'\Deltacycle=0$, so $\lambda_{P,e}$ annihilates the offset response. On a $Q$-block,
\[
\ones'D(s_e)^{-1}\bigl(\alpha(\tilde\beta-\beta)D(s_e)\Deltacycle\omega\bigr)
=
\alpha(\tilde\beta-\beta)\ones'\Deltacycle\omega
=
0.
\]
Hence $\Lambda D_{\beta,\tilde\beta}=0$, i.e.\ $\mathrm{Im}\,D_{\beta,\tilde\beta}\subseteq\ker\Lambda$. There are $2E$ independent functionals and the measurement space has dimension $2EM$, so $\mathrm{rank}\,\Lambda=2E$ and
\[
\dim\ker\Lambda=2E(M-1).
\]
The two spaces have the same dimension, and therefore
\[
\operatorname{Im}D_{\beta,\tilde\beta}=\ker\Lambda.
\]
The composition $\Lambda C(\tilde\beta)$ depends only on current payoff differences. It is independent of the candidate discount. On the product-policy family,
\[
\mathrm{rank}\,\Lambda C(\tilde\beta)\ge 3.
\]
If $B\ge 3$, the first $B$ rows of $R(\eta)$ are linearly independent constant profiles. Those profiles already produce at least three independent current-difference functionals after $\Lambda$ is applied.

If $B=2$, only two constant profiles are available. Write $r_1$ and $r_2$ for the first two state-constant rows of $R(\eta)$, and let $q^3$ be the third Markov profile. The corresponding $P$-block quotient functionals on current-difference arrays $\delta$ are
\[
f_1(\delta)
=
\sum_x r_1'\delta(x,\cdot),
\qquad
f_2(\delta)
=
\sum_x r_2'\delta(x,\cdot),
\qquad
f_3(\delta)
=
\sum_x q_x^{3\prime}\delta(x,\cdot).
\]
Suppose $f_3=af_1+bf_2$. Because the arrays $\delta(x,\cdot)$ may be chosen independently across states, equality of the functionals forces
\[
q_x^3
=
ar_1+br_2
\qquad\text{for every }x.
\]
The right-hand side does not depend on $x$, so $q^3$ would be state-constant. Lemma~\ref{lem:product} constructs the extra profile, in the $B=2$, $\Estar=3$ case, so that it equals $r_1$ at state $1$ and equals $r_2\neq r_1$ at at least one other state. That profile is not state-constant. Hence $f_3\notin\mathrm{span}\{f_1,f_2\}$, and $\mathrm{rank}\,\Lambda C(\tilde\beta)\ge 3$.

In either case the rank is at least three. Combined with the image identity,
\[
\mathrm{rank}\bigl[C(\tilde\beta),D_{\beta,\tilde\beta}\bigr]
=
\mathrm{rank}\,D_{\beta,\tilde\beta}
+
\mathrm{rank}\,\Lambda C(\tilde\beta)
\ge
2\Estar(M-1)+3
\ge
2MB+1.
\]
The last inequality uses $\Estar(M-1)\ge MB-1$.

\subsection{Extension to \texorpdfstring{$A>2$}{A>2}}
\label{app:unknown-A}

For each additional own action $a=2,\ldots,A-1$, form the row difference $y_a-y_1$. Every nonreference own action shares the same transition law in each regime, so both the continuation term and the centered-offset term cancel. The difference reduces to
\[
y_a-y_1
=
Q_e(\delta_a-\delta_1).
\]
Reparameterize the payoff columns by the invertible map
\[
(\delta_1,\delta_2,\ldots,\delta_{A-1})
\mapsto
(\delta_1,\theta_2,\ldots,\theta_{A-1}),
\qquad
\theta_a=\delta_a-\delta_1.
\]
The $A=2$ core uses only the reference array $r$ and the first difference $\delta_1$. After the change of coordinates, the transformed row block for each action $a\ge 2$ depends only on $\theta_a$. Distinct $\theta_a$ occupy disjoint column blocks. Each such block has rank $MB$, because $\mathrm{rank}\,Q_x=B$ at every state. The transformed coefficient matrix is therefore block triangular, and the additive rank increment is
\[
MB(A-2).
\]
Hence
\[
\mathrm{rank}\bigl[C(\tilde\beta),D_{\beta,\tilde\beta}\bigr]
\ge
2MB+1+MB(A-2)
=
MAB+1.
\]

\subsection{Incidence set}
\label{app:unknown-incidence}

Fix a normalized structural payoff $u$ and true $\beta$. For a false candidate $\tilde\beta$, a centered-offset design $\omega$ is bad if there exists $\tilde u$ with $C(\tilde\beta)\tilde u+D_{\beta,\tilde\beta}\omega=C(\beta)u$. The unknown $(\tilde u,\omega)$ after location normalization has dimension $MAB-1+2E(M-1)$. The coefficient matrix has rank at least $MAB+1$, so each false-discount fiber has dimension at most $2E(M-1)-2$. The incidence set over the scalar false discount is semialgebraic on the maintained discount interval and has dimension at most $2E(M-1)-1$. Projection cannot increase dimension. The ambient offset space has dimension $2E(M-1)$, so the bad set has codimension at least one, hence Lebesgue measure zero and empty interior. Any absolutely continuous randomization of $\omega$ supported in a feasible open neighborhood, arbitrarily small, therefore identifies $(u,\beta)$ with probability one. This is Theorem~\ref{thm:patience}. The statement is pointwise in $u$.

Exact modular ranks on the listed small models
$(M,A,B)=(2,2,2)$, $(2,2,3)$, $(2,2,4)$, $(3,2,3)$, $(3,2,4)$, $(2,3,3)$, and $(3,3,3)$
confirm $\mathrm{rank}\,C=MAB-1$, $\mathrm{rank}\,D=2E(M-1)$, and $\mathrm{rank}[C,D]\ge MAB+1$.
Those witnesses are recorded in Appendix~\ref{app:repro}.

\section{Local transfer implementation}
\label{app:transfer}
\label{sec:implementation}
\label{sec:transfer}

\begin{proposition}[Local transfer implementability]
\label{prop:transfer}
Index opponents by $j=1,\ldots,J$ with action counts $B_j$. At a regular interior equilibrium $F(c,s,r)=0$, with $c$ the normalized CCP profile of dimension
\[
M\Bigl[(A-1)+\sum_{j=1}^{J}(B_j-1)\Bigr],
\]
$s$ the vector of normalized nonreference transfers of matching dimension, and $r=\tau_{i0}$ the identified player's reference transfer, assume $D_c F$ is nonsingular. Normalized nonreference transfers enter payoff-difference equations additively, so $D_s F=-I$ up to sign. Hence $D_s c=-(D_c F)^{-1}D_s F$ is nonsingular, and $(s,r)\mapsto(c,r)$ is a local diffeomorphism. The observable target
\[
\Psi(c,r)=\bigl(c_{-i},\,\Pi[g_i(c_i)+r]\bigr)
\]
has dimension $M\sum_{j=1}^{J}(B_j-1)+(M-1)$ and satisfies
\[
\mathrm{rank}\,D\Psi
=
M\sum_{j=1}^{J}(B_j-1)+(M-1).
\]
When $J=1$ and $B_1=2$ the dimension is $2M-1$. Conditional on the relevant target measurement equilibria lying on regular interior branches, small product-policy and centered-offset perturbations are locally implementable by transfers.
\end{proposition}

\begin{proof}[Proof of Proposition~\ref{prop:transfer}]
Index opponents by $j=1,\ldots,J$ with action counts $B_j$. Opponent marginal CCP coordinates have dimension $M\sum_{j=1}^J(B_j-1)$. The identified player's CCP coordinates have dimension $M(A-1)$. The full equilibrium CCP coordinate therefore has dimension
\[
M\Bigl[(A-1)+\sum_{j=1}^J(B_j-1)\Bigr].
\]
Let $c$ be that normalized CCP profile, let $s$ be the vector of normalized nonreference transfers of matching dimension, and let $r=\tau_{i0}\in\R^M$ be the identified player's reference transfer. At a regular interior equilibrium, $F(c,s,r)=0$ with $D_c F$ nonsingular. Normalized nonreference transfers enter the corresponding choice-difference equations additively, so $D_s F=-I$ up to sign. Therefore
\[
D_s c=-(D_c F)^{-1}D_s F
\]
is nonsingular, and $(s,r)\mapsto(c,r)$ is a local diffeomorphism.

The target
\[
\Psi(c,r)
=
\bigl(c_{-i},\,\Pi[g_i(c_i)+r]\bigr)
\]
has dimension
\[
M\sum_{j=1}^J(B_j-1)+(M-1).
\]
In $(c,r)$ coordinates, the derivative of the first component with respect to opponent marginal CCP coordinates contains the identity of that dimension, while the derivative of the centered continuation offset with respect to centered $r$ is the identity on $\ones^\perp$. Hence
\[
\mathrm{rank}\,D\Psi
=
M\sum_{j=1}^J(B_j-1)+(M-1).
\]
When $J=1$ and $B_1=2$ this dimension is $2M-1$. The joint opponent policy remains the product of the implemented marginal CCPs. Composition with the local diffeomorphism yields a submersion from transfers to the observable target. On a regular interior branch, the constant-rank/submersion theorem implies that an absolutely continuous transfer randomization on a sufficiently small open neighborhood pushes forward to an absolutely continuous law of the target offsets. No knowledge of $u$ or $\beta$ is required to define the target. The result is local and does not assert global reachability from an arbitrary distant equilibrium branch.
\end{proof}

\section{Restricted transitions}
\label{app:lag}
\label{sec:restricted}

Write $x=(z,\ell)$ and suppose every primitive regime satisfies
\begin{equation}
\label{eq:lag-P}
P^r(z',\ell'\mid z,\ell,a,b)
=
T^r_{ab}(z'\mid z)\,
\mathbf 1\{\ell'=(a,b)\}.
\end{equation}

\begin{proposition}[Lagged-action transition obstruction]
\label{prop:lag}
Write $x=(z,\ell)$ with $z\in\mathcal Z$ and $\ell$ the previous action profile. Suppose every primitive regime satisfies \eqref{eq:lag-P}. Then for any two such regimes $P,Q$ and any $\beta\in(0,1)$,
\[
\dim\bigl(\mathrm{Im}\,G_{P,\beta}\cap\mathrm{Im}\,G_{Q,\beta}\bigr)
\ge
|\mathcal Z|.
\]
Consequently $\mathrm{rank}\,C(\beta)\le MAB-|\mathcal Z|$ after stacking any finite collection of opponent policies under $P$ and $Q$.
\end{proposition}

\begin{proof}[Proof of Proposition~\ref{prop:lag}]
Write \(x=(z,\ell)\) and assume \eqref{eq:lag-P} for kernels \(P\) and \(Q\). For \(h\in\R^M\),
\[
\bigl[G_{P,\beta}h\bigr](z,\ell,a,b)
=
h(z,\ell)
-
\beta\sum_{z'}T^P_{ab}(z'\mid z)\,h\bigl(z',(a,b)\bigr).
\]
The continuation term does not depend on the incoming lag \(\ell\). Therefore
\[
\bigl[G_{P,\beta}h\bigr](z,\ell,a,b)
-
\bigl[G_{P,\beta}h\bigr](z,\ell',a,b)
=
h(z,\ell)-h(z,\ell'),
\]
and likewise for \(G_{Q,\beta}k\). If \(G_{P,\beta}h=G_{Q,\beta}k\), then \(h-k\) is independent of \(\ell\): there is \(c\in\R^{|\mathcal Z|}\) with
\[
k(z,\ell)=h(z,\ell)-c(z)
\quad\text{for all }\ell.
\]
Substituting and writing \(h_\alpha(z):=h(z,\alpha)\) for each action profile \(\alpha=(a,b)\) yields the linear system
\begin{equation}
\label{eq:lag-sys}
\bigl(T^P_\alpha-T^Q_\alpha\bigr)h_\alpha
+
T^Q_\alpha c
=
\beta^{-1}c
\qquad\text{for every }\alpha.
\end{equation}
The unknowns are the maps \(\{h_\alpha\}_\alpha\) and \(c\), a space of dimension \((|\mathcal L|+1)|\mathcal Z|\) when \(\ell\) ranges over the action profiles \(\mathcal L\). The constraints \eqref{eq:lag-sys} comprise at most \(|\mathcal L|\,|\mathcal Z|\) independent scalar equations. The solution space therefore has dimension at least \(|\mathcal Z|\).

The linear map sending a solution \((h,c)\) to \(G_{P,\beta}h\) is injective on that space: \(G_{P,\beta}h=0\) implies \(h=0\) by Proposition~\ref{prop:env}, hence \(c=0\) and \(k=0\). Each solution therefore produces a distinct vector in \(\mathrm{Im}\,G_{P,\beta}\cap\mathrm{Im}\,G_{Q,\beta}\). Stacking opponent policies cannot leave the intersection, so
\[
\mathrm{rank}\,C(\beta)\le MAB-|\mathcal Z|.
\]
The constant-flow line is included. The bound is strictly larger than the universal intersection lower bound of Lemma~\ref{lem:E1} whenever \(|\mathcal Z|\ge 2$.
\end{proof}
\section{Empirical appendix}
\label{app:emp-data}
\label{app:emp-T}
\label{app:emp-commonu}
\label{app:emp-null}
\label{app:emp-rank}
\label{app:emp-robust}
\label{app:emp-boot}

\subsection{Source, sample, and coding}
\label{sec:emp-data}

The T-100 Domestic Segment (U.S.\ Carriers) extract is the official TranStats table FIM, years 2014--2019, class F \citep{bts_t100}. Annual zip archives, SHA-256 checksums, and scripted transforms are recorded in the project provenance file. Lookups distributed as HTML homepages are not used as dictionaries; field definitions come from \texttt{Documentation.csv} and \texttt{Term.csv} inside the 2014 zip.

The window is 2016--2019, after the US Airways certificate (20355) exits the file as a separate reporter. American (19805) and Southwest (19393) remain distinct certificates. Virgin America is not merged into Alaska, and regionals are not collapsed into majors. Among major certificates present throughout the window, this pair maximizes a support score multiplying route coverage by own and opponent activity and overlap. Three-carrier overlap is about $8$ percent, which is why the airline exercise has one rival and $\dPhi\le 1$.

An action is at least twelve performed departures in the undirected city-pair-quarter; thresholds of one and of twenty-four departures were also constructed, the first being nearly always on among observed carrier-route cells. City ranks use 2014 passenger volumes. Demand cutpoints are quantiles of 2014 route traffic on the routes this pair plays, so the demand state is predetermined and the grid is balanced in sample. Absence of recorded service is coded as zero activity on the complete grid. Choice probabilities are Laplace-smoothed frequencies with addend $1/2$. Requiring a lag and a lead for the transition leaves $2{,}204$ route-quarter observations on the $241$ routes where both carriers appear.

\subsection{Stratification and the first stage}

The six policy environments are the cells of a calendar split, 2016--17 against 2018--19, crossed with route-distance terciles. Within each cell the rival choice probability $q$, the own choice probability $p_i$, and the two marginal transitions $P_s$ and $P_n$ are estimated by smoothed counts on that cell alone. The demand and rival-count transitions are estimated as row-normalized counts rather than by the ridge multinomial specification of the earlier draft, so that every object entering $(\widehat a,\widehat C_\Phi)$ is a function of cell counts and the cluster bootstrap can resample the entire first stage without refitting a penalized likelihood inside each draw. The counted specification is used to keep that bootstrap fully nonparametric in cell counts; T1 is reported only as a predictive robustness diagnostic. Table~\ref{tab:airline-T} records the comparison: the evidence that current actions move the kernel is real but small and not stable across scoring rules, so the headline design does not condition the kernel on current play.

\begin{table}[t]
\centering
\caption{Transition-model fit on shared route-cluster folds. Joint Brier is the fifteen-class quadratic score. $T1^{-ab}$ and $T2^{-ab}$ omit current actions.}
\label{tab:airline-T}
\small
\begin{tabular}{lrrc}
\toprule
model & OOS joint log likelihood & joint Brier & current $(a,b)$ \\
\midrule
T0 & $-0.7875$ & $0.3426$ & no \\
$T1^{-ab}$ & $-0.7826$ & $0.3395$ & no \\
T1 & $-0.7706$ & $0.3364$ & yes \\
$T2^{-ab}$ & $-0.7552$ & $0.3350$ & no \\
T2 & $-0.7534$ & $0.3372$ & yes \\
\bottomrule
\end{tabular}
\end{table}

T0 estimates $P_s(s'\mid s)$ and $P_n(n'\mid n)$ by row-normalized transition counts on consecutive route-quarters. T1 estimates two ridge multinomial logits for $s'$ and $n'$ given $(s,n,a,b)$ and low-order interactions, with an unpenalized intercept; $T1^{-ab}$ is the same specification with current actions omitted, on the same ridge grid and the same four route-cluster folds. T2 and $T2^{-ab}$ are the analogous fifteen-class joint logits. Selected penalties are $0.2$ throughout. T1 improves both scores against its matched action-off comparator; T2 improves the log likelihood and worsens the Brier score. Of $57$ occupied $(s,n,a,b)$ cells, $16$ have fewer than twenty observations and those hold $1.0$ percent of transition observations. No transition model was selected to raise $\operatorname{rank}C$.

\subsection{The bootstrap and the statistic}

The covariance $\widehat\Omega_n(\theta)$ is \eqref{eq:boot-omega}, the sample covariance of the $\sqrt n$-scaled and recentered cluster-bootstrap scores, over $\Bboot=400$ route-market cluster draws, evaluated at the candidate $\theta$ being tested. Each draw resamples the $241$ routes with replacement and rebuilds every first-stage object from the resampled counts. A single set of draws serves every $\theta$ because $m_n$ is affine. Headline intervals invert the profiled statistic \eqref{eq:profile-ci}. The sixty-state widths invert only along the identified ray, which is contained in the profiled interval. The ellipsoid computed from the same draws at a preliminary least-squares $\theta$ is recorded as a diagnostic and is not the reported set. The cutoff is the Hotelling value of Lemma~\ref{lem:hotelling}; on the headline grid this raises it from $92.8$ to $106.0$. The ellipsoid computed from a preliminary $\theta$ rejects the rival-independent restriction against that cutoff; the reported candidate-dependent statistic does not.

\begin{table}[h]
\centering
\caption{Headline grid, saturated $K=2$ unless noted, $\beta=0.95$. Profiled intervals at two bootstrap sizes. Cutoffs are the Hotelling values of Lemma~\ref{lem:hotelling}.}
\label{tab:boot-B}
\small
\begin{tabular}{lrrrrr}
\toprule
$\Bboot$ & cutoff & $\mathrm{AR}_{K=1}$ & $\min\mathrm{AR}_{K=2}$ & $|\mathrm{CI}|_{d=0}$ & $|\mathrm{CI}|_{d=1}$ \\
\midrule
$400$ & $106.0$ & $87.5$ & $11.1$ & $9.0$ & $52.7$ \\
$1000$ & $93.0$ & $86.2$ & $9.4$ & $9.3$ & $54.9$ \\
\bottomrule
\end{tabular}
\end{table}

The confidence set is evaluated through the singular value decomposition of $\widehat\Omega_n^{-1/2}\widehat C_\Phi$ rather than through the eigendecomposition of $\widehat W=\widehat C_\Phi'\widehat\Omega_n^{+}\widehat C_\Phi$. The two are algebraically identical. They are not numerically identical: forming $\widehat W$ squares the spectrum, and at the conditioning of Table~\ref{tab:airline-ar}, where $\sigma_{\min,+}/\sigma_{\max}$ is about $3\times 10^{-5}$, squaring places the weakest identified directions below double-precision resolution and merges them with the null space. An unbounded interval for a merely weak direction, or a finite interval for a direction the design does not restrict, are both artifacts of that step. The second converts an unrestricted payoff contrast into a reported finite interval.

Table~\ref{tab:tol} counts unbounded directions on the headline grid at relative tolerances from $10^{-6}$ to $10^{-14}$. Read from the whitened design matrix, the count is $4$ under the rival-independent restriction and $6$ under saturation at every tolerance in that range, so the identified--unidentified cut is a property of the design and not of the cutoff. Read from the Gram matrix, the same count collapses to zero below $10^{-10}$: the genuine null space has been squared into the noise floor, and every direction, including those the design annihilates exactly, is reported as measured. We use $10^{-11}$ on the whitened singular values, which lies inside the stable range by five decades on either side.

\begin{table}[h]
\centering
\caption{Unbounded payoff directions at the headline grid as a function of the relative tolerance, from the whitened design matrix $\widehat\Omega_n^{-1/2}\widehat C_\Phi$ and from its Gram matrix $\widehat W$. The two are algebraically identical. The Gram column loses the null space entirely below $10^{-10}$.}
\label{tab:tol}
\small
\begin{tabular}{lrrrrrrr}
\toprule
& $10^{-6}$ & $10^{-8}$ & $10^{-10}$ & $10^{-11}$ & $10^{-12}$ & $10^{-13}$ & $10^{-14}$ \\
\midrule
$K=1$, whitened SVD & 4 & 4 & 4 & 4 & 4 & 4 & 4 \\
$K=1$, Gram matrix & 4 & 4 & 0 & 0 & 0 & 0 & 0 \\
$K=2$, whitened SVD & 6 & 6 & 6 & 6 & 6 & 6 & 6 \\
$K=2$, Gram matrix & 6 & 5 & 0 & 0 & 0 & 0 & 0 \\
\bottomrule
\end{tabular}
\end{table}

Interaction degree is assigned by the filtration of Section~\ref{sec:filtration}: a payoff direction $v$ is decomposed at the pooled interior baseline $\qbar$, and $v$ is labelled degree $d$ when at least $60$ percent of its energy lies in layer $d$. Reported widths are for the most favorably conditioned singular direction of each degree class, meaning the one with the largest singular value among those passing the purity threshold. This selection is conservative for the local conditioning comparison. We do not claim that the selected direction globally minimizes the candidate-dependent profiled projection width. The degree-one widths in Tables~\ref{tab:airline-ar} and~\ref{tab:airline-grid} are therefore lower bounds on how badly the design measures rival-dependent payoffs.

\subsection{Sensitivity across grids}
\label{sec:emp-ar}
\label{sec:emp-grid}

\begin{table}[t]
\centering
\caption{Sensitivity across state grids and stratifications, saturated specification $K=2$, $\beta=0.95$, nominal $95$ percent. Headline through distance-only widths invert the profiled statistic \eqref{eq:profile-ci}. Sixty-state widths invert along the identified ray and are therefore a lower bound on that projection. Observed behavior spans between $7.4$ and $9.2$ logit units. Designs marked exact have $L=\dim\theta$, so their minimized statistic is zero and no specification test exists. The last column reports the minimized statistic of the rival-independent restriction against its cutoff.}
\label{tab:airline-grid}
\small
\begin{tabular}{lrrrrrr}
\toprule
grid & $L$ & rank & $|\mathrm{CI}|_{d=0}$ & $|\mathrm{CI}|_{d=1}$ & $\mathrm{AR}_{K=1}$ & cutoff \\
\midrule
headline, $M=12$ & 72 & 42 & $9.0$ & $52.7$ & $87.5$ & $106.0$ \\
four strata, $M=12$ (exact) & 48 & 42 & $3.1\times10^{3}$ & $3.4\times10^{5}$ & $31.4$ & $74.3$ \\
coarse demand, $M=8$ (exact) & 32 & 27 & $7.3$ & $4440$ & $31.0$ & $51.2$ \\
distance only, $M=12$ & 36 & 36 & $3.5$ & $96.0$ & $9.8$ & $55.7$ \\
sixty-state grid, $M=60$ & 360 & 192 & $0.77$ & $25.5$ & $534$ & $730$ \\
\bottomrule
\end{tabular}
\end{table}

On every grid the saturated specification is not rejected and the most favorably conditioned rival-dependent interval exceeds the range of observed behavior. The rival-independent restriction is rejected on none of the five grids.

\subsection{Robustness to the discount factor}

The discount factor is a maintained diagnostic value. Repeating the headline design at $\beta=0.90$ and $\beta=0.99$ leaves every reported conclusion in place. The rival-independent restriction is rejected at none of the three values, with minimized statistics $88.0$, $87.5$, and $88.7$ against a common cutoff of $106.0$. The saturated specification is rejected at none. The most favorably conditioned degree-zero width is $9.9$, $9.0$, and $11.7$, and the most favorably conditioned degree-one width is $60.8$, $52.7$, and $49.6$, so the width ratio moves between $4$ and $6$ across the range. The smallest identified singular value moves between $8.1\times10^{-5}$ and $9.4\times10^{-5}$. The discount $\beta$ enters $C_\Phi$ through the resolvent, which is well conditioned throughout $(0,1)$ on a strictly positive kernel, whereas the attenuation is driven by $\eta$, which does not depend on $\beta$ at all.

\subsection{Common payoffs across strata}

Assumption~\ref{ass:common} is the binding maintained hypothesis, and the strata do not satisfy it exactly. On the consecutive American--Southwest panel, comparing 2016--17 against 2018--19 route-quarter means: passengers $192{,}311$ against $197{,}284$; distance $1{,}080$ against $1{,}090$ miles; other-major count $0.870$ against $0.974$; American activity $0.279$ against $0.287$; Southwest activity $0.497$ against $0.524$. Annual averages of monthly U.S.\ Gulf Coast kerosene-type jet fuel spot prices \citep{eia_jetfuel} are $\$1.25$, $\$1.56$, $\$2.02$, and $\$1.88$ for 2016 through 2019, or $\$1.43$ against $\$1.95$ across the two windows, so the late-window fuel price is $36$ percent higher.

A payoff shifter that moves with the calendar violates Assumption~\ref{ass:common} across the calendar split, and route composition differs across the distance split by construction. Both are the composition problem of Section~\ref{sec:composition}, and neither is repaired by the design. The placebo test of Table~\ref{tab:eta-placebo} is therefore a necessary condition: passing it establishes that the measured cross-stratum dispersion exceeds sampling noise, not that the excess is rival-policy variation holding the payoff fixed. The reported $\eta$ overstates the usable dispersion, so the intervals in Table~\ref{tab:airline-ar} understate how weakly this panel measures rival-dependent payoffs. An ellipsoid computed from a preliminary $\theta$ does reject the rival-independent restriction on the headline grid; the reported set does not, which is one more reason the ellipsoid is not the object Proposition~\ref{prop:ar} covers.

\subsection{What earlier drafts of this exercise reported}
\label{app:emp-archive}

Earlier versions of this section reported a different object on a sixty-state grid, and the contrast is instructive enough to record. There the identifying policy environments were not measured. They were constructed as small interior logit perturbations of the pooled rival policy, clipped to $[10^{-3},1-10^{-3}]$, and the reported result was the rank of the stacked operator: $60$ at one environment, then $120$ and $180$ as designed policies were added, and $225$ once a second designed kernel inside the lagged-action class was introduced, against the ceiling $MAB-|\mathcal Z|=225$ of Proposition~\ref{prop:lag}, leaving the predicted residual gauge $d_G=15$. Conditioning was reported alongside, at $\sigma_{180}/\sigma_1\approx 5.4\times10^{-5}$ and $\sigma_{225}/\sigma_1\approx 5.5\times10^{-6}$, and the integer sequence was stable at diagnostic discounts $0.90$ and $0.99$ and across relative rank tolerances from $10^{-6}$ to $10^{-10}$. A route-market cluster bootstrap of the design found the integers stable in $99$ of $100$ draws while the conditioning ratios moved over more than an order of magnitude.

That construction is not evidence of identification in the observed panel. The perturbations were invented, so the ranks describe an experiment nobody ran; Theorem~\ref{thm:attenuation} says any positive perturbation attains full rank, so attaining the ceiling was guaranteed; and the conditioning ratios, which do carry information, cannot be converted into a statement about any payoff parameter, which is the gap Section~\ref{sec:inference} closes. In the bootstrap the integers were stable while the conditioning ratios moved by an order of magnitude: the stable quantity was the rank, and the quantity that varied was the one that governs precision. The rank sequence is retained here as a check on Proposition~\ref{prop:lag}, for which it is an equality case, and for nothing else. Superseded tables and the archived designs remain in the replication package.

\subsection{Comparison with a three-player panel}
\label{sec:emp-aux}

\citet{dearingblevins2025} study a three-player dynamic retail-entry game. The public analysis file has $19{,}320$ observations on $1{,}610$ markets over 2010--2021 and a forty-state coding $x=(s,a_{1,-1},a_{2,-1},a_{3,-1})$ with $s\in\{1,\ldots,5\}$, so the saturated payoff has $MAB=320$ coordinates. The public file contains no geography, calendar kernels are nearly identical, and the sample cannot isolate a common-payoff environmental split, so this is a benchmark rather than a second application.

Pooled, the stacked operator has rank $40$. Adding a second calendar kernel moves it only to $60$, because those kernels are nearly the same continuation technology; adding a second calendar policy moves it to $74$; crossed two-by-two measurements remain at $74$. Saturation failure is therefore not an artifact of the airline coding. The occupied-state feature-moment diagnostic of Table~\ref{tab:db-K} uses the same file to exhibit the cardinality prediction of Theorem~\ref{thm:structure} at $B=4$, where the interaction filtration has three layers rather than two. It is a policy-geometry computation on a calibration and carries no strength claim: the degree-two layer is attenuated by $\eta^2$, and the only place that exponent is measured against a known truth is the Monte Carlo of Section~\ref{sec:mc}.

\begin{table}[t]
\centering
\caption{Dearing--Blevins occupied-state feature-moment diagnostic, $B=4$. Switches assigned by ascending occupied-state index. A data-calibrated policy-geometry diagnostic, not one-regime identification of the restricted payoff.}
\label{tab:db-K}
\small
\begin{tabular}{lcccc}
\toprule
family & $K$ & $E$ & $\Ephi$ & occupied intersection dim \\
\midrule
active-rival count & 2 & 3 & 3 & 1 \\
additive rivals & 3 & 4 & 4 & 1 \\
pairwise interactions & 4 & 5 & 5 & 1 \\
saturated indicators & 4 & 5 & 5 & 1 \\
\bottomrule
\end{tabular}
\end{table}

\section{Reproducibility}
\label{app:repro}

The theoretical claims used in the paper were frozen in a self-contained regression suite. The command \texttt{python src/run\_regression.py} reruns nine modules and writes \texttt{REGRESSION\_STATUS.json}. Required statuses, each equal to \texttt{PASS}, are
\begin{itemize}
\itemsep0pt
\item \texttt{strength\_hierarchy}
\item \texttt{strength\_construction\_scaling}
\item \texttt{unknown\_beta\_general\_AB}
\item \texttt{bts\_frozen\_geometry}
\item \texttt{restricted\_payoff}
\item \texttt{restricted\_strength}
\item \texttt{design\_free\_attenuation}
\item \texttt{one\_regime\_exact\_count}
\item \texttt{ar\_coverage}
\end{itemize}
Module 1 checks the isotropic multiplicity law of the saturated construction, including the high-precision $(2,2,2)$ order-$4$ witness. Module 2 checks canonical product-policy singular-value slopes and the unrestricted-joint benchmark slopes; it does not assert a universal product-policy upper bound. Module 3 checks exact modular ranks of the general-$(A,B)$ unknown-patience construction, including $\mathrm{rank}\,C=MAB-1$, $\mathrm{rank}\,D=2E(M-1)$, $\operatorname{Im}D=\ker\Lambda$ when $A=2$, $\mathrm{rank}\,\Lambda C\ge 3$, and $\mathrm{rank}[C,D]\ge MAB+1$. Module 4 rereads the frozen BTS confirmation and does not re-estimate airline primitives. Module 9 rereads the frozen Monte Carlo of Section~\ref{sec:mc} by default; rebuild it with \texttt{AR\_RECOMPUTE=1}. Module 5 checks the restricted cardinality theorem, proved one-regime bounds, the explicit $K+\Ephi$ construction, and the exact $K=1$ tradeoff. Module 6 checks filtration integers $m_d^\Phi$ and numerical log-log slopes under the canonical restricted design; slope agreement is a regression diagnostic, not a proof.

Module 7 checks Theorem~\ref{thm:attenuation} as a design-free statement rather than as a property of one construction. It draws random transition geometries, random feature matrices, and random clustered policy families, verifies each step of the proof separately, and confirms that the ratio $\sigma_{\min,+}/\eta^{\dPhi}$ stays below the constant \eqref{eq:kappa} across designs that differ in $R$, $E$, and kernel separation. Module 8 certifies Proposition~\ref{prop:e1exact} by exact rank computation over a prime field. Module 9 is the Monte Carlo of Section~\ref{sec:mc}: coverage of \eqref{eq:ar} that does not deteriorate as identification weakens, failure of the interval that ignores design-matrix error, the width ordering across interaction degrees, and a market-size sweep that separates finite-sample shortfall from Wald undercoverage.

Exact modular witnesses for unknown patience include $(M,A,B)=(2,2,2)$, $(2,2,3)$, $(2,2,4)$, $(3,2,3)$, $(3,2,4)$, $(2,3,3)$, and $(3,3,3)$. High-precision restacks near $\tilde\beta=\beta$ are recorded in the same suite.

The empirical exercise of Section~\ref{sec:empirical} is not part of the frozen theory suite, because it depends on the external BTS extract. It is reproduced by \texttt{python Replication\_package/code/empirical/airline\_ar.py}, which rebuilds the panel from the processed T-100 files, runs the five grid configurations and the two discount-factor robustness configurations, and writes \texttt{airline\_ar\_results.json}. Every random draw is seeded from the configuration itself, so results do not depend on the order in which configurations are run. The tolerance comparison of Table~\ref{tab:tol} is produced by \texttt{diag\_tol.py} in the same directory.

\bibliographystyle{plainnat}
\bibliography{refs}

\end{document}